\documentclass[runningheads,orivec]{llncs}
\usepackage{xspace}
\usepackage{hyperref}
\usepackage{mwe}
\usepackage{wrapfig}
\usepackage{paralist}

\usepackage{todonotes}

\newcommand{\nPN}{$\nu$PN\xspace}
\newcommand{\nPNs}{$\nu$PNs\xspace}

\usepackage{listings}
\usepackage{comment}

\usepackage{amsfonts}
\usepackage{amssymb}
\usepackage{caption}
\usepackage{subcaption}
\usepackage{multirow}
\usepackage{thm-restate}
\usepackage{empheq}
\usepackage[inline, shortlabels]{enumitem}
\usepackage{cleveref}

\newcommand{\multiset}[1]{\boldsymbol{#1}}

\newcommand{\os}{\ensuremath{\mathfrak{E}}}
\newcommand{\tup}[1]{\langle#1\rangle}

\newcommand{\supp}{\mathit{supp}}

\newcommand{\sqleq}{\sqsubseteq}

\newcommand{\lruns}[3][\ell]{\mathit{Runs}_{#1}(#2,#3)}

\newcommand{\support}[1]{\texttt{Supp}(#1)}

\newcommand{\prefun}{{\ensuremath{\tt pre}}}
\newcommand{\postfun}{{\ensuremath{\tt post}}}

\usepackage{tikz}
\usetikzlibrary{patterns,shapes,arrows,calc,fit,arrows.meta,decorations.pathmorphing,arrows.meta,decorations.pathreplacing}
\usetikzlibrary{arrows,petri,automata,positioning}
\tikzstyle{place}=[circle,minimum height=6mm,draw]
\tikzstyle{transvert}=[rectangle,minimum width=2mm, minimum height=8mm,draw]
\tikzstyle{transhor}=[rectangle,minimum width=8mm, minimum height=2mm,draw]
\tikzset{>={Stealth[scale=1.2]}}
\tikzstyle{empty}=[circle,minimum height=6mm]

\tikzset{Rightarrow/.style={double equal sign distance,>={Implies},->},
myDouble/.style={Rightarrow,double},
dashDouble/.style={Rightarrow,double,dashed},
triple/.style={-,preaction={draw,Rightarrow}},
quadruple/.style={preaction={draw,Rightarrow,shorten >=0pt},shorten >=1pt,-,double,double
distance=0.2pt}}

\newcommand{\D}{\mathcal{D}}
\newcommand{\E}{\mathfrak{E}}
\newcommand{\F}{\mathcal{F}}

\newcommand{\K}{\mathcal{K}}
\newcommand{\M}{\mathcal{M}}
\newcommand{\N}{\mathcal{N}}

\newcommand{\T}{\mathcal{T}}

\newcommand{\X}{\mathcal{X}}

\newcommand{\fmset}[1]{\{\!\{#1\}\!\}}
\newcommand{\nestTok}{\T}

\renewcommand{\blacksquare}{\blacktriangle}
\newcommand{\nuPN}{$\nu$PN\xspace}

\newcommand*{\defeq}{\stackrel{\text{def}}{=}}
\newcommand{\var}{\text{Var}}
\newcommand{\out}{out_{\Upsilon}(t)}
\newcommand{\abs}[1]{\lvert#1\rvert}
\newcommand{\GnPN}{c-\nPN}
\newcommand{\GnPNs}{c-\nPNs}
\newcommand{\rnPN}{r-\nPN}

\tikzstyle{triangle}=[draw, regular polygon, regular polygon sides=3]
\tikzstyle{pentagon}=[draw, regular polygon, regular polygon sides=5,minimum height=6mm,]
\tikzset{>={Stealth[scale=1.2]}}

\makeatletter
\def\dasharrowfill@#1#2#3#4{%
        $\m@th
        \thickmuskip0mu
        \medmuskip\thickmuskip
        \thinmuskip\thickmuskip
        \relax
        #4#1\mkern2mu
        \xleaders\hbox{$#4\mkern2mu#2\mkern2mu$}\hfill
        \mkern2mu
        #3$%
}

\def\dashleftarrowfill@{\dasharrowfill@\leftarrow\relbar\relbar}
\def\dashrightarrowfill@{\dasharrowfill@\relbar\relbar\rightarrow}
\def\dashleftrightarrowfill@{\dasharrowfill@\leftarrow\relbar\rightarrow}
\def\dashLeftarrowfill@{\dasharrowfill@\Leftarrow\Relbar\Relbar}
\def\dashRightarrowfill@{\dasharrowfill@\Relbar\Relbar\Rightarrow}
\def\dashLeftrightarrowfill@{\dasharrowfill@\Leftarrow\Relbar\Rightarrow}

\providecommand*\xdashleftarrow[2][]{%
  \ext@arrow 0055{\dashleftarrowfill@}{#1}{#2}}
\providecommand*\xdashrightarrow[2][]{%
  \ext@arrow 0055{\dashrightarrowfill@}{#1}{#2}}
\providecommand*\xdashleftrightarrow[2][]{%
  \ext@arrow 0055{\dashleftrightarrowfill@}{#1}{#2}}
\providecommand*\xdashLeftarrow[2][]{%
  \ext@arrow 0055{\dashLeftarrowfill@}{#1}{#2}}
\providecommand*\xdashRightarrow[2][]{%
  \ext@arrow 0055{\dashRightarrowfill@}{#1}{#2}}
\providecommand*\xdashLeftrightarrow[2][]{%
  \ext@arrow 0055{\dashLeftrightarrowfill@}{#1}{#2}}
\makeatother

\newcommand{\Id}{\mathrm{Id}}

\newcommand{\dest}{\mathit{destroy}}
\newcommand{\cons}{\mathit{cons}}
\usepackage[T1]{fontenc}
\usepackage{graphicx}
\begin{document}
\title{Nets-within-Nets through the Lens of Data Nets}
%
%
\author{Francesco {Di Cosmo}\inst{1}\orcidID{0000-0002-5692-5681} \and
Soumodev {Mal}\inst{2}\orcidID{0000-0001-5054-5664} \and
Tephilla {Prince}\inst{3}\orcidID{0000-0002-1045-3033}}
\authorrunning{F. {Di Cosmo}, S. Mal and T. Prince}
%
\institute{Free University of Bozen-Bolzano, Italy\\ \email{frdicosmo@unibz.it} \and
Chennai Mathematical Institute, India\\
\email{soumodevmal@cmi.ac.in} \and
IIT Dharwad, India\\
\email{tephilla.prince.18@iitdh.ac.in}}
\maketitle              
\begin{abstract}
Elementary Object Systems (EOSs) are a model in the nets-within-nets (NWNs) paradigm, where tokens in turn can host standard Petri nets. We study the complexity of the reachability problem of EOSs when subjected to non-deterministic token losses.
It is known that this problem is equivalent to the coverability problem with no lossiness of conservative EOSs (cEOSs).
We precisely characterize cEOS coverability into the framework of data nets, whose tokens carry data from an infinite domain. Specifically, we show that cEOS coverability is equivalent to the coverability of an interesting fragment of data nets that extends beyond \nPN (featuring globally fresh name creation), yet remains less expressive than Unordered Data Nets (featuring lossy name creation as well as powerful forms of whole-place operations and broadcasts). 
This insight bridges two apparently orthogonal approaches to PN extensions, namely data nets and NWNs. At the same time, it enables us to analyze cEOS coverability taking advantage of known results on data nets.
As a byproduct, we immediately get that the complexity of cEOS coverability lies between \(\mathbf{F}_{\omega 2}\) and \(\mathbf{F}_{\omega^\omega}\), two classes beyond Primitive Recursive.

\end{abstract}

\section{Introduction}
Recent works have studied the Nets Within Nets (NWN) paradigm~\cite{DBLP:conf/apn/Kohler-Bussmeier23a,OurPNSE24,DBLP:conf/ac/Valk03}, i.e., Petri Nets (PNs) whose tokens can in turn host PNs, as a model for the robustness of multiagent systems against agent breakdowns and, more generally, agent imperfections, modeled as token losses. These works focus on the reachability/coverability problems of Elementary Object Systems (EOS), i.e., NWNs where there is only one level of nesting. Other forms of NWNs can be found in~\cite{DBLP:conf/ac/Valk03,kohler-busmeier_survey_2014,DBLP:journals/topnoc/Kohler-BussmeierR23,DBLP:journals/fuin/Lomazova00,DBLP:conf/ershov/LomazovaS99}. Out of the several combinations of problem type (reachability and coverability), lossiness degree (none, finite, unbounded number of token losses), and level (at the outer, nested, or both levels), only reachability/coverability of EOSs under an unbounded amount of lossiness at both levels is decidable~\cite{OurPNSE24}.

The picture is moderately more optimistic when the constraint of conservativity is applied. In EOSs, each place can host tokens of a fixed type. In a conservative EOS (cEOS), if a transition consumes an object of a given type, at least one object of the same type must be produced, i.e., the set of types available in the net is conserved. On cEOSs, any non-zero amount of lossiness at both levels returns a reachability/coverability problem equivalent to perfect cEOS coverability. Instead, perfect cEOS reachability is known to be undecidable. Hence, the decidability boundary of reachability/coverability of lossy/perfect EOS/cEOS is fully charted. While lossy EOS reachability and perfect cEOS coverability are inter-reducible (see App.\ref{Sec:equivalence}) and decidable, their precise complexity class is unknown, besides an $F_{\omega2}$ lower bound~\cite{OurArxiv}.

The complexity of verification problems for NWNs have not been well-studied. In contrast, many results about the complexity of coverability for several data extensions of PNs~\cite{LazicNORW08}, whose tokens carry data from an infinite domain, are available in the literature (see Fig.1 in~\cite{LazicS16}).
For \nPN \cite{DBLP:journals/tcs/Rosa-VelardoF11}
the coverability problem is double-Ackermannian, \(F_{\omega 2}\)-complete~\cite{LazicS16}. 
For unordered data nets (UDNs), the coverability is hyper-Ackermannian, \(F_{\omega^\omega}\)-complete~\cite{Rosa-Velardo17}. Ordered data nets and ordered data Petri nets are both \(F_{\omega^{\omega^\omega}}\)-complete~\cite{Haddad2012}.

In this paper, we pave the way to the study of cEOS coverability and its equivalent forms (such as lossy EOS reachability). Instead of directly attacking this problem, for example by exploiting techniques for complexity over Well Structured Transition Systems~\cite{finkel_well-structured_2001,Rosa-Velardo17,MultiplyRecursive,schmitz2012algorithmic}, we bridge the nesting paradigm with the apparently orthogonal extension of Petri nets with data.
We show that cEOS coverability can be characterized as the coverability of a type of net with data, sitting in between \nPN and UDNs.

Our technical contributions are as follows:
\begin{enumerate}
    \item We introduce channel-$\nu$PN (\GnPN), a clean minimal extension of \nPNs suitable to characterize the complexity power of cEOS coverability.
     \item We show that cEOS coverability and \GnPN coverability are inter-reducible. As a consequence, we obtain inter-reducibility also between \GnPN coverability and lossy EOS reachability under an unbounded number of losses as well as cEOS reachability with any non-zero number of token losses. 
    \item On top of the $\F_{\omega2}$ lower bound from~\cite{OurArxiv}, we obtain a novel $\F_{\omega^\omega}$ upper bound from the literature on Unordered Data Nets (UDNs).
\end{enumerate}
Overall, we provide a novel connection between the nested and data paradigm, useful to study the complexity of perfect/imperfect EOS/cEOS reachability/coverability problems under the lens of data nets.

In Sec.~\ref{sec:prelims} we provide preliminaries on PNs, \nPNs, and cEOSs. In Sec.~\ref{sec:extend} we introduce \GnPNs and the equivalent fragment of rename $\nu$PNs (\rnPN). In Sec.~\ref{Sec:lowerbound} and Sec.~\ref{sec:toEOS}, we provide a reduction from \rnPN coverability to cEOS coverability. In Sec.~\ref{sec:fromEOS}, we provide a reduction from cEOS coverability to \GnPN coverability. In Sec.~\ref{sec:conclusions} we discuss the results and conclude.

\section{Preliminaries}\label{sec:prelims}
\paragraph*{Binary Relations and Multisets}
Given \(n \in \mathbb{N}\), we denote by \([n]\) the set \(\{1,2,\cdots,n\}\).
A \emph{multiset} $\multiset{m}$ on a set $D$ is a mapping $\multiset{m}:D\rightarrow \mathbb{N}$. The \emph{support} of $\multiset{m}$ is the set $\support{m} = \{i \mid \multiset{m}(i) > 0\}$. The multiset $\multiset{m}$ is finite if $\support{\multiset{m}}$ is finite. The family of all multisets over $D$ is denoted by $D^\oplus$. We denote a finite multiset $\multiset{m}$ by enumerating the elements $d\in\support{\multiset{m}}$ exactly $\multiset{m}(d)$ times in between $\{\{$ and $\}\}$, where the ordering is irrelevant. For example, the finite multiset ${\multiset{m}:\{p,q\}\longrightarrow \mathbb{N}}$ such that $\multiset{m}(p)=1$ and $\multiset{m}(q)=2$ is denoted by $\fmset{p,q,q}$.
The empty multiset $\fmset{}$ (with empty support) is also denoted by $\emptyset$. On the empty domain $D=\emptyset$ the only multiset is $\emptyset$; to stress this out we denote 
the empty multiset over the empty domain by $\varepsilon$.
Given two multisets $\multiset{m_1}$ and $\multiset{m_2}$ on $D$, we define the multiset $\multiset{m_1} + \multiset{m_2}$ and $\multiset{m_1} - \multiset{m_2}$ on $D$ as follows:
$(\multiset{m_1} + \multiset{m_2})(d) = \multiset{m_1}(d)  + \multiset{m_2}(d)$ and
$(\multiset{m_1} - \multiset{m_2})(d) = max(\multiset{m_1}(d)  - \multiset{m_2}(d),0)$.
Similarly, for a finite set $I$ of indices, $\sum_{i\in I} \fmset{d_i}$ denotes the multiset $\multiset{m}$ over $\bigcup_{i\in I}\{d_i\}$ such that $\multiset{m}(d)=|\{i\in I \mid d_i=d\}|$ for each $d\in D$. 
We write $\multiset{m_1} \sqleq \multiset{m_2}$ if, for each $d\in D$, we have $\multiset{m_1}(d)  \leq \multiset{m_2}(d)$.

\paragraph*{Petri Nets}

We denote a PN~\cite{murata89} as a tuple $N=(P,T,F)$, where $P$ is a finite place set, $T$ is a finite transition set, and $F$ is a flow function. We equivalently interpret $F$ via the functions ${\prefun}_N : T\rightarrow ( P\rightarrow \mathbb{N})$ where $\prefun_N(t)(p)=F(p,t)$ and ${\postfun}_N : T\rightarrow ( P\rightarrow \mathbb{N})$ where $\postfun_N(t)(p)=F(t,p)$. 
A transition $t\in T$ is enabled on a marking $\mu$ (finite multiset of places) if, for each place $p\in P$, we have $\prefun_N(t)(p)\leq \mu(p)$. Its firing results in the marking $\mu'$ such that $\mu'(p)=\mu(p)-\prefun_N(t)(p)+\postfun_N(t)(p)$, for each $p\in P$. We denote markings according to multiset notation. For example, as in Fig.~\ref{fig:PNInit}, the marking is $\fmset{p_1,p_1,p_1,p_1,p_2,p_3,p_4,p_5}$. 
If $P$ is ordered, we can denote a marking $m$ also as a vector $\tup{\multiset{m}(p)}_{p\in P}$.
We also work with the special \emph{empty PN} $\blacksquare=(\emptyset,\emptyset,\emptyset)$, whose only marking is $\varepsilon$. 

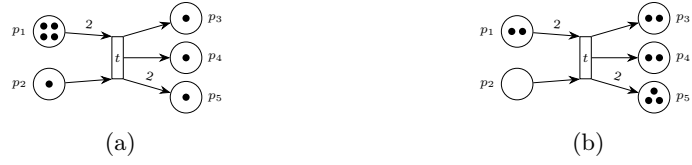
\begin{figure}[t]
    \centering
    \begin{subfigure}[b]{.49\textwidth}\centering
    \scalebox{0.7}{
    \begin{tikzpicture}

\node[place,label={[name=p2Lab]left:\scriptsize $p_2$},tokens=1](p2)at (.2,-0.5){};
\node[place,label={[name=p1Lab]left:\scriptsize $p_1$},tokens=4](p1)at (.2,.5){};

\node[transvert] (t)at (1.5,0){};
\node at (t){\scriptsize $t$};

\node[place,label={[name=p3Lab]right:\scriptsize $p_3$},tokens=1](p3)at (2.8,.75){};
\node[place,label={[name=p4Lab]right:\scriptsize $p_4$},tokens=1](p4)at (2.8,0){};
\node[place,label={[name=p5Lab] right:\scriptsize $p_5$},tokens=1](p5)at (2.8,-.75){};

\draw [->] (p1) to node [above, sloped]  (TextNode1) {\scriptsize $2$} (t.north west);

\draw [->] (p2) -- node[above,midway,sloped]{}(t.south west);
\draw [->] (t.north east) -- node[above,midway,sloped] {}(p3);

\draw [->] (t) -- node[above,midway,sloped,pos=.75] {}(p4);
\draw [->,bend right] (t.south east) --node[above,midway,sloped] {\scriptsize 2} (p5);

\end{tikzpicture}
    }
\caption{}
\label{fig:PNInit}
    \end{subfigure}
    \hfill
    \begin{subfigure}[b]{.49\textwidth}\centering
       \scalebox{0.7}{   
   \begin{tikzpicture}

\node[place,label={[name=p2Lab]left:\scriptsize $p_2$}](p2)at (.2,-0.5){};
\node[place,label={[name=p1Lab]left:\scriptsize $p_1$},tokens=2](p1)at (.2,.5){};

\node[transvert] (t)at (1.5,0){};
\node at (t){\scriptsize $t$};

\node[place,label={[name=p3Lab]right:\scriptsize $p_3$},tokens=2](p3)at (2.8,.75){};
\node[place,label={[name=p4Lab]right:\scriptsize $p_4$},tokens=2](p4)at (2.8,0){};
\node[place,label={[name=p5Lab] right:\scriptsize $p_5$},tokens=3](p5)at (2.8,-.75){};

\draw [->] (p1) to node [above, sloped]  (TextNode1) {\scriptsize $2$} (t.north west);

\draw [->] (p2) -- node[above,midway,sloped]{}(t.south west);
\draw [->] (t.north east) -- node[above,midway,sloped] {}(p3);

\draw [->] (t) -- node[above,midway,sloped,pos=.75] {}(p4);
\draw [->,bend right] (t.south east) --node[above,midway,sloped] {\scriptsize 2} (p5);

\end{tikzpicture}
    }
\caption{}
\label{fig:PNPost}
    \end{subfigure}
    \caption{A simple Petri net and the resulting marking on firing t.}
    \label{fig:pnex}
\end{figure}

\paragraph*{\texorpdfstring{\nPN}{nuPN}}

We recall $\nu$PN as introduced by~\cite{LazicS16}. Let $\Upsilon$ and $\X$ be disjoint infinite sets of variables. The variables in $\X$ (denoted by $x_i$ for some $i\in\mathbb{N}$) are called \textit{standard} variables, while those in $\Upsilon$ are called \textit{fresh}. 
Let $Vars\defeq \X \bigcup \Upsilon$.
\begin{definition}\label{dfn:nuPN}
A \nPN is a PN $\D=\tup{P,T,F}$ with the provision that $F:(P\times T) \bigcup (T \times P) \to Vars^{\oplus}$ is a flow function such that, for each $t\in T$, $\Upsilon\cap\prefun(t)=\emptyset$ and $\postfun(t)\setminus\Upsilon\subseteq\prefun(t)$, where $\prefun(t)=\bigcup_{p\in P} \supp(F(p,t))$ and $\postfun(t)=\bigcup_{p\in P} \supp(F(t,p))$.
\end{definition}

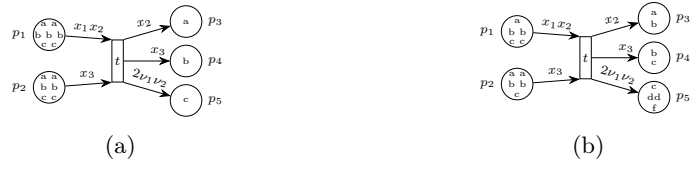
\begin{figure}[t]
    \centering
    \begin{subfigure}[b]{.49\textwidth}\centering
    \scalebox{0.7}{
    \begin{tikzpicture}

\node[place,label={[name=p2Lab]left:\scriptsize $p_2$}](p2)at (.2,-0.5){};
\node[place,label={[name=p1Lab]left:\scriptsize $p_1$}](p1)at (.2,.5){};

\node[transvert] (t)at (1.5,0){};
\node at (t){\scriptsize $t$};

\node[place,label={[name=p3Lab]right:\scriptsize $p_3$}](p3)at (2.8,.75){};
\node[place,label={[name=p4Lab]right:\scriptsize $p_4$}](p4)at (2.8,0){};
\node[place,label={[name=p5Lab] right:\scriptsize $p_5$}](p5)at (2.8,-.75){};

\draw [->] (p1) to node [above, sloped]  (TextNode1) {\scriptsize $x_1 x_2$} (t.north west);

\draw [->] (p2) -- node[above,midway,sloped] {\scriptsize $x_3$}(t.south west);
\draw [->] (t.north east) -- node[above,midway,sloped] {\scriptsize $x_2$}(p3);

\draw [->] (t) -- node[above,midway,sloped,pos=.75] {\scriptsize $x_3$}(p4);
\draw [->,bend right] (t.south east) --node[above,midway,sloped] {\scriptsize $2\nu_1\nu_2$} (p5);

\node at($(p1)+(0,.2)$){\tiny a a };
\node at($(p1)$){\tiny b b b};
\node at($(p1)-(0,.2)$){\tiny c c};

\node at($(p2)+(0,.2)$){\tiny a a};
\node at($(p2)$){\tiny b b};
\node at($(p2)-(0,.2)$){\tiny c c};

\node at($(p3)$){\tiny a};
\node at($(p4)$){\tiny b};
\node at($(p5)$){\tiny c};
\end{tikzpicture}
    }
\caption{}
\label{fig:NuPNInit}
    \end{subfigure}
    \hfill
    \begin{subfigure}[b]{.49\textwidth}\centering
       \scalebox{0.7}{   
   \begin{tikzpicture}

\node[place,label={[name=p2Lab]left:\scriptsize $p_2$}](p2)at (.2,-0.5){};
\node[place,label={[name=p1Lab]left:\scriptsize $p_1$}](p1)at (.2,.5){};

\node[transvert] (t)at (1.5,0){};
\node at (t){\scriptsize $t$};

\node[place,label={[name=p3Lab]right:\scriptsize $p_3$}](p3)at (2.8,.75){};
\node[place,label={[name=p4Lab]right:\scriptsize $p_4$}](p4)at (2.8,0){};
\node[place,label={[name=p5Lab] right:\scriptsize $p_5$}](p5)at (2.8,-.75){};

\draw [->] (p1) to node [above, sloped]  (TextNode1) {\scriptsize $x_1 x_2$} (t.north west);

\draw [->] (p2) -- node[above,midway,sloped] {\scriptsize $x_3$}(t.south west);
\draw [->] (t.north east) -- node[above,midway,sloped] {\scriptsize $x_2$}(p3);

\draw [->] (t) -- node[above,midway,sloped,pos=.75] {\scriptsize $x_3$}(p4);
\draw [->,bend right] (t.south east) --node[above,midway,sloped] {\scriptsize $2\nu_1\nu_2$} (p5);

\node at($(p1)+(0,.2)$){\tiny a };
\node at($(p1)$){\tiny b b };
\node at($(p1)-(0,.2)$){\tiny c c};

\node at($(p2)+(0,.2)$){\tiny a a};
\node at($(p2)$){\tiny b b};
\node at($(p2)-(0,.2)$){\tiny c };

\node at($(p3)+(0,.1)$){\tiny a};
\node at($(p3)-(0,.1)$){\tiny b};
\node at($(p4)+(0,.1)$){\tiny b};
\node at($(p4)-(0,.1)$){\tiny c};
\node at($(p5)+(0,.2)$){\tiny c};
\node at($(p5)$){\tiny dd};
\node at($(p5)-(0,.2)$){\tiny f};
\end{tikzpicture}
    }
\caption{}
\label{fig:NuPNPost}
    \end{subfigure}
    \caption{A $\nu$-net and the resulting configuration on firing $t$ with mode $e(x_1)=a$, $e(x_2)=b$, $e(x_3)=c$, $e(\nu_1)=d$, $e(\nu_2)=f$.}
    \label{fig:nupnex}
\end{figure}

For each $t\in T$, we set $\var(t)=\prefun(t)\cup\postfun(t)$. In this section, we work with a fixed arbitrary \nPN $\D=\tup{P,T,F}$ where $P=\{p_1,\dots,p_\ell\}$.
The flow $F_x$ of a variable $x\in \var$ is $F_x:(P\times T) \bigcup (T \times P) \to \mathbb{N}$ where $F_x (p,t)\defeq F(p,t)(x)$ and $F_x(t,p)\defeq F(t,p)(x)$. We denote the vector $\tup{F_x(p_1,t),\dots,F_x(p_\ell,t)}\in\mathbb{N}^\ell$ by $F_x(P,t)$ and the vector $\tup{F_x(t,p_1),\dots,F_x(t,p_\ell)}\in\mathbb{N}^\ell$ by $F_x(t,P)$.

The set of configurations of $\D$ is the set $(\mathbb{N}^P)^\oplus$. For each $t\in T$, let:
\[
    in(t)\defeq \sum_{x\in \X(t)} \fmset{F_x(P,t)}
\quad
    \out\defeq \sum_{\nu\in \Upsilon(t)} \fmset{F_\nu(t,P)}
\]
Given a configuration $M=\fmset{m_1,\dots,m_{\abs{M}}}$, a transition $t$ is fireable from $M$ if there is a function $e:\X(t)\longrightarrow\{1,\dots,\abs{M}\}$, called mode, such that, for each $x\in\X(t)$, $F_x(P,t)\leq m_{e(x)}$. We write $M\rightarrow^{t,e} M'$ if, for some configuration $M''$, 

\[ M=M''+\sum_{x\in\X(t)}\fmset{m_{e(x)}}
\quad
M'=M''+\out+\sum_{x\in\X(t)}\fmset{m_{e(x)}'}
\]
where, for $x\in\X(t)$, $m'_{e(x)}=m_{e(x)}-F_x(P,t)+F_x(t,P)$. Intuitively, the firing of $t$ with mode $e$ over $M$ applies $F_x$, for each $x\in X(t)$ to a distinct tuple $m\in M$ such that $m_{e(x)}\geq F_x(P,t)$ and replaces it with $m'_{e(x)}$. It also adds the new markings $F_{\nu}(t,P)$ for $\nu\in\Upsilon$. See Fig.~\ref{fig:nupnex} for an example.

\paragraph*{Elementary Object Systems}\label{sec:eos}

An EOS is a PN, where tokens have their own internal PN, and are referred to as system and object nets respectively. The transitions at both system and object levels may fire either autonomously or synchronously, according to well-defined events. We present the formal definition of EOSs below:

\begin{definition}\label{def:bussy14_eos}
An \emph{EOS} $\os$ is a tuple $\os=\tup{\hat{N},\N,d,\Theta}$ where:
\begin{enumerate}
\item $\hat{N}=\tup{\hat{P},\hat{T},\hat{F}}$ is a PN called \emph{system net}; $\hat{T}$ contains a special set $ID_{\hat{P}}=\{id_p\mid p\in \hat{P}\}\subseteq \hat{T}$ of \emph{idle transitions} such that, for each distinct $p,q\in \hat{P}$, we have $\hat{F}(p,id_p)=\hat{F}(id_p,p)=1$ and $\hat{F}(q,id_p)=\hat{F}(id_p,q)=0$.
\item $\N$ is a finite set of PNs, called \emph{object PNs}, such that $\blacksquare\in\N$ and if $(P_1,T_1,F_1), (P_2,T_2,F_2)\in\N \cup \hat{N}$, then $P_1\cap P_2=\emptyset$ and $T_1 \cap T_2 = \emptyset$.\footnote{This way, the system net and the object nets are pairwise disjoint.}
\item $d:\hat{P}\rightarrow \N$ is called the \emph{typing function}. 
\item $\Theta$ is a finite \emph{set of events} where each \emph{event} is a pair $(\hat{\tau},\theta)$, where $\hat{\tau}\in \hat{T}$ and  $\theta:\N \rightarrow \bigcup_{(P,T,F)\in\N} T^\oplus$,
    such that $\theta((P,T,F))\in T^\oplus$ for each $(P,T,F)\in\N$ and, if $\hat{\tau}=id_p$, then $\theta(d(p)) \neq \emptyset$.
\end{enumerate}
\end{definition}

A nested token is a system net token carrying an internal marking.
\begin{definition}
Let $\os=\tup{\hat{N},\N,d,\Theta}$ be an EOS. The set of \emph{nested tokens} $\nestTok(\os)$ of $\os$ is the set $\bigcup_{(P,T,F)\in\N} (d^{-1}{(P,T,F)}\times P^{\oplus})$. The set of \emph{nested markings} $\M(\E)$ of $\os$ is $\nestTok(\os)^{\oplus}$.
Given $\lambda,\rho\in \M(\E)$, we say that $\lambda$ is a \emph{sub-marking} of $\mu$ if $\lambda \sqleq \mu$.
\end{definition}

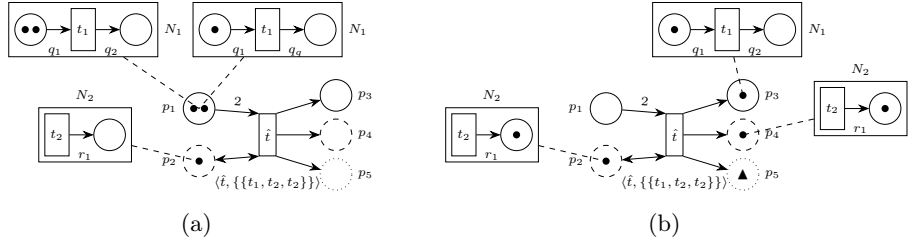
\begin{figure}[t]
    \centering
    \begin{subfigure}[b]{.49\textwidth}\centering
    \scalebox{0.7}{
    \begin{tikzpicture}

\node[place,dashed,label={[name=p2Lab]left:\scriptsize $p_2$},tokens=1](sp2)at (.2,-0.5){};
\node[place,label={[name=p1Lab]left:\scriptsize $p_1$},tokens=2](sp1)at (.2,.5){};
\node[transvert] (st)at (1.5,0){\scriptsize $\hat{t}$};
\node[](lab) at (1.5,-.9){\scriptsize $\langle \hat{t},\{\{ t_1,t_2,t_2\}\}\rangle$};
\node[place,label={[name=p3Lab]right:\scriptsize $p_3$}](sp3)at (2.8,.75){};
\node[place,dashed,label={[name=p4Lab]right:\scriptsize $p_4$}](sp4)at (2.8,0){};
\node[place,dotted,label={[name=p5Lab] right:\scriptsize $p_5$}](sp5)at (2.8,-.75){};
\draw [->] (sp1) to node [above, sloped]  (TextNode1) {\scriptsize $2$} (st.north west);
\draw [<->] (sp2) -- node[above,midway,sloped]{}(st.south west);
\draw [->] (st.north east) -- node[above,midway,sloped] {}(sp3);
\draw [->] (st) -- node[above,midway,sloped,pos=.75] {}(sp4);
\draw [->,bend right] (st.south east) --node[above,midway,sloped] {} (sp5);

\begin{scope}[xshift=-3cm,yshift=2cm]
    \node[place,tokens=2,label={below right:\scriptsize $q_1$}](p0)at (0,0){};
    \node[transvert] (t)at (1,0){\scriptsize $t_1$};
    \node[place,label={below left:\scriptsize $q_2$}](p1)at (2,0){};
    \draw[->](p0) -- (t);
    \draw[->](t) -- (p1);
    \node[draw=black,fit={(p0)(t)(p1)},label={right:\scriptsize $N_1$}](obj0){};
    \draw[dashed] (obj0) -- (sp1.center);
\end{scope}

\begin{scope}[xshift=0.5cm,yshift=2cm]
    \node[place,tokens=1,label={below right:\scriptsize $q_1$}](p0)at (0,0){};
    \node[transvert] (t)at (1,0){\scriptsize $t_1$};
    \node[place,label={below left:\scriptsize $q_q$}](p1)at (2,0){};
    \draw[->](p0) -- (t);
    \draw[->](t) -- (p1);
    \node[draw=black,fit={(p0)(t)(p1)},label={right:\scriptsize $N_1$}](obj0){};
\draw[dashed] (obj0) -- (sp1.center);
\end{scope}

\begin{scope}[xshift=-3.5cm,yshift=0cm]
   \node[transvert] (t)at (1,0){\scriptsize $t_2$};
    \node[place,label={below left:\scriptsize $r_1$}](p0)at (2,0){};
    \draw[->](t) -- (p0);
    \node[draw=black,fit={(p0)(t)},label={above:\scriptsize $N_2$}](obj0){};
\draw[dashed] (obj0) -- (sp2.center);
\end{scope}
\end{tikzpicture}
    }
\caption{}
\label{fig:EOSInit}
    \end{subfigure}
    \hfill
    \begin{subfigure}[b]{.49\textwidth}\centering
       \scalebox{0.7}{   
   \begin{tikzpicture}

\node[place,dashed,label={[name=p2Lab]left:\scriptsize $p_2$},tokens=1](sp2)at (.2,-0.5){};
\node[place,label={[name=p1Lab]left:\scriptsize $p_1$}](sp1)at (.2,.5){};
\node[transvert] (st)at (1.5,0){\scriptsize $\hat{t}$};
\node[](lab) at (1.5,-.9){\scriptsize $\langle \hat{t},\{\{ t_1,t_2,t_2\}\}\rangle$};
\node[place,label={[name=p3Lab]right:\scriptsize $p_3$},tokens=1](sp3)at (2.8,.75){};
\node[place,dashed,label={[name=p4Lab]right:\scriptsize $p_4$},tokens=1](sp4)at (2.8,0){};
\node[place,dotted,label={[name=p5Lab] right:\scriptsize $p_5$}](sp5)at (2.8,-.75){};
\node at(sp5){$\blacktriangle$};
\draw [->] (sp1) to node [above, sloped]  (TextNode1) {\scriptsize $2$} (st.north west);
\draw [<->] (sp2) -- node[above,midway,sloped]{}(st.south west);
\draw [->] (st.north east) -- node[above,midway,sloped] {}(sp3);
\draw [->] (st) -- node[above,midway,sloped,pos=.75] {}(sp4);
\draw [->,bend right] (st.south east) --node[above,midway,sloped] {} (sp5);

\begin{scope}[xshift=1.5cm,yshift=2cm]
    \node[place,tokens=1,label={below right:\scriptsize $q_1$}](p0)at (0,0){};
    \node[transvert] (t)at (1,0){\scriptsize $t_1$};
    \node[place,label={below left:\scriptsize $q_2$}](p1)at (2,0){};
    \draw[->](p0) -- (t);
    \draw[->](t) -- (p1);
    \node[draw=black,fit={(p0)(t)(p1)},label={right:\scriptsize $N_1$}](obj0){};
\draw[dashed] (obj0) -- (sp3.center);
\end{scope}

\begin{scope}[xshift=-3.5cm,yshift=0cm]
   \node[transvert] (t)at (1,0){\scriptsize $t_2$};
    \node[place,tokens=1,label={below left:\scriptsize $r_1$}](p0)at (2,0){};
    \draw[->](t) -- (p0);
    \node[draw=black,fit={(p0)(t)},label={above:\scriptsize $N_2$}](obj0){};
\draw[dashed] (obj0) -- (sp2.center);
\end{scope}

\begin{scope}[xshift=3.5cm,yshift=0.5cm]
   \node[transvert] (t)at (1,0){\scriptsize $t_2$};
    \node[place,tokens=1,label={below left:\scriptsize $r_1$}](p0)at (2,0){};
    \draw[->](t) -- (p0);
    \node[draw=black,fit={(p0)(t)},label={above:\scriptsize $N_2$}](obj0){};
\draw[dashed] (obj0) -- (sp4.center);
\end{scope}

\end{tikzpicture}
    }
\caption{}
\label{fig:EOSPost}
    \end{subfigure}
    \caption{An EOS depicting the firing of synchronized event $\langle \hat{t},\fmset{ t_1,t_2,t_2}\rangle$ for a given mode $(\lambda,\rho)$, where $\lambda=\langle p_1,\fmset{q_1,q_1}\rangle+\langle p_1,\fmset{q_1}\rangle+\langle p_2,\fmset{r_1}\rangle$, 
    $\rho=\langle p_2,\fmset{r_1}\rangle+\langle p_3,\fmset{q_1}\rangle+\langle p_4,\fmset{q_1}\rangle+\langle p_5,\varepsilon\rangle$, $d(p_1)=d(p_3) = N_1$,
    $d(p_2)=d(p_4) = N_2$,
    $d(p_5)=\blacktriangle$.}
    \label{fig:eosex}
\end{figure}

As depicted in Fig.~\ref{fig:eosex}, EOSs inherit the graphical representation of PNs with the provision that we represent nested tokens via a dashed line from the token in the system net place to an instance of the object type where the internal marking is represented in the standard PN way. However, if the nested token is $\tup{p,\varepsilon}$ for a system net place $p$ of type $\blacksquare$, we represent it with a black-token $\blacksquare$ on $p$. 
Each event $\tup{\hat{\tau},\theta}$ is depicted by labeling $\hat{\tau}$ by $\tup{\theta}$ (possibly omitting double curly brackets). If there are several events involving $\hat{\tau}$, then $\hat{\tau}$ has several labels.

\begin{remark}\label{rem:EOSsemantics}
The firing of an event $e=\tup{\tau,\theta}$ can be characterized as follows: first, merge, type by type, all objects handled by the preconditions of $\tau$
(identified by a mode, as in \nPNs)
, obtaining merged tuples; second, fire, at the same time, all transitions in $\theta$ on the merged tuples; third, non-deterministically distribute the updated tuples to new objects in the places indicated by the post-conditions of $\tau$. Possibly, some of the new objects have empty internal marking. 
The next definitions formalize this dynamics.
\end{remark}

\begin{definition}
Let $\os$ be an EOS $\tup{\hat{N},\N,d,\Theta}$. The \emph{projection operator $\Pi^1$} maps each nested marking $\mu=\sum_{i\in I}\tup{\hat{p}_i,M_i}$ for $\E$ to the PN marking $\sum_{i\in I}\hat{p}_i$ for $\hat{N}$. Given an object net $N\in\N$, the \emph{projection operator $\Pi^2_N$} maps each nested marking $\mu=\sum_{i\in I}\tup{\hat{p}_i,M_i}$ for $\E$ to the PN marking $\sum_{j\in J} M_j$ for ${N}$ where $J=\{i\in I\mid d(\hat{p}_i)=N\}$.
\end{definition}

To define the enabledness condition, we need the following notation. We set $\prefun_{N}(\theta(N))=\sum_{i\in I}\prefun_N(t_i)$ where $(t_i)_{i\in I}$ is an enumeration of $\theta(N)$ counting multiplicities. We analogously set $\postfun_{N}(\theta(N))=\sum_{i\in I}\postfun_N(t_i)$.

\begin{definition}\label{def:bussy14_enable} 
Let $\os$ be an EOS $\tup{\hat{N},\N,d,\Theta}$. Given an event $e=\tup{\hat{\tau},\theta}\in \Theta$ and markings $\lambda,\rho\in\M(\os)$, the \emph{enabledness condition} $\Phi(\tup{\hat{\tau},\theta},\lambda,\rho)$ holds iff
\begin{align*}
\Pi^1(\lambda)=\prefun_{\hat{N}}(\hat{\tau})\ \land \Pi^1(\rho)=\postfun_{\hat{N}}(\hat{\tau})\ \land
\forall N\in \N,\ \Pi^2_N(\lambda)\geq \prefun_N(\theta(N))\ \land\\
\forall N\in\N,\ \Pi^2_N(\rho)=\Pi^2_N(\lambda)-\prefun_N(\theta(N))+\postfun_N(\theta(N))
\end{align*}
The event $e$ is \emph{enabled with mode $(\lambda,\rho)$ on a marking $\mu$} iff $\Phi(e,\lambda,\rho)$ holds and $\lambda\sqleq \mu$.
Its firing results in the step $\mu\xrightarrow{(e,\lambda,\rho)}\mu-\lambda+\rho$.
\end{definition}

The coverability problem for EOSs is defined in the usual way, i.e., it asks whether there is a run (sequence of event firings) from an initial marking $\mu_0$ to a marking $\mu_1$ that covers a target marking $\mu_f$ with respect to the order $\leq_f$ such that $\mu\leq_f\mu'$ iff $\mu'$ is obtained from $\mu$ by adding
\begin{enumerate*}
    \item tokens in the inner markings of available nested tokens and/or
    \item nested tokens with some internal marking on the system net places.
\end{enumerate*}
With a slight abuse of notation, where it does not cause confusion, we denote $\leq_f$ by $\leq$.

It is known that EOS coverability is undecidable (Th. 4.3 in~\cite{kohler-busmeier_survey_2014}). However, coverability is decidable on the fragment of \emph{conservative EOSs} (cEOSs; Th. 5.2 in~\cite{kohler-busmeier_survey_2014}), where, for each system net transition $t$, if $t$ consumes a nested token on a place of type $N$, then it produces at least one token on a place of the same type. 

\begin{definition}
A cEOS is an EOS $\os=\tup{\hat{N},\N,d,\Theta}$ with $\hat{N} = \tup{\hat{P},\hat{T},\hat{F}}$ where, for all $\hat{t} \in \hat{T}$, $d(\support{\prefun_{\hat{N}}(\hat{t})})\subseteq d(\support{\postfun_{\hat{N}}(\hat{t})})$

\end{definition}

\begin{definition}
\textit{cEOS coverability} is the following decision problem:\\
\textbf{Input:} A cEOS $\os$, an initial configuration $M_0$, and a target configuration $M_t$.\\
\textbf{Output:} Whether there is a configuration $M_1$ reachable from $M_0$ and $M_t\leq M_1$.
\end{definition}

\begin{figure}[t]
    \centering
    \scalebox{0.65}{
    \begin{tikzpicture}

\begin{scope}
    \path [rectangle, fill=gray!10] (.75,.75) |- (-.75,-.5) |- cycle;
    \node[place,label={[name=startPLab]above:\scriptsize $\text{StartP}$}](start)at(0,0){};
    \node at(start) (playertoken){$\bullet$};
\begin{scope}[yshift=1.25cm,xshift=-2.1cm]
    
    \node[place,label={[name=aliveLab]above:\scriptsize alive}](alive)at(-.5,0){};
    \node at (alive) {$\bullet$};
    \node[place,label={[name=deadLab]above:\scriptsize dead}](dead)at(1.5,0){};

    \node [transvert,label={[name=dieLab]above:\scriptsize die}] (die) at (.5,0){};
    \node[transvert,label={[name=birthLab]above:\scriptsize birth}] (birth)at (-1.5,0){};
    \draw [->](birth) -- (alive);
    \draw [->](die) -- (dead);
    \draw [->](alive) -- (die);

    \node[draw,dashed,fit={(birth)(birthLab)(dead)(alive)(deadLab)(aliveLab)}](monster){};
\end{scope}

\begin{scope}[xshift=-2cm,yshift=-.75cm]
    \node[place,label={[name=expLab]left:\scriptsize Exp}](exp)at(0,0){};

    \node [transhor,label={[name=checkexpLab]left:\scriptsize Check Exp}](checkexp)at (0,1){};
    \draw [<->] (exp) -- node [midway, left] {$3$}(checkexp);

    \node [transhor,label={[name=addexpLab]left:\scriptsize Add Exp}](addexp)at (0,-1){};
    
    \draw[->] (addexp) -- (exp);

    \node[draw,dashed,fit={(addexp)(addexpLab)(exp)(expLab)(checkexp)(checkexpLab)}](player){};
\end{scope}

\begin{scope}[yshift=-2.25cm]
\node[transhor,label={[name=firstmovementLab,align=center]left:\scriptsize first move}](firstMovement){};    
\end{scope}    
\end{scope}

\begin{scope}[xshift=2cm]
    \path [rectangle, fill=gray!10] (.75,1.25) |- (-.75,-1.25) |- cycle;
    \node[place,label={[name=T1PLab]above:\scriptsize $\text{T1P}$}](T1P)at(0,-.8){};
    \node[pentagon,label={[name=T1MLab]below:\scriptsize $\text{T1M}$}](T1M)at(0,.5){};
    \node at (T1M) (T1Mtoken) {$\bullet$};

\begin{scope}[yshift=1.5cm]
    \node[transhor,label={[name=fightLab]above:\scriptsize fight $\tup{\mathit{die},\mathit{add Exp}}$}](fight){};
    \draw [<->](fight) -- (T1M);
    \draw [<->] (fight.east) -- ($(fight.east)+(.5,0)$) |- (T1P);
\end{scope}

\end{scope}

\begin{scope}[xshift=4cm]
    \path [rectangle, fill=gray!10] (.75,1.25) |- (-.75,-1.25) |- cycle;
    \node[place,label={[name=T2PLab]above:\scriptsize T2P}](T2P)at(0,-.8){};
    \node[pentagon,label={[name=T2MLab]below:\scriptsize T2M}](T2M)at(0,.5){};

\begin{scope}[yshift=1.5cm]
    \node[transhor,label={[name=sourceLab,align=center]above right:\scriptsize spawner $\tup{birth}$}] (source){};
\end{scope}    
\end{scope}

\begin{scope}[xshift=6cm]
    \path [rectangle, fill=gray!10] (.75,.75) |- (-.75,-.5) |- cycle;
    \node[place,label={[name=WinPLab]above:\scriptsize $\text{WinP}$}](WinP)at(0,0){};

\begin{scope}[yshift=-1.5cm]
\node[transhor,label={[name=movementLab,align=center]below:\scriptsize last move $\tup{\mathit{Check Exp}}$}](lastMovement){};
    
\end{scope}

\end{scope}

\begin{scope}[xshift=8cm]
    \node[triangle,label={[name=playerMoveLab,text width = 1cm]right:\scriptsize Player Turn}](playerMove)at(-.5,0){};
    \node at (playerMove){$\blacksquare$};
    \node[triangle,label={[name=MonsterMoveLab,text width = 1cm]right:\scriptsize Monster Turn}](monsterMove)at(1,0){};
\end{scope}

\draw [dashed](player) -- (playertoken.center);
\draw [dashed](monster) -- (T1Mtoken.center);

\draw[->](T2P)|-(lastMovement);
\draw[->](playerMove.south west)|-(lastMovement.north east);
\draw[<-](monsterMove.south west)|-(lastMovement.south east);
\draw[<-](WinP)--(lastMovement);

\draw[->](source)--(T2M);
\draw[->](source.south east)-|(playerMove.north);
\draw[<->](source.north east)-|(monsterMove.north);

\draw[->] (start) -- (firstMovement);
\draw[->] (firstMovement.north east) |- (T1P);
\draw[<-] (firstMovement.north east) -| (playerMove.south east);
\draw[->] (firstMovement.south east) -| (monsterMove.south east);
    \end{tikzpicture}
    }
    \caption{A cEOS representing a maze of a RPG game.}
    \label{fig:rpg}
\end{figure}
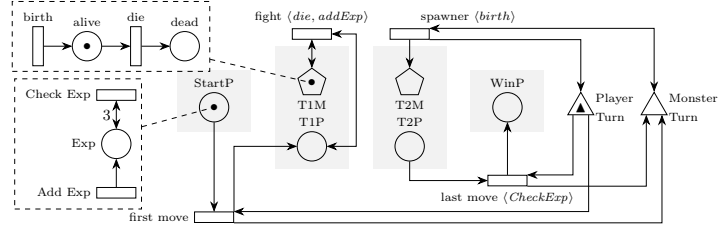

 
\begin{example}\label{ex:rpg}
We use cEOS coverability to model the gameplay of a Roguelike game in Fig.~\ref{fig:rpg}. Existing work~\cite{rpgpnse24} models the navigation map into a Petri net, to determine the viability of an optimal path to reach a goal in a game. This has practical relevance in domains such as navigation systems used in the military, where optimal path planning is crucial in warfare. We show how cEOSs coverability can be flexibly employed in this setting. We represent the players by standard Petri nets places and the adversary by pentagons. The goal of the player is the following: starting at the place \textit{StartP} with a given number of experience units (\textit{Exp}) and reaching the winning destination \textit{WinP}. This can be formalized as a coverability problem. The player and adversary each take turns to \textit{move} and if they are in the same tile they may \textit{fight}. In a successful run, at every tile of the game (in grey boxes), the player (say, $T1P$ at tile $1$) must \textit{move} through the level with the objective of reaching \textit{WinP}. However, the gates between the tiles can be opened only if the player accrues a sufficient number of experience units \textit{Exp}. If the opportunity presents itself, the $alive$ adversaries may fight with the player and when an adversary dies, the $addExp$ increments the player's experience statistics by one unit, represented in the nested player net. If the player loses at any level, he cannot progress further. To make it more challenging, there are $spawner$ transitions, in the game which synchronize with $birth$ transitions in the nested adversary net, to generate new adversaries. In the $lastmove$, if the player has successfully gained sufficient experience by clearing all $3$ levels, $CheckExp$ allows the player to reach the winning destination $WinP$. We have described a basic structure of navigating the maze of an RPG game. It is easy to see how this could be extended to incorporate various structures within the tile of the game such as loops and intersections to make the navigation challenging, or to introduce various attributes of the player (such as treasure collected, weapons, special skills) and of the adversary (such as lethality levels, attack styles) and player-adversary interaction, to enrich the gameplay.
\end{example}
\section{Extended \texorpdfstring{\nPN}{nuPN}}\label{sec:extend}
We now extend the \nPN model so as to allow restricted forms of data-aware whole-place operations in the style of UDNs. We call the resulting model \textit{channel-\nPN}(\GnPN). We also provide a fragment, called \textit{rename-\nPN}(\rnPN), with equivalent expressivity, i.e., at the cost of a polynomial blowup, each \GnPN is captured by a corresponding \rnPN. \rnPN prove themselves valuable in Sec.~\ref{sec:toEOS}.

\subsection{Channel-\texorpdfstring{\nPN}{nPN}}
\newcommand{\confs}{\texttt{Confs}}
We introduce \GnPN, 
whose transitions extend arbitrary \nPN transitions by transfers with renaming (organized in \textit{channels} of transfers).
The mechanism to fire a \GnPN transition $t$ is similar to UDNs (and, more in general, to affine nets~\cite{FinkelMP04}):
\begin{enumerate*}
    \item 
    the pre-conditions are fired,
    \item 
    the transfers with renaming (defined by transformation matrices \(G_t(x_i,x_j)\) for $x_i,x_j\in\X(t)$) are performed
    \item  
    the post-conditions are fired.
\end{enumerate*}

The rows and columns of $G_t(x_i,x_j)$ are indexed by places. $G_t(x_i,x_j)(p,q)=n$ means that, after firing the pre-conditions, the transition puts (up to post-conditions) $n$ times on $q$, in the tuple instantiating $x_j$, the number of tokens available on $p$ from the tuple instantiating $x_i$. In order to capture only transfers with renaming, these matrices must be boolean and enjoy some inter-relationship as in item \ref{G:restrict} of Def.\ref{def:CNUPN} below. We use the following notation: if, for some set $P$ of places, $M\in\{0,1\}^{P\times P}$ is a matrix and $p\in P$ is a place, $M[p]$ denotes the row of $M$ indexed by $p$.
Moreover, given a place $p\in P$, we denote by $\delta_p\in\{0,1\}^{\abs{P}}$ the vector such that, for each $q\in P$, $\delta_p(q) = 1$ if and only if $p=q$. For example, given a tuple $m$ and a place $p$, the tuple $m(p)\delta_p$ is the projection of $m$ to $p$.\footnote{Note that $m(p)\delta_p$ is the multiplication of the number $m(p)$ with the vector $\delta_p$.}

\begin{definition} \label{def:CNUPN}
    A \emph{channel \nPN (\GnPN)} is a tuple \(N = (P,T,F,G)\) where
    \begin{enumerate}
        \item \((P,T,F)\) is a \nPN.
        \item \label{G:restrict}G maps each \(t \in T\) to a function \(G_t: \X(t) \times \X(t) \rightarrow \{0,1\}^{\abs{P} \times \abs{P}}\) such that, for each variable \(x_i \in \X(t)\) and place $p\in P$:
        \begin{enumerate}
            \item for each variable $x_j\in\X(t)$ either $G_t(x_i,x_j)[p]=0$ or $\exists q\in P\;G_t(x_i,x_j)=\delta_q$, i.e., $G_t(x_i,x_j)[p]$ contains at most a single $1$.
            \item if \(\exists q \in P\,\exists x_j\in\X(t)\setminus\{x_i\}\;G_t(x_i,x_j)[p] = \delta_q\), then \(\forall x_k \in \X(t)\setminus \{x_j\}\; G_t(x_i,x_k)[p] = 0^{\abs{P}}\), else \(G_t(x_i,x_i)[p] = \delta_p\).
        \end{enumerate}
    \end{enumerate}
\end{definition}

The graphical representation of a \GnPN $N'=(P,T,F,G)$ adds the representation of $G_t$, for each $t\in T$, on top of the representation of the \nPN $N=(P,T,F)$. This is achieved by special arrows called \textit{channels}.
A channel $c$ is a pair of decorated arrows (e.g., double or triple arrows in Fig.~\ref{fig:CNuPN_example}), one to $t$, called pre-channel of $c$, and one from $t$, called post-channel of $c$. Different channels use different decorations. The pre- and post-channels of $c$ are labeled by two non-empty sequences $\sigma_\text{pre}$ and $\sigma_\text{post}$ of variables, respectively, such that $\abs{\sigma_\text{pre}} =\abs{ \sigma_\text{post}}$. Moreover, for each place $p$, each variable $x\in\X(t)$ can occur at most once in the label of at most one pre-channel from $p$.\footnote{This condition is necessary to forbid the representation of transfers with duplication.}
Intuitively, $c$ represents $\abs{\sigma_\text{pre}}$ transfers from $p$ to $q$ while renaming the tokens taken from the tuple instantiating $\sigma_\text{pre}(i)$ as tokens of the tuple instantiating $\sigma_\text{post}(i)$.\footnote{Note that, given a sequence $\sigma$, $\sigma(i)$ denotes the $i$-th element of $\sigma$.}
Technically, $c$ states that $G_t(\sigma_\text{pre}(i),\sigma_\text{post}(i))[p]=\delta_q$, for $i\leq \abs{\sigma_\text{pre}}$ and $p,q\in P$.\footnote{Consequently, for each $x\in\X(t)\setminus\{\sigma_\text{post}(i)\}$, $G_t(\sigma_\text{pre}(i),x)[p]=0$.} If a variable $x\in\X(t)$ does not label any pre-channel from $p$, then $G_t(x,x)[p]=\delta_p$.\footnote{Consequently, for each $x'\in\X(t)\setminus\{x\}$, $G_t(x,x')[p]=0$.} Note that this notation can represent all and only \GnPN.
We call \textit{special} each \GnPN transition with at least one channel in its graphical representation.

\begin{example} \label{PN=>GnPN}
    Each PN \(N=(P,T,F)\) corresponds (cf. \GnPN semantics in Def.~\ref{def:stepCNUPN} below) to the \GnPN \(N'=(P,T,F,G)\) where, for each $t\in T$ and $x\in\X(t)$, $y\in \X(t)\setminus\{y\}$, $
    G_t(x,x)=\textbf{Id}$,\footnote{Note that \textbf{Id} denotes the identity matrix of the right size. } and $G_t(x,y)=0$. Both $N$ and $N'$ have the same graphical representation. $N'$ has no special transition.
\end{example}

\begin{example} \label{example:CNuPN1}    
    Fig.~\ref{fig:CNuPN_example} depicts a \GnPN \(N = (P,T,F,G)\) with several configurations (cf. Def.~\ref{def:stepCNUPN} below) where \(P = \{p_1,p_2,p_3,p_4,p_5\}\), \(T = \{t\}\), \(\X(t) = \{x_1,x_2,x_3\}\), \(\Upsilon(t) = \{\nu_1,\nu_2\}\), and for each place $p\in P$ and $x\in\var(t)$, $G_t(x_1,x)[p]$ is $\delta_p$ if $x=x_1$ and $0$ otherwise and:
    {\footnotesize \begin{align*}
        F_x(p,t)=&\begin{cases}
        1& \text{if }p=p_2,x\in\{x_2,x_3\}\\
            1& \text{if }p=p_1,x=x_1\\
            0& \text{otherwise}
        \end{cases}
        &
        F_x(t,p)=&\begin{cases}
            1& \text{if }p=p_3,x=x_3\\
            2& \text{if }p=p_3,x=x_2\\
            2& \text{if }p=p_5,x\in\{\nu_1,\nu_2\}\\
            0& \text{otherwise}
        \end{cases}
    \\
    G_t(x_2,x)[p]=&\begin{cases}
        \delta_{p_3}&\text{if }p=p_1, x=x_3\\
        \delta_{p}&\text{if }p\neq p_1, x=x_2\\
        0&\text{otherwise}
    \end{cases}&
    G_t(x_3,x)[p]=&\begin{cases}
        \delta_{p_4}&\text{if }p=p_2, x=x_1\\
        \delta_{p}&\text{if }p\neq p_2, x=x_3\\
        0&\text{otherwise}
    \end{cases}
\end{align*}}%
 
\end{example}

\begin{figure}[t]
\centering
    \begin{subfigure}{0.24\textwidth}
    \scalebox{0.9}{
        \begin{tikzpicture}

\node[place,label={[name=p2Lab]below:\scriptsize $p_2$}](p2)at (.2,0){};
\node[place,label={[name=p1Lab]above:\scriptsize $p_1$}](p1)at (.2,.75){};

\node[transvert] (t)at (1.5,0){};
\node at (t){\scriptsize $t$};

\node[place,label={[name=p3Lab]above:\scriptsize $p_3$}](p3)at (2.8,.75){};
\node[place,label={[name=p4Lab]below:\scriptsize $p_4$}](p4)at (2.8,0){};
\node[place,label={[name=p5Lab]below right:\scriptsize $p_5$}](p5)at (2.35,-.75){};

\draw [->] (p1) to [bend left=30]  node [above, sloped]  (TextNode1) {\scriptsize $x_1$} (t.north west);
\draw [->] (p2) to [bend right=30]  node [below, sloped]  (TextNode1) {\scriptsize $x_2 x_3$} (t.south west);
        
\draw [->,myDouble] (p1) -- node[above,midway,sloped] {\scriptsize $x_2$}(t);
\draw [->,triple] (p2) -- node[below,midway,sloped] {\scriptsize $x_3$}(t);
\draw [->,myDouble] (t) -- node[above,midway,sloped] {\scriptsize $x_3$}(p3);
\draw [->] (t.north east) to [bend left=30]  node [above, sloped]  (TextNode2) {\scriptsize $2x_2,x_3$} (p3);

\draw [->,triple] (t) -- node[below,midway,sloped,pos=.75] {\scriptsize $x_1$}(p4);
\draw [->,bend right] (t.south east) --node[below,midway,sloped] {\scriptsize $2\nu_1\nu_2$} (p5);

\node at($(p1)+(0,.2)$){\tiny a a };
\node at($(p1)$){\tiny b b b};
\node at($(p1)-(0,.2)$){\tiny c c};

\node at($(p2)+(0,.2)$){\tiny a a};
\node at($(p2)$){\tiny b b};
\node at($(p2)-(0,.2)$){\tiny c c};

\node at($(p3)$){\tiny a};
\node at($(p4)$){\tiny b};
\node at($(p5)$){\tiny c};
\end{tikzpicture}}
        \caption{}
        \label{fig:CNuPN_init}
    \end{subfigure}
    \begin{subfigure}{0.24\textwidth}\scalebox{0.9}{
        \begin{tikzpicture}

\node[place,label={[name=p2Lab]below:\scriptsize $p_2$}](p2)at (0.2,0){};
\node[place,label={[name=p1Lab]above:\scriptsize $p_1$}](p1)at (0.2,.75){};

\node[transvert] (t)at (1.5,0){};
\node at (t){\scriptsize $t$};

\node[place,label={[name=p3Lab]above:\scriptsize $p_3$}](p3)at (2.8,.75){};
\node[place,label={[name=p4Lab]below:\scriptsize $p_4$}](p4)at (2.8,0){};
\node[place,label={[name=p5Lab]below right:\scriptsize $p_5$}](p5)at (2.35,-.75){};

\draw [->] (p1) to [bend left=30]  node [above, sloped]  (TextNode1) {\scriptsize $x_1$} (t.north west);
\draw [->] (p2) to [bend right=30]  node [below, sloped]  (TextNode1) {\scriptsize $x_2 x_3$} (t.south west);
        
\draw [->,myDouble] (p1) -- node[above,midway,sloped] {\scriptsize $x_2$}(t);
\draw [->,triple] (p2) -- node[below,midway,sloped] {\scriptsize $x_3$}(t);
\draw [->,myDouble] (t) -- node[above,midway,sloped] {\scriptsize $x_3$}(p3);
\draw [->] (t.north east) to [bend left=30]  node [above, sloped]  (TextNode2) {\scriptsize $2x_2,x_3$} (p3);

\draw [->,triple] (t) -- node[below,midway,sloped,pos=.75] {\scriptsize $x_1$}(p4);
\draw [->] (t.south east) --node[below,midway,sloped,bend right] {\scriptsize $2\nu_1\nu_2$} (p5);

\node at($(p1)+(0,.2)$){\tiny a};
\node at($(p1)$){\tiny b b b};
\node at($(p1)-(0,.2)$){\tiny c c};

\node at($(p2)+(0,.2)$){\tiny a a};
\node at($(p2)$){\tiny b};
\node at($(p2)-(0,.2)$){\tiny c};

\node at($(p3)$){\tiny a};
\node at($(p4)$){\tiny b};
\node at($(p5)$){\tiny c};
\end{tikzpicture}}
        \caption{}
        \label{fig:CNuPN_step1}
    \end{subfigure}
    \begin{subfigure}{0.24\textwidth}\scalebox{0.9}{
        \begin{tikzpicture}

\node[place,label={[name=p2Lab]below:\scriptsize $p_2$}](p2)at (0.2,0){};
\node[place,label={[name=p1Lab]above:\scriptsize $p_1$}](p1)at (0.2,0.75){};

\node[transvert] (t)at (1.5,0){};
\node at (t){\scriptsize $t$};

\node[place,label={[name=p3Lab]above:\scriptsize $p_3$}](p3)at (2.8,0.75){};
\node[place,label={[name=p4Lab]below:\scriptsize $p_4$}](p4)at (2.8,0){};
\node[place,label={[name=p5Lab]below right:\scriptsize $p_5$}](p5)at (2.35,-0.75){};

\draw [->] (p1) to [bend left=30]  node [above, sloped]  (TextNode1) {\scriptsize $x_1$} (t.north west);
\draw [->] (p2) to [bend right=30]  node [below, sloped]  (TextNode1) {\scriptsize $x_2 x_3$} (t.south west);
        
\draw [->,myDouble] (p1) -- node[above,midway,sloped] {\scriptsize $x_2$}(t);
\draw [->,triple] (p2) -- node[below,midway,sloped] {\scriptsize $x_3$}(t);
\draw [->,myDouble] (t) -- node[above,sloped,midway] {\scriptsize $x_3$}(p3);
\draw [->] (t.north east) to [bend left=30]  node [above, sloped]  (TextNode2) {\scriptsize $2x_2,x_3$} (p3);

\draw [->,triple] (t) -- node[below,midway,sloped,pos=.75] {\scriptsize $x_1$}(p4);
\draw [->] (t.south east) --node[below,midway,sloped] {\scriptsize $2\nu_1\nu_2$} (p5);

\node at($(p1)+(0,.1)$){\tiny a};
\node at($(p1)-(0,.1)$){\tiny c c};

\node at($(p2)+(0,.1)$){\tiny a a};
\node at($(p2)-(0,.1)$){\tiny b};

\node at($(p3)+(0,.2)$){\tiny a};
\node at($(p3)$){\tiny c c c };

\node at($(p4)+(0,.1)$){\tiny a};
\node at($(p4)-(0,.1)$){\tiny b};

\node at($(p5)$){\tiny c};
\end{tikzpicture}}
        \caption{}
        \label{fig:CNuPN_step2}
    \end{subfigure}
    \begin{subfigure}{0.24\textwidth}\scalebox{0.9}{
        \begin{tikzpicture}

\node[place,label={[name=p2Lab]below:\scriptsize $p_2$}](p2)at (0.2,0){};
\node[place,label={[name=p1Lab]above:\scriptsize $p_1$}](p1)at (0.2,0.75){};

\node[transvert] (t)at (1.5,0){};
\node at (t){\scriptsize $t$};

\node[place,label={[name=p3Lab]above:\scriptsize $p_3$}](p3)at (2.8,0.75){};
\node[place,label={[name=p4Lab]below:\scriptsize $p_4$}](p4)at (2.8,0){};
\node[place,label={[name=p5Lab]below right:\scriptsize $p_5$}](p5)at (2.35,-0.75){};

\draw [->] (p1) to [bend left=30]  node [above, sloped]  (TextNode1) {\scriptsize $x_1$} (t.north west);
\draw [->] (p2) to [bend right=30]  node [below, sloped]  (TextNode1) {\scriptsize $x_2 x_3$} (t.south west);
        
\draw [->,myDouble] (p1) -- node[above,midway,sloped] {\scriptsize $x_2$}(t);
\draw [->,triple] (p2) -- node[below,midway,sloped] {\scriptsize $x_3$}(t);
\draw [->,myDouble] (t) -- node[above,sloped,midway] {\scriptsize $x_3$}(p3);
\draw [->] (t.north east) to [bend left=30]  node [above, sloped]  (TextNode2) {\scriptsize $2x_2,x_3$} (p3);

\draw [->,triple] (t) -- node[below,midway,sloped,pos=.75] {\scriptsize $x_1$}(p4);
\draw [->] (t.south east) --node[below,midway,sloped, bend right] {\scriptsize $2\nu_1\nu_2$} (p5);

\node at($(p1)+(0,.1)$){\tiny a};
\node at($(p1)-(0,.1)$){\tiny c c};

\node at($(p2)+(0,.1)$){\tiny a a};
\node at($(p2)-(0,.1)$){\tiny b};

\node at($(p3)+(0,.15)$){\tiny a c};
\node at($(p3)$){\tiny c c c};
\node at($(p3)-(0,.15)$){\tiny b b };

\node at($(p4)+(0,.1)$){\tiny a};
\node at($(p4)-(0,.1)$){\tiny b};

\node at($(p5)+(0,.2)$){\tiny c};
\node at($(p5)$){\tiny d d};
\node at($(p5)-(0,.2)$){\tiny f};

\end{tikzpicture}}
        \caption{}
        \label{fig:CNuPN_step3}
    \end{subfigure}
    \caption{The \GnPN discussed in Ex.~\ref{example:CNuPN1} and Ex.~\ref{example:CNuPN2}, with configurations.} 
    \label{fig:CNuPN_example}
\end{figure}
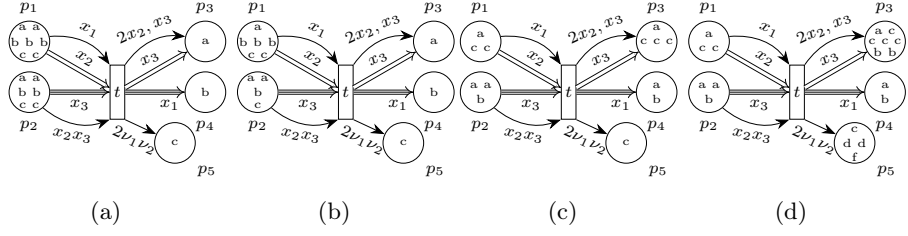

We define the semantics of \GnPN in the style of UDNs: the firing of a transition $t$ serializes its standard pre-conditions, the transfers, and, finally, the post-conditions. In Def.~\ref{def:stepCNUPN} below, $M''$ captures the tokens uninvolved with the firing, $\out$ accounts for the creation of fresh tuples via $\nu$ variables, $m''$ accounts for the firing of the standard preconditions, $\sum\limits_{x_i \in \X(t)} \left( m_{e(x_i)}'' \ast G_t(x_i,x_j) \right)$ formalizes the effect of the transfers with renaming to the tuple instantiating $x_j$,\footnote{$\ast$ denotes standard row-by-column multiplication} while $m'$ takes into account also the final firing of the standard post-conditions.

\begin{definition} \label{def:stepCNUPN}
A configuration $M$ of a \GnPN $N=(P,T,F,G)$ is a configuration of $(P,T,F)$, and a transition $t\in T$ is enabled in $N$ if it is enabled in $(P,T,F)$. A configuration $M'$ is reached from $M$ by firing a transition $t\in T$ with mode $e$, denoted by $M\rightarrow^{t,e}M'$, if
        \[ M=M''+\sum_{x\in\X(t)}\fmset{m_{e(x)}}
        \qquad
        M'=M''+\out+\sum_{x\in\X(t)}\fmset{m_{e(x)}'}
        \]
    where \(m_{e(x_j)}' = \sum\limits_{x_i \in \X(t)} \left( m_{e(x_i)}'' \ast G_t(x_i,x_j) \right) + F_{x_j}(t,P)\) and \(m_{e(x_i)}'' = m_{e(x_i)} - F_{x_i}(P,t)\).
    We define \(\rightarrow^\ast\) as the reflexive and transitive closure of \(\rightarrow^{t,e}\).
\end{definition}

\begin{example} \label{example:CNuPN2}

Fig.~\ref{fig:CNuPN_step3} depicts the firing of transition $t$ from Fig.~\ref{fig:CNuPN_init} with mode $e(x_1)= a$, $e(x_2)= b$, $e(x_3)= c$, $e(\nu_1)=d$, and $e(\nu_2)=f$. Fig~\ref{fig:CNuPN_step1} and Fig.~\ref{fig:CNuPN_step2} depict the intermediate steps.

\end{example}

\subsection{Rename-\texorpdfstring{\nPN}{nPN}}
A \rnPN is a \GnPN where all transitions are either \nPN transitions, i.e., without channels, or \GnPN transitions as in Fig.~\ref{fig:selectiveTransfer}. 

Specifically, Fig.~\ref{fig:selectiveTransferPre} depicts a \rnPN 
with a single special transition $t$ and a configuration. 

The configuration in Fig.~\ref{fig:selectiveTransferPost} is the configuration reached after firing $t$.

Note that, if $p_1$ did not contain any token bidden to $x_1$, then the transition could still fire, but no transfer happens.

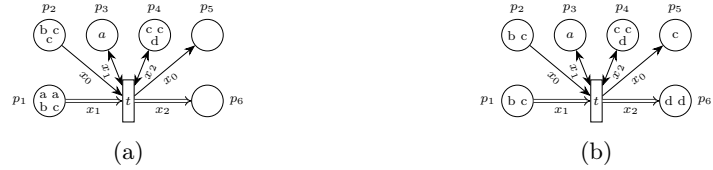
\begin{figure}[t]
    \centering
    \begin{subfigure}[b]{.49\textwidth}\centering
    \scalebox{0.7}{
    \begin{tikzpicture}

\node[place,label={[name=p1Lab]left:\scriptsize $p_1$}](p1)at (0,0){};
\node[place,label={[name=p2Lab]above:\scriptsize $p_2$}](p2)at (0,1.25){};

\node[place,label={[name=p3Lab]above:\scriptsize $p_3$}](p3)at (1,1.25){};

\node[transvert] (t)at (1.5,0){};
\node at (t){\scriptsize $t$};

\node[place,label={[name=p4Lab]above:\scriptsize $p_5$}](p4)at (3,1.25){};
\node[place,label={[name=p5Lab]right:\scriptsize $p_6$}](p5)at (3,0){};

\draw [->,myDouble] (p1) -- node[below,midway,sloped] {\scriptsize $x_1$}(t);

\node [place,label={above:\scriptsize $p_4$}](p6)at(2,1.25){};
\draw[<->] (p6) -- node[midway,below,sloped]{\scriptsize $x_2$}(t);
\node at(p3){\scriptsize $a$};

\draw [->] (p2) -- node[below,midway,sloped] {\scriptsize $x_0$}(t);
\draw [<->] (p3) -- node[below,midway,sloped] {\scriptsize $x_1$}(t);

\draw [->,myDouble] (t) -- node[below,midway,sloped] {\scriptsize $x_2$}(p5);
\draw [->] (t) --node[below,midway,sloped] {\scriptsize $x_0$} (p4);

\node at($(p1)+(0,.1)$){\scriptsize a a};
\node at($(p1)-(0,.1)$){\scriptsize b c};
\node at($(p2)+(0,.1)$){\scriptsize b c};
\node at($(p2)-(0,.1)$){\scriptsize  c};
\node at($(p6)+(0,.1)$){\scriptsize c c};
\node at($(p6)-(0,.1)$){\scriptsize d};
    \end{tikzpicture}
    }
\caption{}
\label{fig:selectiveTransferPre}
    \end{subfigure}
    \hfill
       \begin{subfigure}[b]{.49\textwidth}\centering
       \scalebox{0.7}{
    \begin{tikzpicture}

\node[place,label={[name=p1Lab]left:\scriptsize $p_1$}](p1)at (0,0){};
\node[place,label={[name=p2Lab]above:\scriptsize $p_2$}](p2)at (0,1.25){};

\node[place,label={[name=p3Lab]above:\scriptsize $p_3$}](p3)at (1,1.25){};

\node [place,label={above:\scriptsize $p_4$}](p6)at(2,1.25){};
\node at(p3){\scriptsize $a$};
\node[transvert] (t)at (1.5,0){};
\node at (t){\scriptsize $t$};

\node[place,label={[name=p4Lab]above:\scriptsize $p_5$}](p4)at (3,1.25){};
\node[place,label={[name=p5Lab]right:\scriptsize $p_6$}](p5)at (3,0){};

\draw [->,myDouble] (p1) -- node[below,midway,sloped] {\scriptsize $x_1$}(t);
\draw [->] (p2) -- node[below,midway,sloped] {\scriptsize $x_0$}(t);
\draw [<->] (p3) -- node[below,midway,sloped] {\scriptsize $x_1$}(t);
\draw [<->] (p6) -- node[below,midway,sloped] {\scriptsize $x_2$}(t);

\draw [->,myDouble] (t) -- node[below,midway,sloped] {\scriptsize $x_2$}(p5);
\draw [->] (t) --node[below,midway,sloped] {\scriptsize $x_0$} (p4);

\node at (p5){\scriptsize d d};
\node at (p1){\scriptsize b c};
\node at(p2){\scriptsize b c};
\node at(p4){\scriptsize  c};
\node at($(p6)+(0,.1)$){\scriptsize c c };
\node at($(p6)-(0,.1)$){\scriptsize d};
    \end{tikzpicture}
    }
\caption{}
\label{fig:selectiveTransferPost}
    \end{subfigure}
    \caption{The firing of a special \rnPN transition. We have mapped $x_0$ to $c$, $x_1$ to $a$, and $x_2$ to $d$.}
    \label{fig:selectiveTransfer}
\end{figure}

\GnPN coverability can be reduced to \rnPN coverability:

special \GnPN transitions can be simulated by chains of \rnPN special transitions (see App.~\ref{app:GnPN-rnPN}). Consequently, in the following, we can show reductions regarding \GnPN or \rnPN coverability interchangeably.

\begin{example}
Fig.~\ref{fig:fragmentEquivalenceExample} depicts a \rnPN $N$ consisting of four \rnPN transitions $t_0,\dots,t_3$ that simulate the \GnPN $N'$ consisting of the single special \GnPN transition $t$ in Fig.~\ref{fig:CNuPN_example}. Fig.~\ref{fig:fragmentEquivalenceExample} depicts also the configurations as $t_0,t_1,t_2$ fire. $t_0$ simulates the preconditions. $t_1$ simulates the channel from $p_1$ to $p_3$. $t_2$ simulates the channel from $p_2$ to $p_4$. $t_3$ simulates the postconditions. When $t$ is currently not simulated, a token populates the place $\texttt{step}_0$. When $t_0$ fires, the token is moved to $\texttt{step}_1$, signaling that a simulation round is in progress. In order to perform the complete simulation of the firing of $t$, the transitions $t_0,\dots,t_3$ have to be fired in sequence. This is assured by the passage of a token among the $\texttt{step}_i$ places, for $i\in\{0,\dots,3\}$. When $t_3$ fires, the token is moved back to $\texttt{step}_0$, signaling that the simulation round has finished. The initial configuration of $N$ (Fig.~\ref{fig:fragmentEquivalenceExample1}) is that of $N'$ plus a fresh token on $\texttt{step}_0$. The coverability in the \GnPN in Fig.~\ref{fig:CNuPN_example} of a target $C$ is equivalent to the coverability in the \rnPN in Fig.~\ref{fig:fragmentEquivalenceExample} of the the extension of $C$ by the same fresh token on $\texttt{step}_0$.
\end{example}
\input{figures/fragmentEquivalenceExample}
\begin{restatable}{lemma}{mainGnPNrnPN}
     Given a \GnPN \(N\), there is a polynomial-time constructible \rnPN \(N'\) that simulates \(N\).
\end{restatable}

\section{Complexity Lower Bound}\label{Sec:lowerbound}
In this section, we show how to encode $\nu$PNs into cEOSs. This provides an $F_{\omega 2}$ lower bound for cEOSs, and in the next section, we will show that, with some modifications, we can also encode \GnPNs into cEOSs (giving us one direction of the inter-reducibility between channel $\nu$PNs and cEOSs).
The idea of our reduction is that we can encode each tuple of a configuration of a \nuPN $\D$ into a dedicated object net inside a \textit{simulator} place at the system net level.
 
We now define a cEOS $\os$ that can be used to check coverability in $\D$.  
 We start by defining the object nets. Actually, $\os$ involves only the type $\blacksquare$, used to fire sequences of events, and one more type $N_\D$ that captures the transitions of $\D$ split by variable as specified in the next definitions and example. Since we deal with only two types, in this section we graphically represent system net places of type $N_\D$ by circles, as usual, while we represent places of type $\blacksquare$ by triangles.

$N_\D$ includes all places of $\D$ and the restriction of each transition of $\D$ to each of its variables.

\begin{definition}
    $N_\D$ is the PN $N_\D=(P_\D,T_\D,F_\D)$ such that $P_\D=P$, $T_\D=\biguplus_{t\in T}\{t_x\mid x\in\var(t)\}$, and, for each $p\in P_\D$ and $t_x\in T_\D$, $F_\D(p,t_x)=F_x(p,t)$ and $F_\D(t_x,p)=F_x(t,p)$.
\end{definition}

The system net $\hat{N}=(\hat{P},\hat{T},\hat{F})$ of $\os$ contains the places $\text{sim}$, $\text{selectTran}$, and, for each $t\in T$, the places, transitions, and flow function depicted in Fig.~\ref{fig:lowerboundnuPNtoEOSblock}. The figure depicts also the synchronization structure $\Theta$ of $\os$. Note that the only transitions that can fire concurrently are the $t^\text{fire}_{x}$, for $x\in\var(t)$.

\begin{figure}[t]\centering
\resizebox{1\columnwidth}{!}{
\begin{tikzpicture}

\node[triangle,label={[name=selLab]above:\scriptsize $\text{selectTran}$}](sel)at(0,0){};
\node[transhor,label={[name=t1Lab]above:\scriptsize $t_{x_1}^{\text{select}}$}](t1)at(1.5,0){};
\node[triangle,label={[name=sel1Lab]above:\scriptsize $\text{select}_{x_2}^t$}](sel1)at(3,0){};
\node at(4.5,0)(selDots){$\dots$};
\node[triangle,label={[name=t2Lab]above:\scriptsize $\text{select}_{x_n}^t$}](t2)at(6,0){};
\node[transhor,label={[name=sel2Lab]above:\scriptsize $t_{x_n}^\text{select}$}](sel2)at(7.5,0){};

\node[triangle,label={[name=tnuLab]above:\scriptsize $\text{select}_{\nu}^t$}](tnu)at(9,0){};
\node[transhor,label={[name=selnuLab]above:\scriptsize $t_\nu^\text{select}$}](selnu)at(10.5,0){};

\node[place,label={[name=simLab]left:\scriptsize $\text{sim}$}](sim)at(-1.5,1.25){};

\node[place,label={[name=t1selLab]left:\scriptsize $t^\text{selected}_{x_1}$}](t1sel)at(1.5,-1){};
\node[triangle,label={[name=t1runLab]right:\scriptsize $t^\text{run}_{x_1}$}](t1run)at(2.5,-1){};

\node[transhor,label={[name=t1putLab]right:\scriptsize $t^\text{fire}_{x_1}\tup{t_{x_1}}$}] at (2,-2)(t1put){};

\node[place,label={[name=tnselLab]left:\scriptsize $t^\text{selected}_{x_n}$}](tnsel)at(7.5,-1){};
\node[triangle,label={[name=tnrunLab]right:\scriptsize $t^\text{run}_{x_n}$}](tnrun)at(8.5,-1){};

\node[transhor,label={[name=tnputLab]left:\scriptsize $t^\text{fire}_{x_n}\tup{t_{x_n}}$}] at (8,-2)(tnput){};

\node[triangle,label={[name=tnurunLab]left:\scriptsize $t^\text{run}_{\nu}$}](tnurun)at(10.5,-1){};

\node[transhor,label={[name=tnuputLab]left:\scriptsize $t^\text{fire}_{\nu}\tup{t_\nu}$}] (tnuput)at (10.5,-2){};

\draw[->] (sim.south east) -| (t1Lab.north);
\draw[->] (sim.east) -| (sel2Lab.north);

\draw[->] (t1.south) -- (t1sel.north);
\draw[->] (sel2.south) -- (tnsel.north);

\draw[->] ($(selnu.south)-(.2,0)$) |- ($(t1run)+(0,.75)$) -- (t1run);
\draw[->] ($(selnu.south)-(.1,0)$)|- ($(tnrun)+(0,.6)$) -- (tnrun);
\draw[->] ($(selnu.south)-(0,0)$) -- (tnurun);

\draw[->] (sel) -- (t1);
\draw[->] (t1) -- (sel1);
\draw[->] (sel1) -- (selDots);
\draw[->] (selDots)--(t2);
\draw[->] (t2)--(sel2);
\draw[->] (sel2)--(tnu);
\draw[->] (tnu)--(selnu);

\draw[->] (t1put) -| ($(sim.south)+(.2,0)$);
\draw[->] ($(tnput.south) - (.1,0)$) -- ($(tnput.south)-(.1,.15)$) -| (sim);
\draw[->] ($(tnuput.south) - (.1,0)$)-- ($(tnuput.south)-(.1,.3)$)-| ($(sim.south)-(.2,0)$);

    \node[triangle,label={[name=reportLab]above right:\scriptsize $t^\text{report}$}](report)at (2,-3.1){};

\node[transhor,label={[name=doneLab]left:\scriptsize $t^\text{done}$}](done)at(0,-3.1){};

\draw[->] (t1put) -- (report);
\draw[->] (tnput) |- (report.north east);
\draw[->] (tnuput) |- (report.east);
\draw[->] (report) --node [midway, above]{\scriptsize 
$n+1$} (done);
\draw[->] (done) -- (sel);

\node at (sel){\scriptsize $\blacksquare$};

\draw[->](t1sel) -- (t1put);
\draw[->](t1run) -- (t1put);

\draw[->](tnsel) -- (tnput);
\draw[->](tnrun) -- (tnput);

\draw[->](tnurun) -- (tnuput);

\end{tikzpicture}
}

    \caption{The part of $\os$ dedicated to the simulation of a transition $t$ of $\D$ such that $\var(t)=\{x_1,\dots,x_n,\nu\}$. If $\nu\notin\var(t)$, then, $\text{select}^t_\nu$, $t^\text{select}_\nu$, $t_\nu^\text{run}$, and $t_{\nu}^\text{fire}$ have to be dropped, and $\hat{F}(t^\text{select}_x,t_{x_1}^\text{run})=1$, $\hat{F}(t^\text{select}_x,t_{x_n}^\text{run})=1$, and $\hat{F}(t^\text{report},t^\text{done})=n$ has to be set.
    }
    \label{fig:lowerboundnuPNtoEOSblock}
\end{figure}
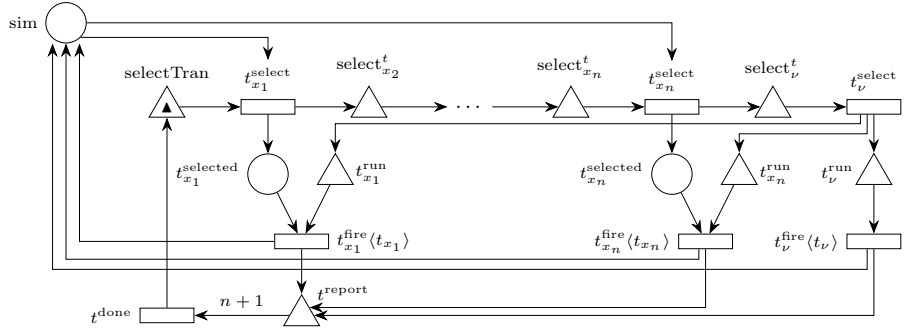

Intuitively, given a marking $\mu$ of $\D$, we capture $\mu$ by using, for each tuple $m\in\mu$ a dedicated $N_\D$ object in a place \textit{sim} with internal marking $m$. 
\begin{definition}\label{def:lowerBoundEncode}
    The configuration $M=\fmset{m_1,\dots,m_\ell}$ of $\D$ is encoded by the configuration $\overline{M}=\sum_{i=1}^\ell \tup{\text{sim}, m_i} + \tup{\text{selectTran}, \varepsilon}$.
\end{definition}

To simulate a transition $t$ of $\D$, we just need to, for each $x\in\var(t)\setminus\{\nu\}$,
\begin{inparaenum}[\itshape (1)]
    \item map $x$ to some object $O_x$ in \textit{sim} and then
    \item concurrently fire the transitions $t_x$ in each $O_x$.
\end{inparaenum} $\os$ performs this mapping by moving the non-deterministically chosen object $O_x$ from $\text{sim}$ to the place $t^{\text{selected}}_{x}$, for $x\in\var(t)\setminus\{\nu\}$. This selection happens sequentially as enforced by the array of places $\text{selectTran}$, $\text{select}^t_{x_i}$, and transitions $t^\text{select}_{x_i}$ and, if present, $\text{select}^t_{\nu}$ and $t^\text{select}_{\nu}$.

Only after firing the last $t^\text{select}_y$ transition,\footnote{That is $y=\nu$ if $\nu\in\var(t)$ and $y=x_n$ otherwise.} the transitions $t^\text{fire}_{x}$ get enabled and, since they are synchronized with $t_x$, they \begin{inparaenum}[\itshape (1)]
    \item move (or create, if $x=\nu$) the object $O_x$ from $t^\text{selected}_{x}$ back to $\text{sim}$,
    \item update the internal marking of $O_x$ according to $t_x$, and
    \item move a $\blacksquare$ token from $t^{\text{run}}_x$ to $t^\text{report}$.
\end{inparaenum}

After firing all the $t^\text{fire}_{x}$ transitions, the $t^\text{done}$ transition gets enabled. Its firing cleans $t^\text{report}$ and produces a single $\blacksquare$ token in $\text{sim}$, completing the simulation of $t$.

The next lemma follows from the previous description. Given some configuration $M$ and $M'$ of $\D$, we say that a run $\overline{M}\rightarrow^*\overline{M'}$ of $\os$ of positive length is \textit{minimal} if in each intermediate configuration the place $\text{selectTran}$ is not marked.
\begin{lemma}
    $M\rightarrow^t M'$ in $\D$ if only if $\overline{M}\rightarrow^{\sigma} \overline{M'}$ in $\os$ for some minimal run $\sigma$ of $\os$.
\end{lemma}
\begin{proof}
We assume $\var(t)=\{x_1,\dots,x_n,\nu\}$: in case $\nu\notin\var(t)$, the argument is analogous. 

If $M\rightarrow^t M'$, then there is a mode $e$ such that $M\rightarrow^{t,e}M'$. Thus, we can write $M=\fmset{m_1,\dots,m_n}+M''$ for some $m_1,\dots,m_n\in\mathbb{N}^\ell$ and configuration $M''$ of $\D$. By \nuPN semantics, $M'=M''+F_\nu(P,t)+\sum_{i=1}^n m_i - F_{x_i}(P,t) + F_{x_i}(t,P)$.
Thus, by Def.~\ref{def:lowerBoundEncode}, $\overline{M}=\overline{M''}+\sum_{i=1}^n \tup{\text{sim},m_i}$ and $\overline{M'}=\overline{M''}+\tup{\text{sim},F_\nu(P,t)}+\sum_{i=1}^n\tup{\text{sim},m_i - F_{x_i}(P,t) + F_{x_i}(t,P)}$.

The transition $t^{\text{select}}_{x_1}$ is enabled on $\overline{M}$. Moreover, we can fire the sequence of transitions $t^{\text{select}}_{x_1},\dots,t^{\text{select}}_{x_n}$ so as to select objects in $\text{sim}$ matching the mode $e$, i.e., by reaching a configuration $M''-\tup{\text{selectTran},\varepsilon}+\sum_{i=1}^n\tup{t^{\text{selected}}_{x_i},m_{e(x_{i})}}$. 

After firing $t^{\text{select}}_\nu$, for $x\in\var(t)$, the transitions $t^{\text{fire}}_{x}$ are enabled and, because of the synchronization structure, their firing results in the configuration
$M''-\tup{\text{selectTran},\varepsilon}+\sum_{i=1}^n\tup{\text{sim},m_{e(x_{i})} 
-F_{x_i}(P,t)+ F_{x_i}(t,P)} + \tup{\text{sim},F_{\nu}(t,P)} + (n+1)\tup{t^{\text{report}},\varepsilon}$. On this configuration, $t^{\text{done}}$ is enabled and its firing returns the configuration $\overline{M'}$. Thus, we have exhibited a run $\sigma$ such that $\overline{M}\rightarrow^{\sigma}\overline{M'}$. Moreover, this sequence is minimal.

Vice-versa, if there is a minimal run $\sigma$ of $\os$ such that $\overline{M}\rightarrow^\sigma\overline{M''}$, then we can organize $\sigma$ in three blocks: 
\begin{inparaenum}[\itshape (1)]
\item a prefix $\sigma'$ that amounts to the firing of the $t^\text{select}_x$ transitions,
\item an intermediate run $\sigma''$ that amounts to the firing of the transitions $t^{\text{fire}}_{x}$, and
\item a suffix $\sigma'''$ that amounts to the firing of $t^{\text{report}}$, for $x\in\var(t)$.
\end{inparaenum}

Moreover, the configuration reached at the end of $\sigma'$ is $\overline{M''}-\tup{\text{selectTran},\varepsilon} + \sum_{i=1}^n \tup{t^{\text{selected}}_{x_i},m_i} + \sum_{x\in\var(t)}\tup{t^\text{run}_{x_i},\varepsilon}$. Hence, $\overline{M}=\overline{M''} +\sum_{i=1}^n\tup{\text{sim},m_i}$. Moreover, by applying $\sigma$ to $\overline{M}$, we obtain that $\overline{M'}=\overline{M''}+\tup{\text{sim},F_\nu(t,P)}+\sum_{i=1}^n \tup{\text{sim}, m_i - F_{x_i}(P,t) + F_{x_i}(t,P)}$.

Consequently, $M=M''+\sum_{i=1}^n m_i $ and $M'=M''+ F_\nu(t,P) + \sum_{1=1}^n m_i - F_{x_i}(P,t)+F_{x_i}(t,P)$.
Thus, $t$ is enabled on $M$ with the mode $e$ such that $e(x_i)=m_i$, for each $i\in\{1,\dots,n\}$ and $M\rightarrow^{t,e}M'$.

\end{proof}

Consequently, the set of reachable configurations in $\os$ that encode some configuration in $\D$ essentially matches the set of configurations reachable in $\D$. Thus, we can check coverability in $\D$ of a configuration $\tau$ from an initial configuration $\iota$ by checking coverability in $\os$ of $\overline{\tau}$ from $\overline{\iota}$. Thus, we have a reduction from \nuPN coverability to cEOS coverability. Note that this reduction is polynomial, since the construction of $\os$ is polynomial. Since \nuPN coverability is $\mathbf{F}_{\omega2}$-complete we obtain the following theorem.

\begin{theorem}
    cEOS coverability is $\mathbf{F}_{\omega2}$-hard.
\end{theorem}

\section{From  \texorpdfstring{\rnPN}{rnuPN} to cEOS}\label{sec:toEOS}
\begin{figure}[t]
    \centering
    \begin{subfigure}[b]{.29\textwidth}\centering
    \scalebox{0.7}{
\begin{tikzpicture}
    \node[place,label=left:\scriptsize $p$] (p){}; 
    \node [transvert,label=right:\scriptsize $\text{check}(p)$]at(1.5,0)(check){};
    \node [transhor,label=right:\scriptsize $\text{rem}(p)$]at(0,-1.5)(rem){};
    \node [transhor,label=right:\scriptsize $\text{add}(p)$]at(0,1.5)(add){};
    \draw[->] (add) -- (p);
    \draw[<->] (check) -- (p);
    \draw[<-] (rem) -- (p);
    \node at(0,-2){};
\end{tikzpicture}
}
\caption{}
\label{fig:toEOSmoduleInternal}
\end{subfigure}
\hfill
    \begin{subfigure}[b]{.69\textwidth}\centering
    \scalebox{0.7}{
\begin{tikzpicture}
    \node[place,label={[name=ready1Lab]above:\scriptsize $t^\text{ready}_{x_1}$}]at(0,0)(ready1){};
\node [transhor,label={[name=tremLab]above:\scriptsize $t^\text{rem}_i\tup{\text{rem}(p_i)}$}] (trem)at (1.5,0){};
    \node[triangle,label=below:\scriptsize $p_i^s$]at(3,0)(pi){};
\node [transhor,label={[name=taddLab]above:\scriptsize $t^\text{add}_i\tup{\text{add}(p_j)}$}] at (4.5,0)(tadd){};
    \node[place,label=right:\scriptsize{$t^\text{copy}_{x_k}
    $}
    ]at(6,0)(ready2){};

\node[triangle,label={above left:\scriptsize $t^\text{run}_\text{tran}$}]at(3,1)(aux1){};
\node[triangle,label={left:\scriptsize $t^\text{running}_\text{tran}$}]at(1.5,-1)(aux2){};

\draw [<->] (ready1) -- (trem);
\draw [->] (trem) -- (pi);
\draw [->] (pi) -- (tadd);
\draw [->] (tadd) -- (ready2);
\draw [<->] (tadd) -- (ready2);

\draw [->](aux1.west) -| (tremLab.north);
\draw [<-](aux1.east) -| (taddLab.north);
\draw [<-](aux2.north) -- (trem.south);
\draw [->](aux2.east) -| (tadd.south);

\node[triangle,label=\scriptsize $t^\text{tran}_\text{done}$]at(6,2)(done){};
\node [transvert,label=\scriptsize $t^\text{stop}_\text{tran}$](stop)at(4.5,2){};

\draw[->](aux1) |- (stop);
\draw[->](stop) -- node[midway,above]{$2$} (done);
\draw[->](ready1Lab.north) |- (stop.north west);

\node[place,label=\scriptsize $\text{trash}$]at(6,1)(trash){};
\draw[<->] (stop) |-  (trash.north west);

\node at(trash)(token){$\bullet$};
\node [xshift=1.75cm,rectangle,draw,dashed]at (trash)(emptyObj){\scriptsize empty marking};
\draw [dashed] (token.center) -- (emptyObj);

\end{tikzpicture}
}
\caption{}
\label{fig:toEOSmoduleGadget}
    \end{subfigure}

\begin{subfigure}[b]{\textwidth}\centering
        \scalebox{0.7}{
\begin{tikzpicture}

\node[place,label={[name=simLab]right:\scriptsize $\text{sim}$}](sim)at(-1.5,0){};

\node[triangle,label={[name=selLab,xshift=-.2cm]above:\scriptsize $\text{selectTran}$}](sel)at(0,0){};

\node[transhor,label={[name=t0Lab,text width=2cm,align=center]above:\scriptsize $t_{x_0}^{\text{select}}$$\langle\text{rem}(p_2),$\\$\text{add}(p_4)\rangle$
}](t0)at(1.5,0){};
\node[triangle,label={[name=sel0Lab,xshift=.2cm]above:\scriptsize $t^\text{selected}_{x_0}$}](sel0)at(3,0){};

\begin{scope}[xshift=.5cm]
\node[transhor,label={[name=t1Lab,xshift=-.4]above:\scriptsize $t_{x_1}^{\text{select}}\tup{\text{check}(p_3)}$}](t1)at(4.5,0){};
\node[triangle,label={[name=sel1Lab,xshift=.2cm]above:\scriptsize $t^\text{selected}_{x_1}$}](sel1)at(6,0){};
\node[transhor,label={[text width=2cm, align=center, name=t2Lab]above:\scriptsize $t_{x_2}^{\text{select}}$$\langle\text{check}(p_4)$
}](t2)at(7.5,0){};
\node[triangle,label={[name=sel1Lab,xshift=.25cm]above:\scriptsize $t^\text{selected}_{x_2}$}](sel2)at(9,0){};
\node[transhor,label={[name=enableLab]above:\scriptsize $t^\text{enabling}$}](enable)at(10.5,0){};
\end{scope}

\begin{scope}
\node[place,label={[name=t0selLab]left:\scriptsize $t^\text{ready}_{x_0}$}](t0sel)at(1.5,-1){};
\node[triangle,label={[name=t0runLab]right:\scriptsize $t^\text{run}_{x_0}$}](t0run)at(2.5,-1){};
\end{scope}

\begin{scope}[xshift=.5cm]
\begin{scope}[xshift=3cm]
\node[place,label={[name=t1selLab]below:\scriptsize $t^\text{ready}_{x_1}$}](t1sel)at(1.5,-1){};

\node[place,label={[name=t1tranLab]left:\scriptsize $t^\text{copy}_{x_1}$}]at(7.5,-2)(t1tran){};

\node[transhor,label={[name=t1putLab]right:\scriptsize $t^\text{move}_{x_2}$}] at (2.5,-2.75)(t1put){};


\end{scope}

\begin{scope}[xshift=6cm]
\node[place,label={[name=t2selLab]below:\scriptsize $t^\text{copy}_{x_2}$}](t2sel)at(1.5,-1){};
\node[triangle,label={[name=t2runLab]right:\scriptsize $t^\text{run}_{\text{tran}}$}](t2run)at(2.5,-1){};

\node[triangle,label={[name=t2tranLab]above right:\scriptsize $t^\text{tran}_{\text{done}}$}]at(2.5,-2)(t2done){};

\node[transhor,label={[name=t2putLab]below right:\scriptsize $t^\text{move}_{x_1}$}] at (2.5,-3)(t2put){};

\draw [->] (t2done) -- (t2put);
\draw [->] (t2done) -| ($(t1put.north)+(0.2,0)$);
\draw [->] (t1tran) |- (t2put.east);
\end{scope}
\end{scope}

\node[transhor,label={[name=t0putLab,text width=.1cm]right:\scriptsize $t^\text{move}_{x_0}$}] at (2,-2)(t0put){};

\draw[->] (sim.north east) -- ($(sim.north east)+(0,1.2)$) -| (t0Lab.north);
\draw[->] (sim.north) -- ($(sim.north)+(0,1.2)$) -| (t1Lab.north);
\draw[->] (sim.north west) -- ($(sim.north west)+(0,1.4)$) -| (t2Lab.north);

\draw[->] (t1.south) -- (t1sel.north);
\draw[->] (t0.south) -- (t0sel.north);
\draw[->] (t2.south) -- (t2sel.north);

\draw[->] ($(enable.south)-(.2,0)$) |- ($(t0run)+(0,.75)$) -- (t0run);
\draw[->] ($(enable.south)-(.1,0)$) |- ($(t2run)+(0,.65)$) -- (t2run);

\draw[->] (sel) -- (t0);
\draw[->] (t0) -- (sel0);
\draw[->] (sel0) -- (t1);
\draw[->] (t1) -- (sel1);

\draw[->] (t0put) -| (sim.south east);

    \node[triangle,label={[name=reportLab]above right:\scriptsize $t^\text{report}$}](report)at (2,-3.75){};

\node[transhor,label={[name=doneLab]left:\scriptsize $t^\text{done}$}](done)at(0,-3.75){};

\draw[->] (t0put) -- (report);

\draw[->] (report) --node [midway, above]{\scriptsize 
$3$} (done);
\draw[->] (done) -- (sel);

\node at (sel){\scriptsize $\blacksquare$};

\draw [->](t2sel.west) -| (t1put.north);

\draw [->](t0run) -- (t0put);
\draw [->](t0sel) -- (t0put);

\draw [->](sel1) -- (t2);
\draw [->](t2) -- (sel2);
\draw [->](sel2) -- (enable);

\draw [->] (enable.south) -- (t1tran);

\draw [->] (t1put) -| (sim.south);
\draw [->] (t2put.south west) -| (sim.south west);

\draw[->] (t1put) |- (report.north east);
\draw[->] (t2put) |- (report.south east);

\node[draw,dashed,fit={(t1sel)(t1tran)(t2run)(t1selLab)(t1tranLab)(t2runLab)},xshift=0.15cm,](dashed){};
\node[rotate=-90]at($(dashed.east)+(.15,0)$){\scriptsize transfer gadgets};

    \end{tikzpicture}
    }
    \caption{}
    \label{fig:toEOSmoduleC}
    \label{fig:toEOS}
    
\end{subfigure}
    
    \caption{
    (\subref{fig:toEOSmoduleInternal}) Transitions $\text{add}(p)$,  $\text{check}(p)$, and $\text{rem}(p)$ assumed to be in the \nPN $N=(P,T,F)$, for each place $p\in P$.~(\subref{fig:toEOSmoduleGadget}) Gadget $\text{transfer}_{p_i}$, with initial marking, for each place $p_i$ in Fig.~\ref{fig:selectiveTransfer}: $j=6$ if $i=1$ and $j=i$ otherwise; $k=2$ if $i=1$ and $k=1$ otherwise.
    Note that all the places except $t^\text{rem}_i$, $p^s_i$, and $t^\text{add}_i$ are shared by all transfer gadgets. 
    ~(\subref{fig:toEOSmoduleC})
    The part of $\os$ dedicated to the simulation of the special transition $t$ in Fig.~\ref{fig:selectiveTransfer}. The dashed area represents
    the gadget $\text{transfer}_{p_i}$.
    Only interface places are shown. Circles and triangles denote places of type $N$ and $\blacksquare$.}
    \label{fig:toEOSmodule}
\end{figure}

We build on top of the construction in previous section and show how to simulate special transitions of a \rnPN $N'= (P,T,F,G)$. Since all special transitions have the same shape (up to the name of the involved places), it is sufficient to show the simulation of the special transition $t$ in Fig.~\ref{fig:selectiveTransfer}. 
We assume, without loss of generality, that the net $N'$ contains, for each $p\in P$, the standard transitions $\text{add}(p),\text{check}(p),\text{rem}(p)\in T$ depicted in Fig.~\ref{fig:toEOSmoduleInternal}.
We encode the arbitrary special \rnPN transition $t$ in Fig.~\ref{fig:selectiveTransfer} into a cEOSs $\os$ (depicted in Fig.~\ref{fig:toEOSmoduleGadget} and Fig.~\ref{fig:toEOS}) that employs the two object net types $N=(P,T,F)$ and $\blacksquare$. The synchronization structure contains no object autonomous event; it contains only the system autonomous and synchronization events depicted in Fig.~\ref{fig:toEOSmoduleGadget} and Fig.~\ref{fig:toEOS}. The most notable difference with the construction in~\cite{OurArxiv} is that, by using the gadgets $\text{transfer}_{p_i}$ in Fig.~\ref{fig:toEOSmodule}, we simulate $t$ up to lossiness, that is, the simulation of the special transition may return either the encoding of the configuration reached by firing $t$ or a sub-configuration w.r.t. to $\leq_f$.

Each gadget $\text{transfer}_{p_i}$ serializes a transfer from the place $p_i$ of an object in the place $t^\text{ready}_{x_1}$ (capturing the instantiation of variable $x_1$) to the respective place of an object in $t^\text{copy}_{x_1}$ or $t^\text{copy}_{x_2}$ (capturing the updated objects instantiating $x_1$ and $x_2$), according to the special transition semantics. We call these movements \textit{transfer steps}. When the gadget non-deterministically stops, it cleans $t^\text{ready}_{x_1}$ from unprocessed internal tokens by using a \textit{trash} place, which, in turn, witnesses the lossiness.
Note that lossiness does not impair the reduction for coverability.

The next lemmas show this behavior.
Lemma~\ref{lem:perfectStep} shows that the gadgets $\text{transfer}_{p_i}$ simulate perfect transfers as long as they do not modify the token in \textit{trash}.

\begin{restatable}{lemma}{perfectStep}\label{lem:perfectStep}
    Let $M=\fmset{m_1,m_2}$ be a marking for the \rnPN $N$ in Fig.~\ref{fig:toEOSmodule} such that $m_1\geq\tup{t^\text{ready}_{x_1},\delta_{p(1)}}$,
$\os'$ be the cEOS obtained by the union, for all places $p_i$ of $N$, of the gadgets $\text{transfer}_{p_i}$ in Fig.~\ref{fig:toEOSmoduleGadget},
and $M_0=\tup{t^\text{ready}_{x_1},m_1}+\tup{t^\text{copy}_{x_1},\varepsilon}+\tup{t^\text{copy}_{x_2},m_2}+\tup{\text{trash},m_3}+\tup{t^\text{run}_\text{tran},\varepsilon}$ for some marking $m_3$. 

The configuration $M_1=\tup{t^\text{copy}_{x_1},m'_1}+\tup{t^\text{copy}_{x_2},m'_2}+\tup{\text{trash},m_3}+2\tup{t^\text{tran}_\text{done},\varepsilon}$ is reachable in $\os'$ from $M_0$ if and only if $m_1'=m_1-m_1(p_1)\delta_{p_1}$ and $m_2'=m_2 + m_1(p_1)\delta_{p_6}$.\footnote{I.e., $m_1'$ and $m_2'$ are obtained by transferring the tokens in $p_1$ of $m_1$ to $p_6$ of $m_2$.}
\end{restatable}
\begin{proof}
We show that if a configuration $M_1$ of the form $\tup{t^\text{copy}_{x_1},m'_1}+\tup{t^\text{copy}_{x_2},m'_2}+\tup{\text{trash},m_3}+2\tup{t^\text{tran}_\text{done},\varepsilon}$ is reachable than $m'_1$ and $m'_2$ are as in the statement. The proof of the opposite direction is analogous.

It is easy to see by induction, on the run length, that the following hold along all runs of $\os$ starting from $M_0$. For $i\in\{1,\dots,6\}$:
\begin{enumerate*}
\item on $M_0$,
$t^\text{stop}_\text{tran}$ is enabled, while $t^\text{rem}_i$ is enabled if and only if the object in $t^\text{ready}_{x_1}$ has one token in $p_i$,\label{item:initialEnalbedToEos}
\item each run from $M_0$ reaching $M_1$ is of the form 
$M_0\rightarrow^{t^\text{rem}_{i_1}}
M_1'\rightarrow^{t^\text{add}_{i_1}}
M_1''\rightarrow^{t^\text{rem}_{i_2}}
\dots\rightarrow^{t^\text{add}_{i_{n}}}
M_n''\rightarrow^{t^{\text{stop}}_\text{tran}}M_1$
for some $n\in\mathbb{N}$,
\item each place of $\os'$ contains at most one token, except for $t^\text{tran}_\text{done}$ that is either empty or contains two $\blacksquare$ tokens,
\item the places $t^\text{copy}_{x_1}$ and $t^\text{copy}_{x_2}$ contain exactly one object, and
\item the place $t^\text{ready}_{x_1}$ contains exactly one object token as long as $t^\text{stop}_\text{tran}$ does not fire.
\end{enumerate*}

We define a \textit{transfer step} as the movement of either
\begin{enumerate*}
\item one token from $p_1$ in the object in $t^\text{ready}_{x_1}$ to $p_6$ in the object in $t^\text{copy}_{x_2}$ or 
\item for $i\neq 1$, one token from $p_i$ in the object in $t^\text{ready}_{x_1}$ to $p_i$ in the object in $t^\text{copy}_{x_1}$.
\end{enumerate*}
Each iteration of the cycle $M''_{k}\rightarrow ^{t^\text{rem}_{i_{k+1}}}M'_{k+1}\rightarrow ^{t^\text{add}_{i_{k+1}}}M''_{k+1}$ performs a transfer step. Each cycle starting with $t^{\text{rem}}_i$ can fire at most $m_1(p_i)$ times.

By cEOS semantics, since both $M_0$ and $M_1$ place the same object in \textit{trash}, with internal marking $m_3$, the configuration $M_1$ is reached only if $t^\text{stop}_\text{tran}$ fires when the object in $t^{\text{ready}}_{x_1}$ is empty. In this case, the cycles have fired the maximal number of times, that is, for each $i\in\{1,\dots,6\}$, the maximum number $m_1(p_i)$ of transfer steps has been fired. Thus $M_1=\tup{t^\text{copy}_{x_1},m'_1}+\tup{t^\text{copy}_{x_2},m'_2}+\tup{\text{trash},m_3}+2\tup{t^\text{tran}_\text{done},\varepsilon}$ is reachable. 

\end{proof}

Lemma~\ref{lem:imperfectStep} shows that, in case that the internal marking of the token in \textit{trash} gets modified, the gadgets $\text{transfer}_{p_i}$ simulate transfers up to some actual lossiness.

\begin{restatable}{lemma}{imperfectStep}\label{lem:imperfectStep}
    With the same conditions as in the previous lemma, the configuration $M_1=\tup{t^\text{copy}_{x_1},m'_1}+\tup{t^\text{copy}_{x_2},m'_2}+\tup{\text{trash},m'_3}+2\tup{t^\text{tran}_\text{done},\varepsilon}$ for some $m_3'>m_3$ is reachable in $\os'$ from $M_0$ if and only if $m_1'< m_1-m_1(p_1)\delta_{p_1}$ and $m_2\leq m_2'< m_2+m_1(p_1)\delta_{p_6}$.    
\end{restatable}
\begin{proof}
By inspecting the previous proof, since $m_3'>m_3$, the configuration $M_1$ can be reached only if $t^{\text{stop}}_\text{tran}$ fires when some transfer step is still possible that is, there is at least one $i\in\{1,\dots,6\}$ such that the transfer step as fired strictly less than $m_1(p_i)$ times. In this case, instead of 
$m_1'= m_1-m_1(p_1)\delta_{p_1}$ and 
$m_2' = m_2+m_1(p_1)\delta_{p_6}$, we obtain 
$m_1'< m_1-m_1(p_1)\delta_{p_1}$ and 
$m_2\leq m_2'< m_2+m_1(p_1)\delta_{p_6}$. Since the run modifies the object in $t^{\text{copy}}_{x_2}$ only by adding internal tokens and initially the object contains $m_2$, we obtain $m_2\leq m_2'$.
\end{proof}

 With these lemmas at hand, it is now easy to see that the runs of the cEOS in Fig.~\ref{fig:toEOS} simulate up to lossiness the firing of the transition in Fig.~\ref{fig:selectiveTransfer}. Specifically, given a \rnPN configuration $M=\fmset{m_1,\dots,m_\ell}$ and letting $M'=\sum_{i=1}^\ell\tup{\text{sim},m_i}$, we start from the initial configuration $M'+\tup{\text{selectTran},\varepsilon}+\tup{\text{trash},\varepsilon}$ that encodes $M$ into $\os$.
A first array of transitions and places select three objects and move them to the places $t^\text{ready}_{x_0}$, $t^\text{ready}_{x_1}$, and $t^\text{copy}_{x_2}$, respectively. At the same time, the synchronization structure on $t^{\text{select}}_{x_0}$, $t^{\text{select}}_{x_1}$, and $t^{\text{select}}_{x_2}$ apply the effect\footnote{I.e., to move one token from $p_2$ to $p_4$.} of the $x_0$ variable in Fig.~\ref{fig:selectiveTransferPre} as well as checking the non-emptiness of the place $p_3$ in the object to be moved to $t^\text{ready}_{x_1}$ and of the place $p_4$ in the object to be moved to $t^\text{ready}_{x_2}$. Thus, the selected objects encode a mode that enables the \rnPN transition $t$. 

A further handling of $\blacksquare$ tokens, by firing the transition $t^\textit{enabling}$, enables the transition $t^\text{move}_{x_0}$ and the \textit{transfer gadgets}. The former puts the first selected, now updated, object back to \textit{sim}. The latter, by the previous lemmas, perform lossy transfers. Overall, after firing $t^\text{move}_{x_1}$ and $t^\text{move}_{x_2}$, three objects are moved back to \textit{sim}, encoding with lossiness the effect of the transfers required by $x_1$ and $x_2$ in Fig.~\ref{fig:selectiveTransferPre}. A final transition $t^\text{done}$ rearranges the $\blacksquare$ objects so as to make the net ready for the simulation of a new \rnPN transition firing. Guided by this insight, we obtain the following theorem.

\begin{theorem}
    There is a polynomial reduction from \rnPN coverability to cEOS coverability.
\end{theorem}

\begin{proof}
Given a \rnPN coverability instance over the \rnPN $N'$, with initial configuration $M_0=\fmset{m_1,\dots,m_h}$, and target configuration $M_1=\fmset{m_1',\dots,m_k'}$, we can build a cEOS $\os$ by applying the construction in~\cite{OurArxiv} to each standard transition of $N'$ and the construction in Fig.~\ref{fig:toEOS} to each special transition of $N'$. We can also encode $M_0$ into the configuration $\overline{M_0}=\sum_{i=1}^h\tup{\text{sim},m_i}+\tup{\text{selectTran},\varepsilon}+\tup{\text{trash},\varepsilon}$ of $\os$. 

We say that a finite run of $\os$ from $\overline{M_0}$ is \textit{perfect} if the last configuration of the run  places a single $\blacksquare$ object on \textit{selectTran} and a single, empty object on \textit{trash}. Instead, we call a finite run of $\os$ from $\overline{M_0}$ \textit{imperfect} if the last configuration of the run places a single $\blacksquare$ object on \textit{selectTran} and a single, but non-empty object on \textit{trash}.

Let $\overline{M_1}=\sum_{i=1}^k\tup{\text{sim},m_i'}+\tup{\text{selectTran},\varepsilon}+\tup{\text{trash},\varepsilon}$. By Lemma~$1$ in~\cite{OurArxiv} and Lemma~\ref{lem:perfectStep}, $N$ covers $M_1$ from $M_0$ if and only if there is a perfect run of $\os$ that covers $\overline{M_1}$ from $\overline{M_0}$. 

However, if a run of $\os$ covers $\overline{M_1}$ from $\overline{M_0}$, then the run is either perfect or imperfect.
Nevertheless, if the run is imperfect, then there is also a perfect run of $\os$ that covers $\overline{M_1}$. In fact, we can inductively (on the simulation steps):
\begin{enumerate*}
\item substitute the sub-run that simulates the special transition, with a sub-run that executes all possible transfer steps and
\item since cEOSs are WSTSs~\cite{kohler-busmeier_survey_2014}, extends this new sub-run by firing the remaining transition in the complete run. 
\end{enumerate*}
Overall, we still obtain a marking $\overline{M_1}'\geq \overline{M_1}$. Summarizing, $N$ covers $M_1$ from $M_0$ if and only if $\os$ covers $\overline{M_1}$ from $\overline{M_0}$.
\end{proof}

\section{From cEOS to \texorpdfstring{\GnPN}{c-nuPN}}\label{sec:fromEOS}
In this section, we show that for any given EOS \( \os= (\hat{N},\N,d,\Theta)\), there is a polynomial-time constructible \GnPN \(W=(P,T,F,G)\) that simulates \(\os\). We assume without loss of generality that, for each transition \(\hat{t} \in \hat{T}\), \(\hat{t}\) is involved in exactly one event.\footnote{If \(\hat{t}\) is involved in multiple events $(\theta_i)_{i=1}^n$, we substitute \(t\) in \(\theta_i\) with a copy \(t_i\) of \(t\), for each $i\in \{1,\cdots,n\}$.}
In what follows, we show only how to encode an arbitrary synchronous event $e=\tup{\tau,\theta}$ according to the semantics in Rem.~\ref{rem:EOSsemantics}, i.e., by serializing merging, internal firing, and distribution phases. In fact, object and system autonomous events are a restricted forms of synchronous events. We assume that $\tau$ involves only one system place type $N\in\N$.\footnote{The construction for events involving several types is analogous, requiring only the serialization of a copy of the given construction for each type.} The \GnPN $W$ combines several modules:
\begin{enumerate*}
    \item \textbf{\(\hat{p}\)-Block} modules, used to encode the configurations of $\os$;
    \item \textbf{$N$-merged} and \textbf{$N$-updated} modules, auxiliary places to store the result of the merging phase and conduct the internal firing phase of the event $e$;
    \item \textbf{$e$-merging} modules, to capture the dynamics of the merging stage;
    \item \textbf{$e$-firing module}, to capture the dynamics of the internal firing stage;
    \item \textbf{$e$-distributing module}, to capture the dynamics of the distribution stage.
\end{enumerate*} Overall these modules can be built in polynomial time and their consecutive firing allows $W$ to simulate the firings of $e$ in $\os$. We enforce such consecutive firings by sequentializing the modules using the variable \(x_0\) and a set of places. Thus, we have the next Theorem, for which we provide a proof sketch by describing the modules above, through the example in Fig.~\ref{fig:ceosex}.
After which we will formally state the encoding of an EOS in \cref{subsec:encoding} and the proof of the construction in \cref{subsec:simulation}.
\begin{theorem}
    Given a cEOS \os = \((\hat{N},\N,d,\Theta)\), there is a polynomial time constructible \GnPN \(W = (P,T,F,G)\) that simulates \os.
\end{theorem}

\begin{figure}[t]
    \begin{subfigure}{0.45\linewidth}
        \centering
        \scalebox{0.8}{

\begin{tikzpicture}
    \begin{scope}
    \node[place,label={[name=p1obj1Lab]left:\scriptsize$p_1$}]at(0,0)(p1obj1){};
    \node at(p1obj1) {\scriptsize$\bullet\bullet$};
    \node[transhor,label=left:\scriptsize$t$]at(0,-1)(tobj1){};
    \node[place,label=left:\scriptsize$p_2$]at(0,-2)(p2obj1){};
    \node at(p2obj1) {\scriptsize$\bullet$};
    \draw[->] (p1obj1) -- (tobj1);
    \draw[<-] (p2obj1) -- (tobj1);

    \node[draw,dashed,fit={(p1obj1)(tobj1)(p2obj1)(p1obj1Lab)}](obj1){};
    \end{scope}
   \begin{scope}[xshift=1.7cm]
    \node[place,label={ left:\scriptsize$\hat{p}_1$}]at(0,0)(p1sys){};
    \node[xshift=-.1cm] (bul1)at(p1sys) {\scriptsize$\bullet$};
        \node[xshift=.1cm] (bul2)at(p1sys) {\scriptsize$\bullet$};
    \node[transhor,label=left:\scriptsize$\hat{t}\tup{2t}$]at(0,-1)(tsys){};
    \node[place,label=left:\scriptsize$\hat{p}_2$]at(0,-2)(p2sys){};

    \draw[->] (p1sys) --node[midway,left]{\scriptsize$2$} (tsys);
    \draw[<-] (p2sys) --node[midway,left]{\scriptsize$2$} (tsys);

    \end{scope}

   \begin{scope}[xshift=2*1.7cm]
    \node[place,label={[name=p1obj2lab]left:\scriptsize$\hat{p}_1$}]at(0,0)(p1obj2){};
    \node at(p1obj2) {\scriptsize$\bullet$};
    \node[transhor,label=left:\scriptsize$t$]at(0,-1)(tobj2){};
    \node[place,label=left:\scriptsize$\hat{p}_2$]at(0,-2)(p2obj2){};
    \node at(p2obj2) {\scriptsize$\bullet\bullet$};
    \draw[->] (p1obj2) -- (tobj2);
    \draw[<-] (p2obj2) -- (tobj2);

    \node[draw,dashed,fit={(p1obj2)(tobj2)(p2obj2)(p1obj2lab)}](obj2){};
    \end{scope}

    \draw[dashed] (obj1)--(bul1.center);
    \draw[dashed] (obj2)--(bul2.center);
\end{tikzpicture}
}
      \caption{A cEOS Example.}
      \label{fig:ceosex}
\end{subfigure}
\begin{subfigure}{0.45\linewidth}
        \centering
\scalebox{0.5}{\begin{tikzpicture}

\path [rectangle, fill=gray!10](2.5,-1.7) to (3.5,-1.7) to (3.5,1.5) to (2.5,1.5);
\node[label={[name=p1BlockLab]above:\scriptsize  $\hat{p}_1$-block}](p1Block)at(3,1.3){};
\node[place,label={[name=idpLab]below:\scriptsize $Id_{\hat{p}}$}](idp)at (3,0){};
\node[empty] (txt) at (3,0.1){\scriptsize \tiny $a_1$};
\node[empty] (txt) at (3,-0.1){\scriptsize \tiny $a_2$};
\node[place,label={[name=p1Lab]below:\scriptsize $p_1$}](p1)at (3,1){};
\node[empty] (txt) at (3,1.1){\scriptsize \tiny $a_1 a_1$};
\node[empty] (txt) at (3,0.9){\scriptsize \tiny $a_2$};
\node[place,label={[name=p2Lab]below:\scriptsize $p_2$}](p2)at (3,-1){};
\node[empty] (txt) at (3,-0.9){\scriptsize \tiny $a_1$};
\node[empty] (txt) at (3,-1.1){\scriptsize \tiny $a_2a_2$};

 \path [rectangle, fill=gray!10](5,-1.7) to (6,-1.7) to (6,1.5) to (5,1.5);

\node[label={[name=idpnBlockLab]above: \scriptsize $\hat{p}_2$\text{-block}}](pNBlock)at(5.5,1.3){};
\node[place,label={[name=idpnLab]below:\scriptsize $\Id_N^m$}](idpn)at (5.5,0){};

\node[place,label={[name=p1NLab]below:\scriptsize $p_1$}](p1N)at (5.5,1){};
\node[place,label={[name=p2NLab]below:\scriptsize $p_2$}](p2N)at (5.5,-1){};

\end{tikzpicture}
}     
      \caption{Config in Fig. \ref{fig:ceosex} captured by W.}
      \label{fig:flatten}
  \end{subfigure}
\caption{A cEOS and the encoding of its marking in \GnPN.}
\end{figure}
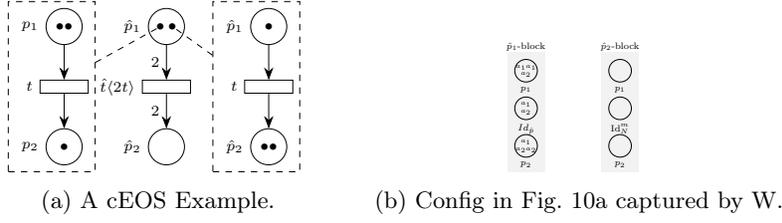

\noindent\textbf{{Module \textbf{\(\hat{p}\)-Block}}.}
For each $\hat{p}\in \hat{P}$, the module \textbf{\(\hat{p}\)-Block},  is used to encode in $W$ the objects hosted by $\hat{p}$. Specifically, \textbf{\(\hat{p}\)-Block} contains (a disjoint copy of) the places of $N$ as well as a place $\Id_{\hat{p}}$. 
The latter is used to store identifiers dedicated to each encoded object. 
A configuration $\tup{\hat{p},m}$ of $\os$ is captured by the tuple $m'$ of \textbf{\(\hat{p}\)-Block} that extends $m$ with a token on $\Id_{\hat{p}}$.
Note that, thanks to the place $\Id_{\hat{p}}$, this encoding can even witness the case $m=0$, by placing a single token on $\Id_{\hat{p}}$.

\begin{definition}
    The encoding $\hat{M}$ of a configuration $M$ of $\os$ is the configuration $\hat{M}=\K(M)+\fmset{\delta_{p^{init}}}$ of $W$ where \(\K(M)=\sum_{i=1}^\ell \fmset{\K(c_i)}\) and \(\K(c = \tup{\hat{p},m})=\sum_{p\in P_N} m(p)\delta_{p^{\hat{p}}}+\delta_{\Id_{\hat{p}}}\).
\end{definition}

\smallskip
\noindent\textbf{Modules \textbf{$N$-merged} and \textbf{$N$-updated}.}
These blocks are used to store the objects after the merging and internal firing stages and are analogous to the blocks above. 
The \textbf{$N$-merged} module and the \textbf{$N$-updated} modules contain a disjoint copy of the places of $N$
as well as the place $\Id_{N}^m$ and $\Id_{N}^u$, respectively. 
The encoding of a merged/updated tuple is similar to that in \textbf{$\hat{p}$-Block}.

\smallskip
\begin{wrapfigure}{l}{4.1cm}
        \centering
        \vspace{-20pt} 
\scalebox{0.7}{\begin{tikzpicture}

\node[place,label={[name=objenbLab]below:\scriptsize $p^{init}$}](pmerge)at (5.5,-1.25){};
\node[place,label={[name=objenbLab]above:\scriptsize $p^{merged}_{e}$}](pmerged)at (5.5,1.25){};
\path [rectangle, fill=gray!10](2.5,-1.7) to (3.5,-1.7) to (3.5,1.5) to (2.5,1.5);
\node[label={[name=p1BlockLab]above:\scriptsize  $\hat{p}_1$-block}](p1Block)at(3,1.3){};
\node[place,label={[name=idpLab]below:\scriptsize $Id_{\hat{p}}$}](idp)at (3,0){};
\node[empty] (txt) at (3,0.1){\scriptsize \tiny $a_1$};
\node[empty] (txt) at (3,-0.1){\scriptsize \tiny $a_2$};
\node[place,label={[name=p1Lab]below:\scriptsize $p_1$}](p1)at (3,1){};
\node[empty] (txt) at (3,1.1){\scriptsize \tiny $a_1 a_1$};
\node[empty] (txt) at (3,0.9){\scriptsize \tiny $a_2$};
\node[place,label={[name=p2Lab]below:\scriptsize $p_2$}](p2)at (3,-1){};
\node[empty] (txt) at (3,-0.9){\scriptsize \tiny $a_1$};
\node[empty] (txt) at (3,-1.1){\scriptsize \tiny $a_2a_2$};

\node[transvert,minimum height=12mm,minimum width=5mm] (tmerge)at (5.5,0){};
\node [rotate=-90]at(5.5,0){\tiny$\tau^{merge}_e$};

 \path [rectangle, fill=gray!10](7,-1.7) to (8,-1.7) to (8,1.5) to (7,1.5);

\node[label={[name=idpnBlockLab]above:\tiny N\text{-Merged}}](pNBlock)at(7.5,1.3){};
\node[place,label={[name=idpnLab]below:\scriptsize $\Id_N^m$}](idpn)at (7.5,0){};

\node[place,label={[name=p1NLab]below:\scriptsize $p_1$}](p1N)at (7.5,1){};
\node[place,label={[name=p2NLab]below:\scriptsize $p_2$}](p2N)at (7.5,-1){};

\draw[->](idp)--node[above,near start]{\scriptsize $x_1^{\hat{p}_1},x_2^{\hat{p}_1}$}(tmerge);
\draw[->,myDouble] (p1) -- node[above,midway,sloped] {\scriptsize $x_1^{\hat{p}_1},x_2^{\hat{p}_1}$}(tmerge);
\draw[->,triple] (p2) -- node[below,midway,sloped] {\scriptsize $x_1^{\hat{p}_1},x_2^{\hat{p}_1}$}(tmerge);

\draw[->](tmerge) -- node[right,midway] {\scriptsize $x_{cs}$}(pmerged);

\draw[->,myDouble] (tmerge) -- node[above,midway,sloped] {\scriptsize $x_1^{\hat{p}_1},x_1^{\hat{p}_1}$}(p1N);
\draw[->] (tmerge) to (idpn);
\node[empty] (txt) at (7,0.2){\scriptsize $x_1^{\hat{p}_1}$};
\draw[->,triple] (tmerge) -- node[below,midway,sloped] {\scriptsize $x_1^{\hat{p}_1},x_1^{\hat{p}_1}$}(p2N);

\draw[<-](tmerge.south) -- node[midway,right]{\scriptsize $x_{cs}$}(pmerge.north);

\end{tikzpicture}
}   \vspace{-20pt}   
      \caption{Module $e$-merging.}
      \label{fig:emerge}
      \vspace{-20pt}
\end{wrapfigure}
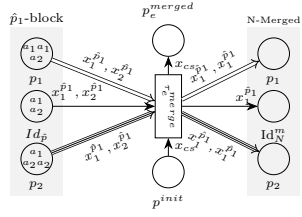
\noindent\textbf{{Module \textbf{\(e\)-merging}}.}
This module captures the dynamics of the merging phase. It contains the places $p^{init}$ and $p^{merged}_e$ as well as the transition $\tau^{merge}_e$. The two places behave as enable and acknowledgment for the transition. Thus, initially, $p^{init}$ is populated by a dedicated \textit{control sequence} token (moved around by the variable $x_{cs}$).
In \(\os\), for each place $\hat{p}$ from which $\tau$ consumes, there is a channel from each non-$\Id$ place of \textbf{$\hat{p}$-Block} to the corresponding place in \textbf{$N$-merged}. 
The pre-channels from $\hat{p}$ are labeled by $x_1^{\hat{p}},\dots,x_n^{\hat{p}}$, where $n$ is the the number of objects consumed by $\tau$ from $\hat{p}$ in \(\os\). The corresponding post-channel is labeled by the constant sequence $x_1^{\hat{p}},\dots,x_1^{\hat{p}}$ of length $n$. 
Moreover, to take care of the identifiers of the encoded objects consumed by \(\tau\), $\tau^{merge}_e$ consumes from $\Id_{\hat{p}}$ one token for each $x_1^{\hat{p}},\dots,x_n^{\hat{p}}$ and produces only one token $x_1^p$ in $\Id_N^m$, where $p$ is the minimum of the places from which $\tau$ consumes according to some fixed order among the system net places.

\smallskip
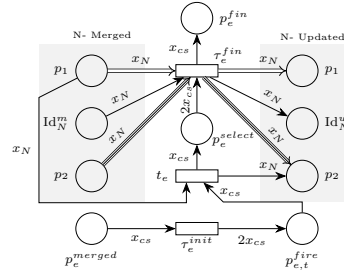
\begin{wrapfigure}{r}{5cm}
        \centering
        \vspace{-25pt}
\scalebox{0.7} {
\begin{tikzpicture}
\path [rectangle, fill=gray!10](8,-1.5) to (10,-1.5) to (10,1.5) to (8,1.5);
\node[label={[name=idpnBlockLab]below:\scriptsize \tiny N\text{- Merged}}](pNBlock)at(9.2,2){};
\node[place,label={[name=idpnLab]left:\scriptsize $\Id_N^m$}](idpn)at (9,0){};
\node[place,label={[name=p1NLab]left:\scriptsize $p_1$}](p1N)at (9,1){};
\node[place,label={[name=p2NLab]left:\scriptsize $p_2$}](p2N)at (9,-1){};

\node[place,label={[name=Lab]right:\scriptsize $p_e^{fin}$}](pfin)at (11,2){};

\node[transhor] (tfin)at (11,1){};
\node[empty] (txt) at (11.6,1.3){\scriptsize$\tau_{e}^{fin}$};

\node[place](psel)at (11,-0.1){};
\node[empty] (txt) at (11.7,-0.3){\scriptsize$p_{e}^{select}$};
\node[transhor,label={[name=tNLab]left:\scriptsize  $t_e$}] (t)at (11,-1){};
\node[place,label={[name=Lab]below:\scriptsize $p_e^{merged}$}](pinit)at (9,-2){};
\draw[->](t)--node[left,midway] {\scriptsize $ x_{cs}$}(psel);
\draw[->](psel)--(tfin);
\node[empty,rotate=90](txt) at (10.8,0.45){\scriptsize $2x_{cs}$};
\draw[->](tfin)--node[left,midway] {\scriptsize $x_{cs}$}(pfin);
\draw[->,myDouble] (p1N) -- node[above,midway] {\scriptsize $x_{N}$}(tfin);
\draw[->,triple] (p2N) -- node[above,near start,sloped] {\scriptsize $x_{N}$}(tfin);

\path [rectangle, fill=gray!10](12.2,-1.5) to (14,-1.5) to (14,1.5) to (12.2,1.5);
\node[label={[name=BlockLab]above:\scriptsize \tiny N\text{- Updated}}](NUBlock)at(13.2,1.3){};
\node[place,label={[name=idpnLab]right:\scriptsize $\Id_N^u$}](idpu)at (13,0){};
\node[place,label={[name=Lab]right:\scriptsize $p_1$}](p1u)at (13,1){};
\node[place,label={[name=Lab]right:\scriptsize $p_2$}](p2u)at (13,-1){};
\draw[->,myDouble] (tfin) -- node[above,near end] {\scriptsize $x_{N}$}(p1u);
\draw[->,triple] (tfin) -- node[above,near end,sloped] {\scriptsize $x_{N}$}(p2u);

\draw[->](idpn)--node[above,near start,sloped] {\scriptsize $x_{N}$}(tfin);
\draw[->](tfin)--node[above,near end,sloped] {\scriptsize $x_{N}$}(idpu);
\draw[->](p1N) --($(p1N.south)+(-1,-0.2)$) -- node[left,pos=0.5] {\scriptsize $x_{N}$}($(p1N.south)+(-1,-2.2)$) -| ($(t.south)-(0.2,0)$);
\draw[->](t)--node[above,near end] {\scriptsize $x_{N}$}(p2u);
\node[transhor] (tinit)at (11,-2){};
\node[empty] (txt) at (11,-2.3){\scriptsize$\tau_{e}^{init}$};

\node[place,label={[name=Lab]below:\scriptsize $p_{e,t}^{fire}$}](pfire)at (13,-2){};
\draw[->](pinit)--node[below]{\scriptsize $x_{cs}$}(tinit);
\draw[->](tinit)--node[below]{\scriptsize $2 x_{cs}$}(pfire);
\draw[->](pfire)--($(pfire.north)+(0,0.2)$)--($(pfire.north)+(-1.5,0.2)$)--node[right]{\scriptsize$x_{cs}$}(t);
\end{tikzpicture}
}   \vspace{-20pt}   
      \caption{Module $e$-updating.}
      \label{fig:eupdating}
      \vspace{-20pt}
\end{wrapfigure}
\noindent\textbf{Module \textbf{\(e\)-updating}.}
This module captures the dynamics of the internal firing phase applied to the \textbf{$N$-merged} modules after the merging phase, storing its result in the module \textbf{$N$-updated}. It contains the places $p_e^{select}$ as well as a set $\{p^{x}_{e,t} \mid x\in\{fire,fired\}, t \in \supp(\theta)\}$ of new places. The modules contains the transitions $\tau^{init}_e$ and $\tau^{fin}_e$ as well as a copy $t_e$ of the transitions $t\in\supp(\theta)$.
Places $p^{init}_e$, $p^{select}_e$, and $p^{fin}_e$ are used to fire $\tau^{init}_e$ and $\tau^{fin}_e$ in sequence.
The places $p^\text{fire}_{e,t}$ and $p^\text{fired}_{e,t}$ behave as enable to fire and acknowledgment of firing of $t$, respectively.
For each $t\in\supp(\theta)$, $\tau^{init}_e$ puts $\theta(t)$ tokens in $p^{fire}_{e,t}$. Thus, the transition $t_e$ in the module gets enabled $\theta(t)$ times. When firing, $t_e$ behaves as $t$, with the provision that its preconditions stem from places in \textbf{$N$-merged}, while its post-conditions are directed towards places in \textbf{$N$-updated}. The movement of tokens from \textbf{$N$-merged} to \textbf{$N$-updated} ensures that either the transitions in $\theta$ are all enabled at the start of this phase, or a deadlock is reached, i.e., the firing of a $t'\in\theta$ is not responsible for the enabling of another transition $t''\in\theta$.
Finally, all tokens remaining in \textbf{$N$-merged} are moved to the respective places in \textbf{$N$-updated} using the channels of $\tau^{fin}_e$, as well as the movement of the token from $\Id_{N}^m$ to $\Id_{N}^u$, for each $\hat{p}$ involved in the preconditions of $\tau$.

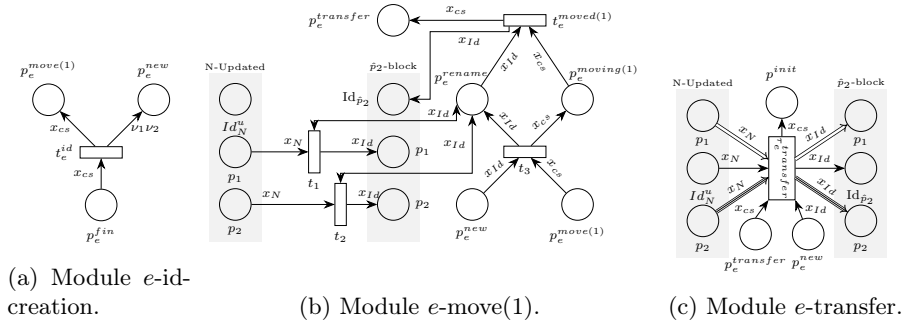
\begin{figure}[t]
    \begin{subfigure}{0.2\linewidth}
        \centering
        \scalebox{0.7}{
\begin{tikzpicture}

\node[place,label=below:\scriptsize$p^{fin}_e$]at(0,-2)(pfin){};
\node[transhor,label=left:\scriptsize$t_e^{id}$]at(0,-1)(t){};
\node[place,label=above:\scriptsize$p^{move(1)}_e$]at(-1,0)(pmove){};
\node[place,label=above:\scriptsize$p^{new}_e$]at(1,0)(pnew){};

\draw[->] (pfin)--node[midway,left]{\scriptsize$x_{cs}$} (t);
\draw[->] (t)-- node[midway,right]{\scriptsize$\nu_1\nu_2$}(pnew);
\draw[->] (t)--node[midway,left]{\scriptsize$x_{cs}$} (pmove);
    
\end{tikzpicture}
} 
      \caption{Module $e$-id-creation.}
      \label{fig:eidcreation}
    \end{subfigure}
    \begin{subfigure}{0.48\linewidth}
        \centering
\scalebox{0.7}{\begin{tikzpicture}

\path [rectangle, fill=gray!10](3.5,-1.7) to (4.5,-1.7) to (4.5,1.5) to (3.5,1.5);
\node[label={[name=BlockLab]above: \tiny N-Updated}](NBlock)at(4,1.3){};
\node[place,label={[name=idpLab]below:\scriptsize $Id_{N}^u$}](idN)at (4,1){};

\node[place,label={[name=p1Lab]below:\scriptsize $p_1$}](p1N)at (4,0){};
\node[place,label={[name=p2Lab]below:\scriptsize $p_2$}](p2N)at (4,-1){};

\node[transvert,label={[name=tLab]below:\scriptsize  $t_1$}](t1)at(5.5,0){};
\node[transvert,label={[name=tLab]below:\scriptsize  $t_2$}](t2)at(6,-1){};
 \path [rectangle, fill=gray!10](6.5,-1.7) to (7.5,-1.7) to (7.5,1.5) to (6.5,1.5);

\node[label={[name=pBlockLab]above:\tiny$\hat{p}_2$-block}](pBlock)at(7,1.3){};
\node[place,label={[name=idpnLab]left:\scriptsize $\Id_{\hat{p}_2}$}](idp)at (7,1){};

\node[place,label={[name=p1NLab]right:\scriptsize $p_1$}](p1p)at (7,0){};
\node[place,label={[name=p2NLab]right:\scriptsize $p_2$}](p2p)at (7,-1){};

\draw[->](p2N)--node[above, near start]{\scriptsize $x_N$}(t2);
\draw[->](t2)--node[above,near end]{\scriptsize $x_{Id}$}(p2p);
\draw[->](p1N)--node[above,near end]{\scriptsize $x_N$}(t1);
\draw[->](t1)--node[above,near end]{\scriptsize $x_{Id}$}(p1p);


\node[place](prename)at (8.5,1){};
\node[empty](txt) at (8.3,1.4){\scriptsize $p^{rename}_{e}$};
\node[place,label={[name=Lab]below:\scriptsize $p^{new}_{e}$}](pnew)at (8.5,-1){};

\node[transhor,label={[name=tLab]below:\scriptsize  $t_3$}](t3)at(9.5,0){};

\node[place](pmoving)at (10.5,1){};
\node[empty](txt) at (11,1.5){\scriptsize $p^{moving(1)}_{e}$};
\node[place,label={[name=Lab]below:\scriptsize $p^{move(1)}_{e}$}](pmove)at (10.5,-1){};
\node[transhor,label={[name=tLab]right:\scriptsize  $t^{moved(1)}_e$}](tmoved)at(9.5,2.5){};

\node[place,label={[name=Lab]left:\scriptsize $p^{transfer}_{e}$}](ptrans)at (7,2.5){};

\draw[->](pnew)--node[above,midway,sloped]{\scriptsize $x_{Id}$}(t3);
\draw[->](t3)--node[above,midway,sloped]{\scriptsize $x_{Id}$}(prename);
\draw[->](prename)--node[below,midway,sloped]{\scriptsize $x_{Id}$}(tmoved);

\draw[->](pmove)--node[above,midway,sloped]{\scriptsize $x_{cs}$}(t3);
\draw[->](t3)--node[above,midway,sloped]{\scriptsize $x_{cs}$}(pmoving);
\draw[->](pmoving)--node[below,midway,sloped]{\scriptsize $x_{cs}$}(tmoved);

\draw[<->]($(prename.south)+(-0.3,0.2)$)--node[left]{\scriptsize $x_{Id}$}($(prename.south)-(0.3,0.1)$)-|(t1.north);

\draw[<->](prename.south)--node[left]{\scriptsize $x_{Id}$}($(prename.south)-(0,1.1)$)-|(t2.north);

\draw[->]($(tmoved.south)-(0.3,0)$)|-node[below,near end]{\scriptsize $x_{Id}$}($(tmoved.south)-(1.8,0.1)$)|-($(idp)+(0.5,0)$)--(idp.east);
\draw[->](tmoved)--node[above,midway,sloped]{\scriptsize $x_{cs}$}(ptrans);
\end{tikzpicture}
}
        \caption{Module $e$-move(1).}
        \label{fig:emove}
    \end{subfigure}  
    \begin{subfigure}{0.3\linewidth}
        \centering
\scalebox{0.7}{\begin{tikzpicture}

\node[place,label={[name=objenbLab]below:\scriptsize $p^{transfer}_e$}](pmerge)at (5,-1.25){};
\node[place,label={[name=newLab]below:\scriptsize $p^{new}_e$}](pnew)at (6,-1.25){};
\node[place,label={[name=objenbLab]above:\scriptsize $p^{init}$}](pmerged)at (5.5,1.25){};
\path [rectangle, fill=gray!10](3.5,-1.7) to (4.5,-1.7) to (4.5,1.5) to (3.5,1.5);
\node[label={[name=BlockLab]above:\tiny N-Updated}](p1Block)at(4,1.3){};
\node[place,label={[name=idpLab]below:\scriptsize $Id_{N}^u$}](idp)at (4,0){};
\node[place,label={[name=p1Lab]below:\scriptsize $p_1$}](p1)at (4,1){};
\node[place,label={[name=p2Lab]below:\scriptsize $p_2$}](p2)at (4,-1){};
\node[transvert,minimum height=12mm,minimum width=5mm] (tmerge)at (5.5,0){};
\node [rotate=-90]at(5.5,0){\tiny$\tau^{transfer}_e$};
\path [rectangle, fill=gray!10](6.5,-1.7) to (7.5,-1.7) to (7.5,1.5) to (6.5,1.5);
\node[label={[name=idpnBlockLab]above:\scriptsize \tiny $\hat{p}_2$\text{-block}}](pNBlock)at(7,1.3){};
\node[place,label={[name=idpnLab]below:\scriptsize $\Id_{\hat{p}_2}$}](idpn)at (7,0){};

\node[place,label={[name=p1NLab]below:\scriptsize $p_1$}](p1N)at (7,1){};
\node[place,label={[name=p2NLab]below:\scriptsize $p_2$}](p2N)at (7,-1){};

\draw[->](idp)--node[above,near start]{\scriptsize $x_{N}$}(tmerge);
\draw[->,myDouble] (p1) -- node[above,midway,sloped] {\scriptsize $x_{N}$}(tmerge);
\draw[->,triple] (p2) -- node[above,midway,sloped] {\scriptsize $x_{N}$}(tmerge);

\draw[->](tmerge) -- node[right,midway] {\scriptsize $x_{cs}$}(pmerged);

\draw[->,myDouble] (tmerge) -- node[above,midway,sloped] {\scriptsize $x_{Id}$}(p1N);
\draw[->] (tmerge) to node[midway, above]{\scriptsize$x_{Id}$} (idpn);
\draw[->,triple] (tmerge) -- node[above,midway,sloped] {\scriptsize $x_{Id}$}(p2N);

\draw[<-](tmerge.south west) -- node[midway,left]{\scriptsize $x_{cs}$}(pmerge.north);

\draw[->](pnew) -- node[midway,right]{\scriptsize$x_{Id}$} (tmerge.south east);

\end{tikzpicture}
}
        \caption{Module $e$-transfer.}
        \label{fig:etransfer}
    \end{subfigure}
    \caption{Module \(e\)-distributing.}
    \label{fig:edist}
\end{figure}

\noindent\textbf{Module \textbf{\(e\)-distributing}.}
This module captures the dynamics of the final distribution of updated tokens.
Let $n$ be the number of objects created by $\tau$ and let $\hat{p}_1,\dots,\hat{p}_n$ be an enumeration (possibly) with repetition of the places in the post-conditions of $\tau$. The module concatenates the sub-module \textbf{$e$-id-creation}, a sequence of sub-modules \textbf{$e$-move$(i)$} for $i\in\{1,\dots,n-1\}$, and the sub-module \textbf{$e$-transfer}. The sub-module \textbf{$e$-id-creation} generates $n$ new identifiers\footnote{This is the only place where we need name creation via $\nu$ variables.} for the encoding of the objects to be created. Then, each \textbf{$e$-move$(i)$} selects one such identifier, say $a$, and moves while renaming to $a$, one by one, some token not in $\Id_{N}^u$ from \textbf{$N$-updated} to the corresponding place in \textbf{$\hat{p}_i$-block}. By taking advantage of concurrency and of a couple of enable/acknowledge places as in the previous modules, each \textbf{$e$-move$(i)$} passes the turn to the next sub-module after a non-deterministic number of movements. 
While passing the turn, \textbf{$e$-move$(i)$} finally moves the identifier $a$ to $\Id_{\hat{p}_i}$, completing the encoding of the next object created by $\tau$ at place $\hat{p}_i$. Finally, the \textbf{$e$-transfer} sub-module creates the last encoding by moving, via channels, all remaining tokens in \textbf{$N$-updated} to \textbf{$\hat{p}_n$-block} while renaming to the last new identifier.

\begin{remark}
    The construction above (or slight modifications) cannot simulate non-conservative EOSs. In general EOS, if the transition $\tau$ in the event $e$ destroys a type $N$, then $e$ can fire only when destroying an empty object, performing a sort of zero-check. Otherwise, the transition cannot fire. Instead, the corresponding module $e$-merging, can fire even when consuming non-empty objects. Thus, this construction works only for cEOSs.
\end{remark}
 
\subsection{Configuration Encoding} \label{subsec:encoding}
Here we formally define the configuration encoding.
Say that $p$ is a place of $N\in \N$, we denote by $p^{\hat{q}}$ the corresponding place in $\hat{q}$-block for each system net place $\hat{q}$ of type $N$, by $p^{mer}$ the corresponding place in $N$-merged, and by $p^{upd}$ the corresponding place in $N$-updated. The following function $\K$ is used to transform objects in $\os$ as tuples in the modules of $W$.

\begin{definition}
    Given a configuration $c=\tup{\hat{p},m}$ of $\os$ where $d(\hat{p})=N=(P_N,T_N,F_N)$, we define the tuple 
    \begin{align*}
        \K(c)=&\sum_{p\in P_N} m(p)\delta_{p^{\hat{p}}}+\delta_{\Id_{\hat{p}}}&
        &
        \\
        \K_{mer}(c)=&\sum_{p\in P_N} m(p)\delta_{p^{mer}}+\delta_{\Id_N^m}&
        \bar{\K}_{mer}(N,{m})=&\sum_{p\in P_N} m(p)\delta_{p^{\hat{p}}}+\delta_{\Id_{\hat{p}}}\\
        \K_{upd}(c)=&\sum_{p\in P_N} m(p)\delta_{p^{upd}}+\delta_{\Id_N^u}&
        \bar{\K}_{mer}(N,{m})=&\sum_{p\in P_N} m(p)\delta_{p^{\hat{p}}}+\delta_{\Id_{\hat{p}}}
    \end{align*}
\end{definition}

We extend $\K$, $\K_{mer}$, and $\K_{upd}$ to sums of pairs $c_i=\tup{\hat{p}_i,m_i}$ in the trivial way: given a configuration $M=\sum_{i=1}^\ell \tup{\hat{p}_i,m_i}$ of $\os$, letting $c_i=\tup{\hat{p}_i,m_i}$, we define
    \begin{align*}
        \K(M)=&\sum_{i=1}^\ell \fmset{\K(c_i)}\\
        \K_{mer}(M)=&\sum_{i=1}^\ell \fmset{\K_{mer}(c_i)}\\
        \K_{upd}(M)=&\sum_{i=1}^\ell \fmset{\K_{upd}(c_i)}\\
    \end{align*}
We can now define the encoding of an arbitrary configuration.
\begin{definition}
    The encoding $\hat{M}$ of a configuration $M$ of $\os$ is the configuration $\hat{M}=\K(M)+\fmset{\delta_{p^{init}}}$ of $W$.
\end{definition}

\subsection{Simulation} \label{subsec:simulation}

\begin{theorem}
    Given an EOS \(\os\), two markings \(M,M'\), an event \(e =\tup{\tau,\theta}\), and mode \((\lambda,\rho)\), \(M \rightarrow^{(e,\lambda,\rho)}M'\) iff \(\hat{M} \rightarrow^\ast \hat{M'}\) in the \GnPN \(N\) using the above construction. 
\end{theorem}

Given an \(\os = (\hat{N},\N,d,\Theta)\), let the event $e=\tup{\tau,\theta}$ is enabled on the configuration $M=\fmset{\tup{\hat{p}_1,m_1}, \tup{\hat{p}_2,m_2}, \cdots, \tup{\hat{p}_{\abs{M}},m_{\abs{M}}}}$ with mode $(\lambda,\rho)$, and \(M'=\fmset{\tup{\hat{p}'_1,m'_1}, \tup{\hat{p}'_2,m'_2}, \cdots, \tup{\hat{p}'_{\abs{M'}},m'_{\abs{M'}}}}\) where $\lambda=\fmset{\tup{\hat{p}_{i_1},m_{i_1}}, \tup{\hat{p}_{i_2},m_{i_2}}, \cdots, \tup{\hat{p}_{i_{\abs{\lambda}}},m_{i_{\abs{\lambda}}}}}$ and $\rho=\fmset{\tup{\hat{p}_{j_1}',m'_{j_1}}, \tup{\hat{p}_{j_2}',m_{j_2}'}, \cdots, \tup{\hat{p}_{j_{\abs{\rho}}'},m_{j_{\abs{\rho}}}'}}$ where each \(i_k \in [\abs{M}]\) and \(j_{k'} \in [\abs{M'}]\). Recall that this proof is about the construction presented in Sec.~\ref{sec:fromEOS}, which in turn assumes that $\lambda$ and $\rho$ deal with a single system net type, i.e., $d(\hat{p}_{i_1}) = \cdots = d(\hat{p}_{i_{\abs{\lambda}}})= d(\hat{p}_{j_1})=\dots=d(\hat{p}'_{j_{\abs{\rho}}})=N$, for some $N\in\N$.

We define an injection \(H: [\abs{\lambda}] \rightarrow \Pi_1(\lambda) \times [\abs{\lambda}]\) that maps each pair of lambda to its system place and position in $\lambda$ when restricted only to that place, i.e., \(H(k)=(\hat{p}_{i_k},\ell)\) where \(\ell =|\{j \mid j\in\{1,\dots,k\}, \hat{p}_{i_j}=\hat{p}_{i_k}\}|\).

The transition $\tau_e^{merge}$ is enabled with mode $\varepsilon_{mer}$ such that :
\begin{itemize}
    \item $\varepsilon_{mer}(x_{cs})=\delta_{p^{init}}$. From now on, we will always associate the variable \(x_{cs}\) to the control sequence tuple.
    \item $\varepsilon_{mer}(x^{\hat{p}}_\ell)=\K(m_k)$ such that $H(k)=(\hat{p},\ell)$.
\end{itemize}
Note that we can write \(M = M'' + \lambda\) and \(M' = M'' + \rho\), for some \(M''\). 
After firing \(\tau_e^{merge}\) using the mode \(\varepsilon_{mer}\) on the configuration \(\hat{M} = \K(M'') + \K(\lambda) + \fmset{\delta_{p^{init}}}\), 
we get the resultant marking 
\[\hat{M}_{merged} = \K(M'') + \fmset{m_{merged}} + \fmset{\delta_{p^{merged}_e}}\]
where \(m_{merged} = \sum\limits_{\substack{<\hat{p},m> \in \lambda}} \sum\limits_{p \in P_{N}} m(p)\delta_{p^{mer}} + \delta_{\Id^m_N} = \K_{mer}(\Pi^2_{N}(\lambda))\). 
This is the configuration obtained after executing the module \textbf{$e$-merging}. 

Now, \(\delta_{p^{merged}_e}\) enables the execution of the \textbf{\(e\)-updating} module. After firing transitions \(\tau^{init}_e\), all the \(t_e\) transitions, and \(\tau^{fin}_e\) sequentially in the module \textbf{\(e\)-updating},  with a unique enabling mode (this mode exists since $(\lambda,\rho)$ is an enabling mode for \(e\)), we get the resultant marking 
\[\hat{M}_{updated} = \K(M'') + \fmset{m^0_{updated}} + \fmset{\delta_{p_e^{fin}}}\]
where 
\(m^0_{updated} = 
\sum\limits_{\substack{\tup{\hat{p}',m'} \in \rho}} \sum\limits_{p \in P_{N}}
m'(p)\delta_{p^{upd}} + \fmset{\delta_{\Id^u_{N}}}= \bar{\K}_{upd}(N,\Pi^2_{N}(\lambda) - \prefun_{N}(\theta(N)) + 
\postfun_{N}(\theta(N))) = 
\bar{\K}_{upd}(N,\Pi^2_{N}(\rho))\) by EOS semantics.

The token in $\delta_{p^{fin}_e}$ ensures that the control flows to the \textbf{$e$-distribution} module, consisting of sub-modules: \textbf{$e$-id-creation}, \textbf{$e$-move(i)} for each \(\tup{\hat{p}'_{i},n'_{i}} \in \rho\) and \textbf{$e$-transfer}, in that order). 

In the \textbf{$e$-id-creation} sub-module, on firing $t_e^{id}$,
a token is placed in $p^{new}_e$ for each new tuple generated, against each object net created by the event \(e\). Hence, we end up with the following resultant marking: 
\[\hat{M}_{created} = \K(M'') + \fmset{m^0_{updated}} + \fmset{c_1,c_2,\cdots,c_{\abs{\rho}}} + \fmset{\delta_{p^{move(1)}_e}}\]
where \(c_i = \delta_{p^{new}_e}\) for each \(i \in [\abs{\rho}]\).

Now, we can execute\textbf{ \(e\)-move($i$)} and \textbf{$e$-transfer} modules sequentially according to each \(r_i = \tup{\hat{p}'_{j_i},m'_{j_i}} \in \rho\). 
Specifically, for each $p\in P_N$, we fire $m'_{j_i}(p)$ times the transition of $e$-\textbf{move}-$(i)$ that consumes from $p^{upd}$.
After firing all the transitions in an \textbf{\(e\)-move($i+1$)} module for \(i \in \{0,\cdots,\abs{\rho}-2\}\), with the mode that instantiates \(x_{N}\) with \(m^i_{updated}\) and \(x_{new}\) with \(c_{i+1}\), we get,
\begin{align*}
    \hat{M}^{i+1}_{moved} =& \hat{M}^{i}_{moved} - \fmset{m^{i}_{updated}} + \fmset{m^{i+1}_{updated}} \\ 
    & - \fmset{c_{i+1}} + \K(r_{i+1})  - \fmset{\delta_{p^{move(i+1)}_e}} + \fmset{\delta_{p^{move(i+2)}_e}} 
\end{align*}
where $\hat{M}^{0}_{moved} = \hat{M}_{created}$, \(m^{i+1}_{updated} = \fmset{m^{i}_{updated} - \K_{upd}(r_{i+1})+\delta_{\Id^u_N}}\) 
and, with a slight abuse of notation, we define \(p^{move(\abs{\rho})}_e = p^{transfer}_e\).
Note that
\(\hat{M}_{moved}^i = \K(M'') + \fmset{m^i_{updated}} + \fmset{\K(r_1),\cdots,\K(r_i)} + \fmset{c_{i+1},\cdots,c_{\abs{\rho}}} + \fmset{\delta_{p^{move(i+1)}_e}}\)

Now, after executing module \textbf{\(e\)-move($\abs{\rho}-1$)} (or module $e$\textbf{-id-creation} if $\abs{\rho}=1$), \(\tau^{transfer}_e\) is enabled. After firing \(\tau^{transfer}_e\), we get the resultant marking:
\begin{align*}
    \hat{M}_{trans} =& \hat{M}^{\abs{\rho}-1}_{moved} - \fmset{m^{\abs{\rho}-1}_{updated}} + \K(r_{\abs{\rho}}) -\fmset{\delta_{p^{transfer}_e}} + \fmset{\delta_{p_{init}}} \\
=&\K(M'') + \fmset{\K(r_1),\cdots, \K(r_{\abs{\rho}})} + \fmset{\delta_{p^{init}}} \\
=&\K(M'') + \K(\rho) + \fmset{\delta_{p^{init}}} = \hat{M}'
\end{align*}
where the first equation follows from $m^{\abs{\rho}-1}_{updated} = m^0_{updated} - \sum\limits_{i \in \abs{\rho}-1}(\K_{upd}(r_{i})) = \K_{upd}(N,\Pi^2_{N}(\rho)) - \sum\limits_{i \in \abs{\rho}-1}(\K_{upd}(r_{i})) = \K_{upd}(r_{\abs{\rho}})$. Thus, $\hat{M}\rightarrow^\ast \hat{M}'$.

The other direction of the proof is analogous.
 
\section{Conclusions}\label{sec:conclusions}
\begin{figure}[t]
    \centering
    \scalebox{0.7}{
    \begin{tikzpicture}

\node[rectangle]at(-3,0)(h1){\textit{complexity classes}};

\node at (-3,-1)(top){$F_{\omega^\omega}$};
\node []at (-3,-1.5)(first){strictly intermediate};
\node at (-3,-2)(bot){$F_{\omega 2}$};

\node[rectangle]at(1,0)(h2){\textit{data perspective}};

\node at (1,-1)(UDN){UDN coverability};
\node [text=green!60!black]at (1,-1.5)(second){\GnPN/\rnPN coverability};
\node at (1,-2)(NU){\nPN coverability};

\node[rectangle]at(6,0)(h3){\textit{process perspective}};

\begin{scope}
\node at (6,-1.25)(prob1){EOS/cEOS lossy reachability};
\node at (6,-1.75)(prob3){cEOS (lossy) coverability};
\node[fit={(prob1)(prob3)}](probs){};

\draw [green!60!black,decorate,decoration={brace,amplitude=5pt}]
   (probs.south west) --(probs.north west);

\end{scope}

\draw[](3.5,.25) -- (3.5,-2.5);
\draw[](-1.25,.25) -- (-1.25,-2.5);
\draw[dashed,green!60!black] (first.east)--(second.west);
\draw[dashed,green!60!black] (second.east)--($(probs.west)-(.2,0)$);

\draw[dashed] (top.east)--(UDN.west);
\draw[dashed] (bot.east)--(NU.west);
\draw (h1.south west)--(h3.south east);

\end{tikzpicture}
    }
    \caption{Complexity of EOS coverability under the lens of data nets. Contributions of this paper are in green.}
    \label{fig:dataprocess}
\end{figure}

We have detected \GnPN as the data-aware counterpart of cEOS, when restricted to coverability.\footnote{The reduction from cEOS to \GnPN is a full-fledged simulation that works beyond the scope of coverability, unlike the reduction from \rnPN coverability to cEOS coverability.} \nPNs are captured by \GnPN which, in turn, are captured by UDNs. In fact, since we are interested only in coverability in a WSTS setting, we can seamlessly substitute the fresh-name creation of \nPN with the lossy-name creation of UDNs: if coverability is witnessed by a run where lossy-name creation deletes some tuple, then there is another run where lossy-name creation behaves as fresh-name creation.
This immediately yields three contributions: \begin{enumerate*}
    \item the insight that the essence of the computational power of cEOS coverability corresponds to anonymous unordered data (as in \nPN) plus selective transfer with renaming\footnote{while anonymous data means that data labels remain implicit, renaming of anonymous data means that tokens change their implicit name.}
    \item an $F_{\omega2}$ lower bound from standard \nPN coverability and
    \item an $F_{\omega^\omega}$ upper bound from UDN coverability.
\end{enumerate*}
Arguably, both these bounds are sub-optimal. On the one hand, by $F_{\omega2}$-completeness of \nPN coverability, either \begin{enumerate*}
    \item \nPN can express, up to coverability, selective transfers with renaming, which seems unlikely, or
    \item the complexity of cEOS coverability is higher than that of \nPNs.
\end{enumerate*}
On the other hand, the $F_{\omega^\omega}$-hardness of UDN in~\cite{Rosa-Velardo17} is obtained by taking advantage of arbitrary affine transformations (amounting to a general matrix $G$ in \GnPN) and broadcasts. Thus, either \begin{enumerate*}
    \item the notion of broadcast in UDN does not have any impact on the complexity of coverability, which seems again unlikely, or
    \item cEOS coverability is not $F_{\omega^\omega}$-complete.
\end{enumerate*}
Consequently, we conjecture that cEOS coverability is $F_{\alpha}$-complete for some ordinal $\alpha$ such that $\omega2<\alpha<\omega^\omega$.
As future works, we plan to obtain the precise complexity class of cEOS coverability. This endevour is now easier, since we can directly take advantage of the literature on coverability of various forms of data nets. Furthermore, we plan to study the impact of several iterations of nesting, as in general NWNs. Results of this kind would help us understand the precise effect of nesting (and iterated nesting) on complexity.

%
%
%
\bibliographystyle{splncs04}
\bibliography{biblio}

\appendix
\section{Complexity Equivalence of Lossy EOS and perfect cEOS coverability}\label{Sec:equivalence}
In this section, we show that lossy EOS coverability and cEOS coverability are polynomially inter-reducible, and, thus, are equivalent from a complexity standpoint. In turn, lossy EOS coverability coincides with lossy EOS reachability~\cite{OurPNSE24}. We begin by precisely stating these two coverability problems using the definitions below.

\begin{definition}    
    A \emph{$(\leq_f,\omega)$-run} in EOS $\os$ is a finite sequence of configurations $C_1,C_2,\dots$ such that, for each $i\in\mathbb{N}$, $C_1\rightarrow^e C_2$ for some event $e$, called \emph{standard steps}, or $C_2\leq_f C_1$, called \emph{lossy steps}.    
    The set of $(\leq_f,\omega)$-runs from $\mu_0$ is denoted by $\lruns[\omega]{\leq_f}{\mu_0}$.
\end{definition}

\begin{definition}[$(\leq_f,\omega)$-reachability/coverability for EOSs (cEOSs)]
    \begin{compactitem}
        \item[Input:] An EOS (cEOS) $E$, an initial marking $\mu_0$ and a target marking $\mu_1$ for $E$.
        \item[Output of reachability:] Is there a run $\sigma\in\lruns[\omega] {\leq_f}{\mu_0}$ from $\mu_0$ to $\mu_1$?
        \item[Output of coverability:] Is there a run $\sigma\in\lruns[\omega] {\leq_f}{\mu_0}$ from $\mu_0$ to some $\mu$ such that $\mu\geq\mu_1$?
    \end{compactitem}
\end{definition}

Now, with the definitions in place, we show that cEOS $\leq_f$-coverability and EOS $(\leq_f,\omega)$-coverability are inter-reducible.
The reduction from cEOS $\leq_f$-coverability to EOS $(\leq_f,\omega)$-coverability is immediate.

\begin{lemma}\label{lem:polynomialReductionTrivial}
    There is a polynomial reduction from cEOS $\leq_f$-coverability to EOS $(\leq_f,\omega)$-coverability.
\end{lemma}
\begin{proof}
    cEOS $\leq_f$-coverability coincides with cEOS $(\leq_f,\omega)$-coverability, which, in turn, is a special case of EOS $(\leq_f,\omega)$-coverability, since each cEOS is an EOS.
\end{proof}

We now show the converse of Lemma~\ref{lem:polynomialReductionTrivial}.
In the remainder of the section, except explicitly specified, we work with a fixed arbitrary EOS $\os=\tup{\hat{N},\N,d,\Theta}$ where $\hat{N}=\tup{\hat{P},\hat{T},\hat{F}}$. 
The intuition is that we can make each transition of $\os$ conservative, by adding to each transition $t$ a post-condition to a \textit{trash} place for each type \textit{destroyed} by $t$, that is, consumed but not produced back. The resulting 
EOS is called the \textit{conservative closure} of $t$ and of $\os$, respectively.

\begin{definition}
    A transition $t\in\hat{T}$ \textit{destroys} a type $N\in\N$ if $t$ consumes some token from a place of type $N$, but does not produce any token in any place of type $N$:
    \[(\exists p\in \hat{P}\,d(p)=N\land F(p,t)\neq 0) \land (\forall q \in \hat{P}\, d(q)=N\rightarrow F(t,q)= 0)\]
    The \textit{set of types destroyed by $t$} is denoted by $\dest(t)$.
\end{definition}

For technical reasons, we assume without loss of generality that, for each transition $t\in\hat{T}$, \begin{inparaenum}[\itshape (1)]
    \item $t$ is involved in one and only one event in $\Theta$ and
    \item if $\dest(t)\neq\emptyset$ then $t$ is involved in a system autonomous event.
\end{inparaenum}
In fact, if $t$ is involved in zero events, then it cannot fire and does not contribute to the set of reachable/coverable configuration. Thus, it can be removed. Instead, if it is involved in $\theta_1,\dots,\theta_n$ events, we can substitute $t$ in $\theta_i$ with a dedicated copy $t_i$ of $t$, for each $i\in\{1,\dots,n\}$. Moreover, if $\dest(t)\neq\emptyset$ but $t$ is involved in a synchronization event, then we can split $t$ into two transitions $t_\text{pre}$ and $t_\text{post}$ as depicted in Fig.~\ref{fig:destroySystemAutonomous}, where $\dest(t_\text{pre})=\emptyset$, $t_\text{pre}$ is involved in the corresponding synchronization event, and $t_\text{post}$ is involved in a system autonomous event. On the one hand, the firing of $t_\text{pre}$ followed by the firing of $t_\text{post}$ simulates a single firing of $t$. On the other, the \textit{enable} places force the two transitions to fire in turns. If lossiness affects some of the \textit{enable} places or some of the $\text{inter}$ places, then $t_\text{post}$ and, consequently, $t_\text{pre}$ get disabled forever and do not contribute anymore to the set of reachable/coverable configurations. Note that this transformation is polynomial.

\begin{figure}[t]

\begin{subfigure}[b]{0.39\textwidth}
\resizebox{1\columnwidth}{!}{
\begin{tikzpicture}

\node[place,label={[name=p1Lab]left:{\scriptsize $p_1$}}] (p1)at (0,0){};
\node [rotate=90](pDots)at (0,-0.625){$\dots$};
\node[place,label={[name=pmLab]left:{\scriptsize $p_m$}}] (pm)at (0,-1.25){};
\node [rotate=90](pDots)at (0,-0.625*3){$\dots$};
\node[place,label={[name=pnLab]left:{\scriptsize $p_n$}}] (pn)at (0,-2.5){};
\node[transvert,label={[name=tLab]above:\scriptsize $t$},label={[name=tSynchLab]below:\scriptsize $\tup{M}$}] (t)at (1.5,-1.25){};
\draw[->] (p1) -- (t);
\draw[->] (pm) -- (t);
\draw[->] (pn) -- (t);

\node[place,label={[name=q1Lab]right:{\scriptsize $q_1$}}] (q1)at (3,0){};
\node [rotate=90](qDots)at (3,-0.625){$\dots$};
\node[place,label={[name=qmLab]right:{\scriptsize $q_h$}}] (qm)at (3,-1.25){};
\draw[<-] (q1) -- (t);
\draw[<-] (qm) -- (t);

\draw [dashed] (pmLab.south east) -- (pmLab.south west) |- (pnLab.south west) -- (pnLab.south east);
\node[rotate=90] at (-1.25,-0.625*3-0.25){$\dest$};

\node[place,draw=white] (enablePre)at (0,.5){};
    \node[place,draw=white] (enablePost)at (0,-3){};

\end{tikzpicture}
}
    \caption{}
    \label{fig:destroySystemAutonomousPre}
\end{subfigure}
\hfill
\begin{subfigure}[b]{0.59\textwidth}
\resizebox{1\columnwidth}{!}{
\begin{tikzpicture}

    \node[place,label={[name=p1Lab]left:{\scriptsize $p_1$}}] (p1)at (0,0){};
    \node [rotate=90](pDots)at (0,-0.625){$\dots$};
    \node[place,label={[name=pmLab]left:{\scriptsize $p_m$}}] (pm)at (0,-1.25){};
    \node [rotate=90](ppDots)at (0,-0.625*3){$\dots$};
    \node[place,label={[name=pnLab]left:{\scriptsize $p_n$}}] (pn)at (0,-2.5){};
    \node[transvert,label={[name=tpreLab]above:\scriptsize $t_\text{pre}$},label={[name=tpreSynchLab]below:\scriptsize $\tup{M}$}] (tpre)at (1.5,-1.25){};
    \draw[->] (p1) -- (tpre);
    \draw[->] (pm) -- (tpre);
    \draw[->] (pn) -- (tpre);

    \node[place,label={[name=N1Lab]right:{\scriptsize $\text{inter}_1$}}] (N1)at (3,0){};
    \node [rotate=90](NDots)at (3,-1.25){$\dots$};
    \node[place,label={[name=NkLab]right:{\scriptsize $\text{inter}_k$}}] (Nk)at (3,-2.5){};
    \draw[<-] (N1) -- (tpre);
    \draw[<-] (Nk) -- (tpre);

    \node[transvert,label={[name=tpostLab]above:\scriptsize $t_\text{post}$}] (tpost)at (4.5,-1.25){};
    
    \node[place,label={[name=q1Lab]right:{\scriptsize $q_1$}}] (q1)at (6,0){};
    \node [rotate=90](qDots)at (6,-0.625){$\dots$};
    \node[place,label={[name=qmLab]right:{\scriptsize $q_h$}}] (qm)at (6,-1.25){};
    \draw[->] (N1) -- (tpost);
    \draw[->] (Nk) -- (tpost);
    \draw[<-] (q1) -- (tpost);
    \draw[<-] (qm) -- (tpost);

    \node[place,label={[name=enablePreLab] left:{\scriptsize $\text{enable}_\text{pre}$}}] (enablePre)at (1.5,.5){};
    \node at (enablePre){$\bullet$};
    \node[place,label={[name=enablePostLab]right:{\scriptsize $\text{enable}_\text{post}$}}] (enablePost)at (4.5,-3){};
    \draw[->] (enablePre.south) -- (tpreLab.north);
    \draw[<-] (enablePre.east) -| (tpostLab.north);
    \draw[->] (enablePost.north) -- (tpost.south);
    \draw[->] (tpreSynchLab.south) |- (enablePost.west);  

\end{tikzpicture}
}
    \caption{}
    \label{fig:destroySystemAutonomousPost}
\end{subfigure}

    \caption{(\subref{fig:destroySystemAutonomousPre}) A system-net transition involved in a synchronization event, where $\{d(p_1,)\dots,d(p_m)\}\cap\{d(p_{m+1},\dots,d_n)\}=\emptyset$ and $\dest(t)=\{d(p_{m+1}),\dots,d(p_n)\}$ and (\subref{fig:destroySystemAutonomousPost}) its simulation with two transitions $t_\text{pre}$ and $t_\text{post}$ such that $\dest(t_\text{pre})=\emptyset$, $\dest(t_\text{post})=\dest(t)$, $\{d(p_1),\dots,d(p_n)\}=\{\text{inter}_1,\dots,\text{inter}_k\}$, $d(\text{inter}_i)\neq d(\text{inter}_j)$ for $i\neq j$, $d(\text{enable}_\text{pre})=d(\text{enable}_\text{post})$ is a new type, and $t_\text{post}$ is involved in a system autonomous event.}
    \label{fig:destroySystemAutonomous}
\end{figure}

\begin{definition}\label{def:consEOS}
    The \textit{conservative closure $\os^\cons$ of $\os$} is the EOS 
    
    $\os^\cons=\tup{\hat{N}^\cons,\N,d^\cons,\Theta}$ such that 
    \begin{enumerate}
        \item $\hat{N}^\cons=\tup{\hat{P}^\cons,\hat{T},\hat{F}^\cons}$ where 
            \begin{enumerate}
                \item $\hat{P}^\cons=\hat{P}\cup\bigcup_{N\in\N}\{\mathit{trash}_{N}\}$, for new places $\mathit{trash}_{N}$ not already in $\hat{P}$, for each $N\in\N$.
               
                \item For each place $p\in \hat{P}^\text{cons}$ and transition $t\in\hat{T}$, we have $F^\cons(p,t)=F(p,t)$ and
                \[F^\cons(t,p)=\begin{cases}
                    F(t,p) &\text{if } p\in\hat{P}\\
                    1 & \text{if } p=\text{trash}_{d(q)}\text{ for some }q\in\dest(t)\\
                    0 & \text{otherwise}
                \end{cases}\]
            \end{enumerate}
            \item For each $p\in\hat{P}^\cons$:
            \[d^\cons(p)=\begin{cases}
                d(p)&\text{if }p\in\hat{P}\\
                d(N)&\text{if }p=\text{trash}_N\text{ for some}N\in\N
            \end{cases}\]
            
    \end{enumerate}
\end{definition}

Note that $\os^\cons$ is a cEOS. For clarity, when referring to a transition $t\in\hat{T}$ or event $e\in\Theta$ in the context of $\os^\cons$, we denote it by $t^\cons$ or $e^\cons$, respectively. For example, for each $t\in\hat{T}$, we have that $\dest(t)$ may be non-empty, while $\dest(t^\cons)=\emptyset$. 
\begin{example}
Fig.~\ref{fig:equivalenceNonConservativeTransitionPre} depicts a non-conservative EOS $\os$ with a configuration $m$. The conservative closure $\os^\cons$ of $\os$ is depicted in Fig.~\ref{fig:equivalenceConservativeTransitionPre}. The transition in $\os$ is enabled only after loosing two tokens in the $k_1$ places, so that the object tokens in $p_1$ and $p_2$ can be erased by the transition firing, reaching the configuration $m'$ in Fig.~\ref{fig:equivalenceNonConservativeTransitionPre}. Instead, the transition in $\os^\cons$ is enabled and its firing reaches the same configuration $m'$, up to a number of tokens in the \textit{trash} places.
\begin{figure}[t]
    \centering

\begin{subfigure}[b]{0.54\textwidth}
    \resizebox{.99\columnwidth}{!}{
\begin{tikzpicture}
\begin{scope}
        \node[place,label={[name=p1Lab]below:{\scriptsize $p_1$}}] (p1)at (0,0){};
        \node[place,label={[name=p2Lab]below:{\scriptsize $p_2$}}] (p2)at (0,-1.25){};
        \node[place,label={[name=p3Lab]below:{\scriptsize $p_3$}}] (p3)at (0,-2.5){};
        \node[transvert,label={[name=tLab]right:\scriptsize $t$}] (tran)at (1.5,-1.25){};
        \node (p3Token)at (p3) {$\bullet$};
        \node (p1Token)at (p1) {$\bullet$};
        \node (p2Token)at (p2) {$\bullet$};

        \draw[->] (p1.east) -| (tran.north); 
        \draw[->] (p2.east) -- (tran.west); 
        \draw[->] (p3.east) -| ($(tran.south)+(-.1,0)$); 

        \node[place,label={[name=p4Lab]below:{\scriptsize $p_4$}}] (p4) at (3,-2.5){};
        \draw[->] ($(tran.south)+(.1,0)$) |- (p4.west);

        \begin{scope}[xshift=-3cm,yshift=-2.5cm]
        \node[place,label={[name=q1Lab]below:{\scriptsize $q_1$}}] (q1) at (0,0){};
        \node[place,label={[name=q2Lab]below:{\scriptsize $q_2$}}] (q2)at (2,0){};
        \node[transvert,label={[name=t1Lab]below:\scriptsize $t_2$}] (t1)at (1,0){};
        \node [draw,fit={(q1)(q1Lab)(q2)(q2Lab)(t1)(t1Lab)}] (obj1) {};
        \draw[->] (q1.east) -- (t1.west); 
        \draw[<-] (q2.west) -- (t1.east); 
        \draw[dashed] (obj1) -- (p3Token); 
        \node (q1Token)at (q1) {$\bullet$};
        \node (q1Token)at (q2) {$\bullet$};
        \end{scope}

        \begin{scope}[xshift=-5cm,yshift=0cm]
        \node[place,label={[name=k1Lab]below:{\scriptsize $k_1$}}] (k1)at (1,0){};
        \node[transvert,label={[name=t2Lab]below:\scriptsize $t_1$}] (t2)at (2,0){};
        \node [draw,fit={(k1)(k1Lab)(t2)(t2Lab)}] (obj2) {};
        \draw[->] (k1.east) -- (t2.west); 
        \draw[dashed] (obj2) -- (p1Token); 
        \node (k1Token)at (k1) {$\bullet$};
        \end{scope}

        \begin{scope}[xshift=-3cm,yshift=-1.25cm]
        \node[place,label={[name=k1bisLab]above:{\scriptsize $k_1$}}] (k1bis)at (1,0){};
        \node[transvert,label={[name=t2bisLab]above:\scriptsize $t_1$}] (t2bis)at (2,0){};
        \node [draw,fit={(k1bis)(k1bisLab)(t2bis)(t2bisLab)}] (obj3) {};
        \draw[->] (k1bis.east) -- (t2bis.west); 
        \draw[dashed] (obj3) -- (p2Token); 
        \node (k1bisToken)at (k1bis) {$\bullet\bullet$};
        \end{scope}
    \end{scope}
\end{tikzpicture}
    }
    \caption{}
    \label{fig:equivalenceNonConservativeTransitionPre}    
\end{subfigure}
\hfill
\begin{subfigure}[b]{0.45\textwidth}
    \resizebox{.99\columnwidth}{!}{
\begin{tikzpicture}
\begin{scope}[xshift=4cm]
        \node[place,label={[name=p1secondLab]below:{\scriptsize $p_1$}}] (p1second)at (0,0){};
        \node[place,label={[name=p2secondLab]below:{\scriptsize $p_2$}}] (p2second)at (0,-1.25){};
        \node[place,label={[name=p3secondLab]below:{\scriptsize $p_3$}}] (p3second)at (0,-2.5){};
        \node[transvert,label={[name=tLabsecond]right:\scriptsize $t$}] (transecond)at (1.5,-1.25){};

        \draw[->] (p1second.east) -| (transecond.north); 
        \draw[->] (p2second.east) -- (transecond.west); 
        \draw[->] (p3second.east) -| ($(transecond.south)+(-.1,0)$); 

        \node[place,label={[name=p4Labsecond]below:{\scriptsize $p_4$}}] (p4second) at (3,-2.5){};
        \draw[->] ($(transecond.south)+(.1,0)$) |- (p4second.west);
        \node (p4Tokensecond)at (p4second) {$\bullet$};

        \begin{scope}[xshift=4cm,yshift=-2.5cm]
        \node[place,label={[name=q1Labsecond]below:{\scriptsize $q_1$}}] (q1second) at (0,0){};
        \node[place,label={[name=q2Labsecond]below:{\scriptsize $q_2$}}] (q2second)at (2,0){};
        \node[transvert,label={[name=t1Labsecond]below:\scriptsize $t_2$}] (t1second)at (1,0){};
        \node [draw,fit={(q1second)(q1Labsecond)(q2second)(q2Labsecond)(t1second)(t1Labsecond)}] (obj1second) {};
        \draw[->] (q1second.east) -- (t1second.west); 
        \draw[<-] (q2second.west) -- (t1second.east); 
        \draw[dashed] (obj1second) -- (p4Tokensecond); 
        \node (q2Tokensecond)at (q2second) {$\bullet$};
        \node (q1Token)at (q1second) {$\bullet$};
        \end{scope}

    \end{scope}
\end{tikzpicture}
    }
    \caption{}
    \label{fig:equivalenceNonConservativeTransitionPost}
\end{subfigure}
    \caption{(\subref{fig:equivalenceNonConservativeTransitionPre}) A configuration $m$ of a non-conservative EOS $\os$ with two types $N_1=d(p_1)=d(p_2)$ and $N_2=d(p_3)=d(p_4)$, and 
    (\subref{fig:equivalenceNonConservativeTransitionPost}) the configuration $m'$ obtained by, first, loosing all tokens from the places $k_1$ and, second, firing $t$. 
    Note that $t$ could have fired without the preliminary lossy step only if the objects in $p_1$ and $p_2$ had no internal token.
    }
    \label{fig:equivalenceNonConservativeTransition}
\end{figure}

 
\begin{figure}[t]
    \centering
\begin{subfigure}[b]{0.54\textwidth}
    \resizebox{.99\columnwidth}{!}{
\begin{tikzpicture}
    \begin{scope}
        \node[place,label={[name=p1Lab]below:{\scriptsize $p_1$}}] (p1)at (0,0){};
        \node[place,label={[name=p2Lab]below:{\scriptsize $p_2$}}] (p2)at (0,-1.25){};
        \node[place,label={[name=p3Lab]below:{\scriptsize $p_3$}}] (p3)at (0,-2.5){};
        \node[transvert,label={[name=tLab]right:\scriptsize $t^\cons$}] (tran)at (1.5,-1.25){};
        \node (p3Token)at (p3) {$\bullet$};
        \node (p1Token)at (p1) {$\bullet$};
        \node (p2Token)at (p2) {$\bullet$};

        \draw[->] (p1.east) -| ($(tran.north)+(-.1,0)$); 
        \draw[->] (p2.east) -- (tran.west); 
        \draw[->] (p3.east) -| ($(tran.south)+(-.1,0)$); 

        \node[place,label={[name=p4Lab]below:{\scriptsize $p_4$}}] (p4) at (3,-2.5){};
        \draw[->] ($(tran.south)+(.1,0)$) |- (p4.west);

        \node[place,label={[name=trashLab]below:{\scriptsize $\mathit{trash}_{N_1}$}}] (trash) at (3,0){};
        \node[place,label={[name=trashLabsecond]below:{\scriptsize $\mathit{trash}_{N_2}$}}] (trash2second) at (3,-1.25){};
        \draw[<-] (trash.west) -| ($(tran.north)+(.1,0)$); 

        \begin{scope}[xshift=-3cm,yshift=-2.5cm]
        \node[place,label={[name=q1Lab]below:{\scriptsize $q_1$}}] (q1) at (0,0){};
        \node[place,label={[name=q2Lab]below:{\scriptsize $q_2$}}] (q2)at (2,0){};
        \node[transvert,label={[name=t1Lab]below:\scriptsize $t_2$}] (t1)at (1,0){};
        \node [draw,fit={(q1)(q1Lab)(q2)(q2Lab)(t1)(t1Lab)}] (obj1) {};
        \draw[->] (q1.east) -- (t1.west); 
        \draw[<-] (q2.west) -- (t1.east); 
        \draw[dashed] (obj1) -- (p3Token); 
        \node (q1Token)at (q1) {$\bullet$};
        \node (q2Token)at (q2) {$\bullet$};
        \end{scope}

        \begin{scope}[xshift=-5cm,yshift=0cm]
        \node[place,label={[name=k1Lab]below:{\scriptsize $k_1$}}] (k1)at (1,0){};
        \node[transvert,label={[name=t2Lab]below:\scriptsize $t_1$}] (t2)at (2,0){};
        \node [draw,fit={(k1)(k1Lab)(t2)(t2Lab)}] (obj2) {};
        \draw[->] (k1.east) -- (t2.west); 
        \draw[dashed] (obj2) -- (p1Token); 
        \node (k1Token)at (k1) {$\bullet$};
        \end{scope}

        \begin{scope}[xshift=-3cm,yshift=-1.25cm]
        \node[place,label={[name=k1bisLab]above:{\scriptsize $k_1$}}] (k1bis)at (1,0){};
        \node[transvert,label={[name=t2bisLab]above:\scriptsize $t_1$}] (t2bis)at (2,0){};
        \node [draw,fit={(k1bis)(k1bisLab)(t2bis)(t2bisLab)}] (obj3) {};
        \draw[->] (k1bis.east) -- (t2bis.west); 
        \draw[dashed] (obj3) -- (p2Token); 
        \node (k1bisToken)at (k1bis) {$\bullet\bullet$};
        \end{scope}
    \end{scope}
\end{tikzpicture}
}
\caption{}
\label{fig:equivalenceConservativeTransitionPre}
\end{subfigure}
\hfill
\begin{subfigure}[b]{0.45\textwidth}
    \resizebox{.99\columnwidth}{!}{
\begin{tikzpicture}
    \begin{scope}[xshift=4cm]
        \node[place,label={[name=p1secondLab]below:{\scriptsize $p_1$}}] (p1second)at (0,0){};
        \node[place,label={[name=p2secondLab]below:{\scriptsize $p_2$}}] (p2second)at (0,-1.25){};
        \node[place,label={[name=p3secondLab]below:{\scriptsize $p_3$}}] (p3second)at (0,-2.5){};
        \node[transvert,label={[name=tLabsecond]right:\scriptsize $t^\cons$}] (transecond)at (1.5,-1.25){};

        \draw[->] (p1second.east) -| ($(transecond.north)+(-.1,0)$); 
        \draw[->] (p2second.east) -- (transecond.west); 
        \draw[->] (p3second.east) -| ($(transecond.south)+(-.1,0)$); 

        \node[place,label={[name=p4Labsecond]below:{\scriptsize $p_4$}}] (p4second) at (3,-2.5){};
        \draw[->] ($(transecond.south)+(.1,0)$) |- (p4second.west);
        \node (p4Tokensecond)at (p4second) {$\bullet$};

        \node[place,label={[name=trashLabsecond]below:{\scriptsize $\mathit{trash}_{N_1}$}}] (trashsecond) at (3,0){};
        \draw[<-] (trashsecond.west) -| ($(transecond.north)+(.1,0)$); 
        \node[place,label={[name=trashLabsecond]below:{\scriptsize $\mathit{trash}_{N_2}$}}] (trash2second) at (3,-1.25){};
        \draw[<-] (trashsecond.west) -| ($(transecond.north)+(.1,0)$); 

        \node (trashsecondToken)at(trashsecond){$\bullet$};

        \begin{scope}[xshift=4cm,yshift=-2.5cm]
        \node[place,label={[name=q1Labsecond]below:{\scriptsize $q_1$}}] (q1second) at (0,0){};
        \node[place,label={[name=q2Labsecond]below:{\scriptsize $q_2$}}] (q2second)at (2,0){};
        \node[transvert,label={[name=t1Labsecond]below:\scriptsize $t_2$}] (t1second)at (1,0){};
        \node [draw,fit={(q1second)(q1Labsecond)(q2second)(q2Labsecond)(t1second)(t1Labsecond)}] (obj1second) {};
        \draw[->] (q1second.east) -- (t1second.west); 
        \draw[<-] (q2second.west) -- (t1second.east); 
        \draw[dashed] (obj1second) -- (p4Tokensecond); 
        \node (q2Tokensecond)at (q2second) {$\bullet$};
        \node (q1Tokensecond)at (q1second) {$\bullet$};
        \end{scope}

        \begin{scope}[xshift=4cm,yshift=0cm]
        \node[place,label={[name=k1Labsecond]below:{\scriptsize $k_1$}}] (k1second)at (2,0){};
        \node[transvert,label={[name=t2Labsecond]below:\scriptsize $t_1$}] (t2second)at (1,0){};
        \node [draw,fit={(k1second)(k1Labsecond)(t2second)(t2Labsecond)}] (obj2second) {};
        \draw[<-] (k1second.west) -- (t2second.east); 
        \draw[dashed] (obj2second) -- (trashsecondToken); 
        \node [yshift=.065cm]
        (k1Tokensecond)at (k1second) {$\bullet \bullet$};
        \node [yshift=-.085cm](k1MoreTokensecond)at (k1second) {$\bullet$};
        \end{scope}
    \end{scope}
\end{tikzpicture}
}
\caption{}
\label{fig:equivalenceConservativeTransitionPost}
\end{subfigure}
    \caption{(\subref{fig:equivalenceConservativeTransitionPre}) The conservative closure $\os^\cons$ of the EOS $\os$ in Fig.~\ref{fig:equivalenceNonConservativeTransitionPre} with the same configuration $m$, and (\subref{fig:equivalenceConservativeTransitionPost}) the configuration $m'$ obtained by firing $t^\cons$ (without any lossy step). The configuration $m'$ covers the configuration in \ref{fig:equivalenceNonConservativeTransitionPost}. 
    }
    \label{fig:equivalenceConservativeTransition}
\end{figure}
 
\end{example}

We can use $\os^\cons$ to solve $(\leq_f,\omega)$-coverability problems in $\os$. In fact, we now show that, given an initial configuration $\iota$ and a target configuration $\tau$ of $\os$, we have that $\tau$ is $(\leq_f,\omega)$-coverable from $\iota$ in $\os$ if and only if $\tau$ is $\leq_f$ coverable from $\iota$ in $\os^\cons$. 

\begin{definition}
    A configuration $k$ of $\os^\cons$ is \textit{over the trash places} if $k=\sum_{i=1}^k\tup{\text{trash}_{N_i},\mu_i}$ for some $k\in\mathbb{N}$ and $N_1,\dots,N_k\in\N$.
\end{definition}
Note that each configuration $m$ of $\os^\cons$ can be written, in a unique way, as $m=m'+k'$ where $m'$ is a configuration of $\os$ and $k'$ is a configuration of $\os^\cons$ over the trash places.
\begin{lemma}\label{lemma:stepEOScEOS}
    Let $e=\tup{t,\theta}\in\Theta$. If $m$ and $m'$ are configurations of $\os$ such that $m\rightarrow^e m'$ in $\os$, then there is some configuration $k$ of $\os^\cons$ over the trash places such that $m\rightarrow^{e^\cons} m'+k$ in $\os^\cons$.
\end{lemma}
\begin{proof}
If $\dest(t)=\emptyset$, then $t^\cons$ and $t$ enforce the same dynamics. Thus, the statement holds when setting $k=\varepsilon$.
Otherwise, $e$ is system autonomous, , i.e., $\theta=\varepsilon$. In this case the statement holds for $k=\sum_{N\in\dest(t)}\tup{\text{trash}_N,0}$. In fact, since $m\rightarrow^{t,\lambda,\rho} m'$ for some mode $(\lambda,\rho)$, we have $m'=m-\lambda+\rho$ and $\lambda\leq_s m$. By EOS semantics, since $e$ is system autonomous, $\Pi^2_N(\lambda)=\Pi^2_N(\rho)=0$ for each $N\in\dest(t)$. Let $\rho'=\rho + k$. We have:
\begin{itemize}
    \item $\Pi^1(\lambda)=\prefun_{\hat{N}}(t)=\prefun_{\hat{N}^\cons}(t^\cons)$.
    \item For each $p\in\hat{P}^\cons$, we have 
    \[\Pi^1(\rho')(p)=\Pi^1(\rho)(p)+\begin{cases}1&\text{if }p=\text{trash}_{d(p)}\\ 0&\text{otherwise}\end{cases}=\postfun_{\hat{N}^\cons}(t^\cons)(p)\]
    \item For each $N\in\dest(t)$, $\Pi^2_N(\rho')= \Pi^2_N(\rho)+\Pi^2_N(k)=0+0=\Pi^2_N(\lambda)$.
\end{itemize}
Thus, in $\os^\cons$, we have that $e^\cons$ is enabled on $m$ with mode $(\lambda,\rho')$ and $m\rightarrow^{e^\cons} m-\lambda+\rho'=m'+k$.
\end{proof}

By EOS semantics, we immediately obtain that, for each configuration $k$ of $\os^\cons$ over the trash places, $m\rightarrow^e m+k$ in $\os$ implies $m+k\rightarrow^{e^\cons} m' + k'+k$ in $\os^\cons$. Thus, Lemma~\ref{lemma:stepEOScEOS} can be generalized to the transitive closure of $\rightarrow$ in $\os$ and $\os^\cons$. Hence, if there is some $m\geq_f\tau$ $(\leq_f,\omega)$-reachable from $\iota$ in $\os$, then, there is some configuration $k$ of $\os^\cons$ over the trash places such that $m+k\geq_f\tau$ is reachable from $\iota$ in $\os^\cons$. We have obtained the following theorem.

\begin{theorem}\label{theorem:equivalenceEOScEOS}
    If $\tau$ is $(\leq_f,\omega)$-coverable from $\iota$
 in $\os$, then $\tau$ is $\leq_f$-coverable from $\iota$
 in $\os^\cons$.
 \end{theorem}

We now show the converse of Theorem~\ref{theorem:equivalenceEOScEOS}.
\begin{lemma}\label{lemma:stepcEOSEOS}
    Let $e=\tup{t,\theta}\in\Theta$. If $m$ and $m'$ are configurations of $\os$ and $k$ is a configuration of $\os^\cons$ over the trash places such that $m\rightarrow^{e^\cons} m'+k$ in $\os^\cons$, then there is some configuration $m''$ of $\os$ such that $m\geq_f m '' \rightarrow^{t^\cons} m'$ in $\os$.
\end{lemma}
\begin{proof}
If $\dest(t)=\emptyset$, then $t^\cons$ and $t$ enforce the same dynamics. Thus, the statement trivially holds for $k=0$. 
Otherwise, $t$ is involved in a system autonomous event, i.e., $\theta=\varepsilon$.

We now define a mode $(\lambda'',\rho'')$ for $e$ by taking advantage of the following notation. Given a configuration $\mu=\sum_{i=1}^k\tup{p_i,M_i}$, we denote by $\overline{\mu}$ the configuration $\overline{\mu}=\sum_{i=1}^k\tup{p_i,\overline{M_i}}$ where 
\[\overline{M_i}=\begin{cases}
    0&\text{if }d(p_i)\in\dest(t)\\
    1&\text{otherwise}
\end{cases}\]
For each configuration $\mu$, we have $\mu\geq_o\overline{\mu}$. We set $\lambda''=\overline{\lambda}$.
Let $\rho''$ and $k''$ be the unique configuration of $\os$ and configuration of $\os^\cons$ over the trash places such that $\rho=\rho''+k''$. Note that, since $m'+k=m-\lambda+\rho''+k''$ and $m'$, $m$, $\lambda$, and $\rho''$ are configurations of $\os$,\footnote{Recall that the pre-conditions of $t^\cons$ do not involve the trash places.} while $k$ and $k''$ are over the trash places, $k''=k$. Moreover, for each place $p\in\hat{P}$ such that $d(p)\in\dest(t)$, we have $\rho''(p)=\rho(p)=0$.

We now show that $(\lambda'',\rho'')$ enables $e$ over $m''=m-\lambda+\lambda''$ in $\os$. Since $(\lambda,\rho)$ enables $e$ over $m$ in $\os^\cons$:
\begin{itemize}
    \item $\Pi^1(\lambda'')=\Pi^1(\lambda)=\prefun_{\hat{N}^\cons}(t)(t)=\prefun_{\hat{N}}(t)$.
    \item For each $p\in\hat{P}$, we have $\Pi^1(\rho'')(p)=\Pi^1(\rho)(p)-\Pi^1(k'')(p)=\Pi^1(\rho)(p)=\postfun_{\hat{N}^\cons}(t)(p)=\postfun_{\hat{N}}(t)(p)$
    \item If $N\in\dest(N)$, then $\Pi^2_N(\rho'')=0=\Pi^2_N(\lambda'')$.
    \item If $N\notin\dest(N)$, then $\Pi^2_N(\rho'')=\Pi_N(\rho)=\Pi^2_N(\lambda)=\Pi^2_N(\lambda'')$.
\end{itemize}
Moreover, $m''\geq_s m-\lambda+\lambda''\geq_s\lambda''$. Thus, $(\lambda'',\rho'')$ enables $e$ over $m''$ in $\os$.

Finally, $m''\rightarrow^{e,\lambda'',\rho''} m'$ because, for each place $p\in\hat{P}$, we have $m''(p)-\lambda''(p)+\rho''(p)= m(p)-\lambda(p)+\overline{\lambda}(p)-\overline{\lambda}(p)+\rho(p)=m(p)+\lambda(p)-\rho(p)=m'(p)$.

\end{proof}

For each transition $t^\cons$ of $\os^\cons$ and $N\in\N$, we have that $F(\text{trash}_N,t)=0$. Thus, by EOS semantics, for each configuration $k$ of $\os^\cons$ over the trash places, we conclude that, if $m+k\rightarrow m'+k'$ in $\os^\cons$ for some configurations $m$, $m'$, and $k'$ as in the statement of Lemma~\ref{lemma:stepcEOSEOS}, then $k'\geq_s k$ and $m\rightarrow m'+ k' - k$. We can now apply Lemma~\ref{lemma:stepcEOSEOS} and obtain again that there is some configuration $m''$ of $\os$ such that $m\geq_f m'' \rightarrow m'$. This fact can be generalized to the transitive closure of $\rightarrow$ in $\os^\cons$, which yields that if a configuration $m$ of $\os$ is reachable in $\os^\cons$ from a configuration $\iota$ of $\os$, then it is $(\leq_f,\omega)$-reachable in $\os$.

\begin{theorem}\label{theorem:equivalencecEOSEOS}
    If $\iota$ and $\tau$ are configurations of $\os$ and $\tau$ is $\leq_f$-coverable in $\os^\cons$, then $\tau$ is $(\leq_f,\omega)$-coverable in $\os$.
\end{theorem}

The combination of Theorem~\ref{theorem:equivalenceEOScEOS} and Theorem~\ref{theorem:equivalencecEOSEOS} allows us to solve any instance $(\iota,\tau,\os)$ of $(\leq_f,\omega)$-coverability by solving the instance $(\iota, \tau,\os^\cons)$ of $\leq_f$-coverability. This provides the desired reduction, which is polynomial because the construction of $\os^\cons$ is polynomial.
\begin{theorem}\label{thm:equivalenceReduction}
    There is a polynomial reduction from EOS $(\leq_f,\omega)$-coverability to cEOS $\leq_f$-coverability.
\end{theorem}
 
\newcommand{\channelId}[1]{\text{Ch}_{#1}}
\newcommand{\transform}{\mathcal{T}}

\section{Proof of \texorpdfstring{\rnPN}{rnPN} capturing \texorpdfstring{\GnPN}{GnPN}} \label{app:GnPN-rnPN}
This section shows that \GnPN can be simulated using \rnPN, i.e., given a \GnPN \(N\), there is a polynomial-time constructible \rnPN \(N'\) instance that simulates \(N\).

For \(n_1,n_2 \in \mathbb{N}\) where \(n_1 \leq n_2\), \([n_1\cdots n_2]\) is the set \(\{n_1, \cdots,n_2\}\). 

We give a formal definition of \rnPN below:
We now formally define the syntax of \rnPN.
\begin{definition}\label{def:rnpns}
    A \rnPN is a tuple $N=(P,T,F,T',F',R')$ where:
    \begin{enumerate}
        \item $(P,T,F)$ and $(P,T',F')$ are \nPNs where $T\cap T'=\emptyset$.
        \item $R':T'\rightarrow (P\times P)$. For each $t'\in T'$, we denote $R'(t')$ by $(r^{t'}_1,r^{t'}_2)$.\footnote{$R$ stands for \textit{rename} and is used to select the places $p_1$ and $p_6$, in Fig.~\ref{fig:selectiveTransfer}.}
        \item For each $t'\in T'$, $\var(t)=\{x_0,x_1,x_2\}$ and there is a set $P_{t'}=\subseteq P\setminus\{r^{t'}_1,r^{t'}_2\}$ of four distinct places, say  $p_2,p_3,p_4,p_5\in P$, such that, for each $q\in P$:

\begin{scriptsize}
\begin{align*}
F'_{x_0}(q,t')=\begin{cases}
1 &\text{if }q=p_2\\
0 &\text{otherwise}
\end{cases}
\,\,&\,
F'_{x_1}(q,t')=\begin{cases}
1 &\text{if }q=p_3\\
0 &\text{otherwise}
\end{cases}
\,&\,
F'_{x_2}(q,t')=\begin{cases}
1 &\text{if }q=p_4\\
0 &\text{otherwise}
\end{cases}
\\         
F'_{x_0}(t',q)=\begin{cases}
1 &\text{if }q=p_5\\
0 &\text{otherwise}
\end{cases}
\,\,&\,
F'_{x_1}(t',q)=\begin{cases}
1 &\text{if }q=p_3\\
0 &\text{otherwise}
\end{cases}   
\,&\,
F'_{x_2}(t',q)=\begin{cases}
1 &\text{if }q=p_4\\
0 &\text{otherwise}
\end{cases}   
\end{align*}
\end{scriptsize}
    \end{enumerate}
\end{definition}

\begin{definition}
    A configuration of a \rnPN $N=(P,T,F,T',F',R')$ is a configuration for $(P,T,F)$. A transition $t' \in T'$ is \textit{enabled} on a configuration $M=\fmset{m_1,\dots,m_\ell}$ with \nPN mode $e':\var(t')\rightarrow\{1,\dots,\ell\}$ if $t'$ is both enabled on $M$ with mode $e'$ as a transition of $(P,T',F')$.\footnote{I.e., if it is enabled according to $F'$, disregarding $R'$.}
In this case, its \textit{firing} in $N$ returns the configuration $M'$, denoted by $M\rightarrow^{t',e'} M'$, such that $M\rightarrow^{t',e'}M''$ in $(P,T',F')$ and $M'$ is obtained from $M''$ by firing the transfers:
    \begin{align*}
        M'=M'' - \fmset{m_{e'(x_1)},m_{e'(x_2)}}
    &+ \fmset{ m_{e'(x_1)}-m_{e'(x_1)}(r^{t'}_1)\delta_{r^{t'}_1}}\\
    &+ \fmset{m_{e'(x_2)}+m_{e'(x_1)}(r^{t'}_1)\delta_{r^{t'}_2}}
    \end{align*}
    For $t\in T\cup T'$, we write $M\rightarrow^t M'$ in $N$ if there is some mode $e$ such that either $t\in T$ and $M\rightarrow^{t,e}M'$ in $(P,T,F)$, or $t\in T'$ and $M\rightarrow^{t,e}M'$ in $N$.
\end{definition}

\begin{example}
The \rnPN configuration in Fig.~\ref{fig:selectiveTransferPost} is reached from the configuration in Fig.~\ref{fig:selectiveTransferPre} by firing $t$ with mode $e'$ such that $e'(x_0)=\tup{1,2,0,2,0,0}$, $e'(x_1)=\tup{2,0,1,0,0,0}$, and $e'(x_2)=\tup{0,0,0,1,0,0,0}$, i.e., graphically, $e'$ binds  $x_0$ to $c$, $x_1$ to $a$, and $x_2$ to $d$. Note that $e'$ enables the special transition $t$ since $e'(x_0)(p_2)=2\geq 1 = F_{x_0}(p_2,t)$, $e'(x_1)(p_3)=1= F_{x_1}(p_3,t)$, and $e'(x_2)(p_4)=1=F_{x_2}(p_4,t)$, and for all other combinations of variable $x\in\var\{t\}$ and place $q$, $F_x(q,t)=0$.
\end{example}

\mainGnPNrnPN*
\paragraph*{Proof Sketch}Given a \GnPN \(N = (P, T, F, G)\), we will construct a \rnPN \(\bar{N} = (\bar{P}, \bar{T}, \bar{F}, \bar{T}', \bar{F}', \bar{R}')\) in polynomial time that simulates \(N\).

Take any two configurations \(M_i, M_j\) in \(N\), such that for some \(t \in T\) and mode \(e\), 
\(M_i \rightarrow^{e,t} M_j\). 
The aim is to show that, for the configurations that correspond to \(M_i,M_j\) 
(i.e., \(\transform(M_i), \transform(M_j)\) by Def.~\ref{def:transform}) 
\(\transform(M_i) \rightarrow^\ast \transform(M_j)\) and vice-versa.

\begin{definition} \label{def:transform}
    Given a configuration \(M = \fmset{m_1,m_2,\cdots,m_{\abs{M}}}\) in \GnPN \(N = (P,T,F,G)\), we define a transformation function \(\transform: (\mathbb{N}^{\abs{P}})^\oplus \rightarrow (\mathbb{N}^{\abs{\bar{P}}})^\oplus\) that transforms \(M\) to a corresponding configuration \(\transform(M)\) in \rnPN \(\bar{N} = (\bar{P},\bar{T},\bar{F},\bar{T'},\bar{F}',\bar{R}')\) such that \(\transform(M) = \fmset{\bar{m}_1,\bar{m}_2,\cdots,\bar{m}_{\abs{M}},\bar{m}_{cs}}\) in \rnPN \(\bar{N}\) where for all \(i \in [\abs{M}]\), 
    \begin{enumerate*}
        \item for all \(p \in [|P|]\), \(\bar{m}_i{\downarrow_p} = m_i{\downarrow_p}\),
        \item for all \(p' \in [|P|+1 \cdots \bar{P}]\), \(\bar{m}_i{\downarrow_{p'}} = 0\)
        \item \(\bar{m}_{cs}{\downarrow_{|\bar{P}|}} = 1\) and 
        \item for all \(p'' \in [|\bar{P}|-1]\), \(\bar{m}_{cs}{\downarrow_{p''}} = 0\) where \(m\downarrow_{i}\) denotes the \(i^{th}\) component of the tuple \(m\).
    \end{enumerate*}
      
\end{definition}

Intuitively, we show that firing \(t \in T\) is simulated by firing a sequence of transitions in \(\bar{T}\). \(q_{select} \in \bar{P}\) is a special place that will choose which transition \(t \in T\) to simulate in \(\bar{N}\) (Notice how we use \(q_{select}\) as an indicator in order to correspond to a configuration in \(N\) in Def.~\ref{def:transform}). The set of selective transfers defined by \(G_t\) will be decomposed into one selective transfer per transition in \(\bar{N}\), such that each resulting transition belongs to \(\bar{T}'\). Before performing these selective transfers sequentially, for each transition \(t \in \bar{T}'\), we provide a gadget—named \emph{move block}—that extracts all the tokens from the place \(p \in P\) of the tuple instantiated by \(x_i\) and puts them in an intermediate place \(\bar{p} \in \bar{P}\) if there is a selective transfer depicted by \(G_t(x_i,x_j)[p,q]=1\) for some $x_j \in \X(t)$, $q \in P$. This is important because there might be cyclic dependencies like, from $p$ to $q$ we have a selective transfer, and then one from $q$ to $p$, all of them being labeled with the same variable. This will result in a swapping of tokens between $p$ and $q$ in \GnPN which needs to be captured in \rnPN. If we don't put the tokens in intermediate places and try to do the selective transfer sequentially then we will be adding tokens and then transferring instead of swapping. One other way would be to produce fresh tuples and use them as intermediate labels to do the transfers and then relabel them to simulate swapping. 

The sequence of transitions begin with the preconditions (defined by the gadget \emph{pre block}), followed by all the selective transfers with or without renaming (in an order determined by the indexing function in \cref{def:channelId} defined by the gadget \emph{rename block} and \emph{no-rename block}), and finally the post-conditions(defined by the gadget \emph{post block}).

\begin{definition} \label{def:channelId}
    Given a \GnPN \(N = (P,T,F,G)\), for each \(t\in T\), we define 
    \[k_t = \sum\limits_{\substack{x_i,x_j \in \X(t) \\ x_i \neq x_j}} \sum\limits_{\substack{p,q \in P \\ p \neq q}} G_t(x_i,x_j)[p,q]\]
    to be the number of selective transfers and a partial injective function \(\channelId{t}: (\X \times \X \times P \times P) \rightarrow [k_t]\) such that \(\channelId{t}(x_i,x_j,p,q)=i\) if \(G_t(x_i,x_j)[p] = \delta_q\) that is \(\channelId{t}\) maps each selective transfer operation to a unique index $i$. 
\end{definition}

To enforce this sequential execution, we introduce a sequence of places named using the letter \(q\)—referred to as the \emph{control sequence}—that enables firing these decomposed transitions in the desired order. Moreover, the control sequence ensures that all other transitions not involved in simulating \(t\) remain disabled during this process. As a result, it guarantees a deterministic sequence of transition firings that faithfully simulate the firing of \(t\). In fact, we have a dedicated tuple \(\bar{m}_{cs}\) for the control sequence (see Def.~\ref{def:transform}), which will always be captured using a special variable \(x_{cs}\) by construction. The construction is such that the tuple has a token precisely in one place at any point of the execution(this simulates a DFA that sequentializes the transitions).

\begin{figure}[t]
    \begin{subfigure}{0.5\linewidth}
        \centering
        \scalebox{0.7}
{\begin{tikzpicture}

\node[triangle,label={[name=qselLab]above:\scriptsize $q_{select}$}](qsel)at (-2.2,2){};

\path [rectangle, fill=gray!10](-2,-1) to (-1,-1) to (-1,1.2) to (-2,1.2);
\node[label={[name=preLab]above:\scriptsize \tiny $\text{pre block}$}](preB)at(-1.5,1){};
\node[empty] (predot)at (-1.5,-0.5){\scriptsize \tiny $\vdots$};
\node[transvert,label={[name=tpreLab]above:\scriptsize  $\tau_t^{pre}$}] (tpre)at (-1.5,0){};
\node[triangle,label={[name=qtLab]above:\scriptsize $q_{t}^{move}$}](qtmove)at (-0.5,0){};
\path [rectangle, fill=gray!10](0,-1) to (1,-1) to (1,1.2) to (0,1.2);
\node[label={[name=preLab]above:\scriptsize \tiny $\text{move block}$}](moveB)at(0.5,1){};
\node[empty] (empty)at (0.5,0.5){};

\node[triangle,label={[name=qt1Lab]above:\scriptsize $q_{t}^{1}$}](qt1)at (1.5,0){};
\path [rectangle, fill=gray!10](2,-1) to (3,-1) to (3,1.2) to (2,1.2);
\node[label={[name=tBlockLab]above:\scriptsize \tiny $\text{transfer block 1}$}](tBlock)at(2.5,1){};
\node[transvert,label={[name=t1Lab]above:\scriptsize  $\tau_t^{1}$}] (t1)at (2.5,0){};
\node[empty] (tdots)at (2.5,-0.5){\scriptsize \tiny $\vdots$};
\node[empty] (dots)at (3.5,0){\scriptsize \tiny $\cdots$};
\node[triangle,label={[name=qtktLab]above:\scriptsize $q_{t}^{k_t}$}](qtkt)at (4.5,0){};
\path [rectangle, fill=gray!10](5,-1) to (6,-1) to (6,1.2) to (5,1.2);
\node[label={[name=tktBlockLab]above:\scriptsize \tiny $\text{transfer block $k_t$}$}](tktBlock)at(5.5,1){};
\node[transvert,label={[name=tktLab]above:\scriptsize  $\tau_{t}^{k_t}$}] (tkt)at (5.5,0){};
\node[empty] (tktdots)at (5.5,-0.5){\scriptsize \tiny $\vdots$};
\node[triangle,label={[name=qtpostLab]above:\scriptsize $q_{t}^{post}$}](qtpost)at (6.5,0){};
\path [rectangle, fill=gray!10](7,-1) to (8,-1) to (8 ,1.2) to (7,1.2);
\node[transvert,label={[name=tpostLab]above:\scriptsize  $\tau_{t}^{post}$}] (tpost)at (7.5,0){};
\node[label={[name=postLab]above:\scriptsize \tiny $\text{post block}$}](postB)at(7.5,1){};
\draw[->] (qsel) |- node [midway,below] {$x_{cs}$} (tpre);
\draw[->] (tpre)-- node[midway,below]{$x_{cs}$} (qtmove);
\draw[->] (qtmove)-- node[midway,below]{$x_{cs}$} ($(empty)-(0.2,0.5)$);
\draw[->] ($(empty)+(0.2,-0.5)$)-- node[midway,below]{$x_{cs}$} (qt1);
\draw[->] (qt1)-- node[midway,below]{$x_{cs}$} (t1);
\draw[->] (t1)-- node[midway,below]{$x_{cs}$} (dots);
\draw[->] (dots)-- node[midway,below]{$x_{cs}$} (qtkt);
\draw[->] (qtkt)-- node[midway,below]{$x_{cs}$} (tkt);
\draw[->] (tkt)-- node[midway,below]{$x_{cs}$} (qtpost);
\draw[->] (qtpost)-- node[midway,below]{$x_{cs}$} (tpost);
\draw[->] ($(tpost.east)$)--($(tpost.east)+(0.6,0)$) |- node[pos=0.8,above]{$x_{cs}$}(qsel);

\end{tikzpicture}
}        
        \caption{Control Sequence}
        \label{fig:CNuPN_cs}
    \end{subfigure}
    \begin{subfigure}{0.4\linewidth}
        \centering
        \scalebox{0.7}{\begin{tikzpicture}

\node[place,label={[name=p1Lab]above:\scriptsize $p_{pre(1)}$}](p1)at (0,0){};
\node[empty](txt) at (0,-0.5){$\vdots$};
\node[place,label={[name=pnLab]below:\scriptsize $p_n$}](pn)at (0.5,-1){};

\node[transvert,label={[name=tLab]above:\scriptsize  $\tau_t^{pre}$}] (t)at (1.5,0){};

\node[place,label={[name=px1Lab]above:\scriptsize $p_{x_1}$}](px1)at (3,0){};
\node[empty](txt) at (3,-0.5){$\vdots$};
\node[place,label={[name=pxtLab]below:\scriptsize $p_{x_{\lvert \X(t)\rvert}}$}](pxt)at (2.5,-1){};

\node[triangle,label={[name=qselLab]above:\scriptsize $q_{select}$}](qsel)at (0.5,1.5){};

\node[triangle,label={[name=qt1Lab]above:\scriptsize $q_{t}^{move}$}](qt1)at (2.5,1.5){};

\draw[->] (p1) to node[above,midway,sloped] {\scriptsize $m_1$}(t);

\draw[->] (pn) to node[above,midway,sloped] {\scriptsize $m_n$}($(t.west)-(0,0.2)$);

\draw[->](t.east) to node[above,midway] {\scriptsize$x_1$}(px1);
\draw[->]($(t.east)-(0,0.2)$) to node[above,midway,sloped] {\scriptsize$x_{\lvert \X(t)\rvert}$}(pxt);

\draw[->](qsel) to node[below,midway,sloped] {\scriptsize$x_{cs}$}(t);
\draw[->](t) to node[below,midway,sloped] {\scriptsize$x_{cs}$}(qt1);

\end{tikzpicture}
}
        \caption{pre<$t$> block}
        \label{fig:CNuPN_pre}
    \end{subfigure}
    \begin{subfigure}{0.5\linewidth}
        \centering
        \scalebox{0.7}{\begin{tikzpicture}

\node[place,label={[name=pLab]left:\scriptsize $\bar{p}$}](p)at (0,0){};

\node[triangle,label={[name=qtiLab]above:\scriptsize $q_{t}^{i}$}](qti)at (0,1.5){};
\node[place,label={[name=pxLab]above:\scriptsize $p_x$}](px)at (1,1.75){};
\node[place,label={[name=pyLab]above:\scriptsize $p_y$}](py)at (2,1.75){};
\node[transvert,label={[name=tLab]below:\scriptsize  $\tau_t^{i}$}] (t)at (1.5,0){};
\node[place,label={[name=qLab]right:\scriptsize $q$}](q)at (3,0){};
\node[triangle,label={[name=qti1Lab]above:\scriptsize $q_{t}^{i+1}$}](qti1)at (3,1.5){};

\draw[->,myDouble] (p) to node[below,midway,sloped] {\scriptsize $x$}(t);
\draw[->](qti) to node[below,midway,sloped] {\scriptsize$x_{cs}$}(t);
\draw[<->] (px) to node[below,midway,sloped] {\scriptsize $x$}(t);
\draw[<->] (py) to node[above,midway,sloped] {\scriptsize $y$}(t);
\draw[->](t) to node[below,midway,sloped] {\scriptsize$x_{cs}$}(qti1);
\draw[->,myDouble](t.east) to node[below,midway] {\scriptsize$y$}(q);
\end{tikzpicture}
}
        \caption{rename<$i,t$> block}
        \label{fig:CNuPN_rename}
    \end{subfigure}
    \begin{subfigure}{0.5\linewidth}
        \centering
        \scalebox{0.7}{\begin{tikzpicture}

\node[place,label={[name=px1Lab]above:\scriptsize $p_{1}$}](px1)at (0,0){};
\node[place,label={[name=pxtLab]below:\scriptsize $p_{x_{\lvert \X(t)\rvert}}$}](pxt)at (0.5,-1.5){};

\node[transvert,label={[name=tLab]above:\scriptsize  $\tau_t^{post}$}] (t)at (1.5,0){};

\node[place,label={[name=p1Lab]above:\scriptsize $p_{post(1)}$}](p1)at (3,0){};
\node[place,label={[name=pnLab]below:\scriptsize $p_{post(n)}$}](pn)at (2.5,-1.5){};

\node[triangle,label={[name=qpostLab]above:\scriptsize $q_{post}$}](qpost)at (0.5,1.5){};
\node[triangle,label={[name=qselLab]above:\scriptsize $q_{select}^1$}](qsel)at (2.5,1.5){};
\draw[->] (px1) to node[above,midway] {\scriptsize $x_1$}(t);
\draw[] (px1) to node[below,midway] {\scriptsize $\vdots$}(t);
\draw[->] (pxt) to node[below,midway,sloped] {\scriptsize $x_{\lvert \X(t)\rvert}$}(t);
\draw[->](t.east) to node[above,midway] {\scriptsize$x_1$}(p1);
\draw[](t.east) to node[below,midway] {\scriptsize$\vdots$}(p1);
\draw[->](t) to node[below,midway,sloped] {\scriptsize$x_n$}(pn);
\draw[->](qpost) to node[below,midway,sloped] {\scriptsize$x_{cs}$}(t);
\draw[->](t) to node[below,midway,sloped] {\scriptsize$x_{cs}$}(qsel);

\end{tikzpicture}
}       
        \caption{post<$t$> block}
        \label{fig:CNuPN_post}
    \end{subfigure}
    
    \begin{subfigure}{0.5\linewidth}
         \centering
         \scalebox{0.7}{\begin{tikzpicture}

\node[place,label={[name=pLab]left:\scriptsize $\bar{p}$}](p)at (0,0){};

\node[triangle,label={[name=qtiLab]above:\scriptsize $q_{t}^{i}$}](qti)at (0,1.5){};
\node[place,label={[name=pxLab]above:\scriptsize $p_x$}](px)at (1,1.75){};
\node[place,label={[name=pyLab]above:\scriptsize $p_z$}](py)at (2,1.75){};
\node[transvert,label={[name=tLab]below:\scriptsize  $\bar{\tau}_t^{i}$}] (t)at (1.5,0){};
\node[place,label={[name=qLab]below:\scriptsize $\bar{p}^{i}_{t}$}](pti)at (3,0){};
\node[triangle,label={[name=qti1Lab]above:\scriptsize $q_{t}^{-i}$}](qt-i)at (3,1.5){};

\draw[->,myDouble] (p) to node[below,midway,sloped] {\scriptsize $x$}(t);
\draw[->](qti) to node[below,midway,sloped] {\scriptsize$x_{cs}$}(t);
\draw[<->] (px) to node[below,midway,sloped] {\scriptsize $x$}(t);
\draw[<->] (py) to node[above,midway,sloped] {\scriptsize $z$}(t);
\draw[->](t) to node[below,midway,sloped] {\scriptsize$x_{cs}$}(qti1);
\draw[->,myDouble](t.east) to node[below,midway] {\scriptsize$z$}(pti);

\node[transvert,label={[name=tLab]below:\scriptsize  ${\tau}_t^{i}$}] (bart)at (4.5,0){};
\draw[->,myDouble](pti.east) to node[below,midway] {\scriptsize$z$}(bart);
\draw[->](qt-i) to node[below,midway,sloped] {\scriptsize$x_{cs}$}(bart);

\node[place,label={[name=pxLab]above:\scriptsize $p_x$}](pxr)at (4,1.75){};
\node[place,label={[name=pyLab]above:\scriptsize $p_z$}](pyr)at (5,1.75){};
\node[triangle,label={[name=qti1Lab]above:\scriptsize $q_{t}^{i+1}$}](qti1)at (6,1.5){};
\draw[<->] (pxr) to node[below,midway,sloped] {\scriptsize $x$}(bart);
\draw[<->] (pyr) to node[above,midway,sloped] {\scriptsize $z$}(bart);
\draw[->](bart) to node[below,midway,sloped] {\scriptsize$x_{cs}$}(qti1);
\node[place,label={[name=qLab]below:\scriptsize $q$}](q)at (6,0){};
\draw[->,myDouble](bart) to node[below,midway] {\scriptsize$z$}(q);
\end{tikzpicture}
}
         \caption{no-rename<$i,t$> Block }
        \label{fig:norename}
    \end{subfigure}
    \begin{subfigure}{0.5\linewidth}
          \centering
         \scalebox{0.7}{\begin{tikzpicture}
\node[place,label={[name=pLab]above:\scriptsize $\bar{p}$}](p)at (0,0){};
\node[transvert,label={[name=tLab]above:\scriptsize  $\tau_t^{i}$}] (t)at (1.5,0){};
\node[place,label={[name=qLab]above:\scriptsize $\bar{p}^{i}_{t}$}](pti)at (3,0){};
\node[transvert,label={[name=tLab]above:\scriptsize  $\bar{\tau}_t^{i}$}] (bart)at (4.5,0){};
\node[place,label={[name=qLab]above:\scriptsize $q$}](q)at (6,0){};
\node[triangle,label={[name=qtiLab]above:\scriptsize $q_{t}^{i}$}](qti)at (0,1){};
\node[triangle,label={[name=qti1Lab]above:\scriptsize $q_{t}^{-i}$}](qt-i)at (3,1){};
\node[triangle,label={[name=qti1Lab]above:\scriptsize $q_{t}^{i+1}$}](qti1)at (6,1){};
\draw[->,myDouble] (p) to node[above,midway,sloped] {\scriptsize $x$}(t);
\draw[->,myDouble](t.east) to node[above,midway] {\scriptsize$z$}(pti);
\draw[->,myDouble](pti.east) to node[above,midway] {\scriptsize$z$}(bart);
\draw[->,myDouble](bart) to node[above,midway] {\scriptsize$z$}(q);
\draw[->](qti) to node[above,midway,sloped] {\scriptsize$x_{cs}$}(t);
\draw[->](t) to node[above,midway,sloped] {\scriptsize$x_{cs}$}(qt-i);
\draw[->](qt-i) to node[above,midway,sloped] {\scriptsize$x_{cs}$}(bart);
\draw[->](bart) to node[above,midway,sloped] {\scriptsize$x_{cs}$}(qti1);
\node[place,label={[name=pxLab]below:\scriptsize $p_x$}](px)at (2,-1.75){};
\node[place,label={[name=pyLab]below:\scriptsize $p_z$}](pz)at (4,-1.75){};

\draw[<->] (px) to node[below,midway,sloped] {\scriptsize $x$}(t);
\draw[<->] (px) to node[below,near start,sloped] {\scriptsize $x$}(bart);
\draw[<->] (pz) to node[above,near start,sloped] {\scriptsize $z$}(t);
\draw[<->] (pz) to node[above,midway,sloped] {\scriptsize $z$}(bart);
\end{tikzpicture}
}
         \caption{move<$i,t$> Block}
            \label{fig:enter-label}
    \end{subfigure}
    \caption{Schematic diagrams representing the simulation of \GnPN by \rnPN} 
    \label{fig:CNuPN_ex}
\end{figure}

We provide a schematic diagram of the construction in \cref{fig:CNuPN_ex}. We assume that for each transition \(t \in T\), \(\X(t) = \{x_1,x_2,\cdots,x_{\abs{\X(t)}}\}\) without loss of generality. From now on, for a transition \(t\), the function \(\bar{F}\) returns 0, w.r.t. any place or variable for which \(\bar{F}\) is not defined.
For each transition \(t \in T\), we define the following gadgets (refer to the diagrams provided in \cref{fig:CNuPN_ex}):

\textbf{Pre-block<t>} make use of the the places \(\{p \mid p \in \prefun(t)\}\), auxiliary places \(\{q_{select},q^{move}_t\} \cup \{p_{x_i} \mid i \in \abs{\X(t)}\}\), and the transition \(\tau^{pre}_t\). When \(\tau^{pre}_t\) is fired, it does the following:
\begin{enumerate*}
    \item Move a token from \(q_{select}\) to \(q^{move}_t\), i.e., \(\bar{F}_{x_{cs}}(q_{select},\tau^{pre}_t) = \bar{F}_{x_{cs}}(\tau^{pre}_t,q^{move}_t) = 1\)
    \item simulates the firing of the pre-conditions of \(t \in T\), i.e., for all \(x_i \in \X(N)\),
    $\bar{F}_{x_i}(p,\tau^{pre}_t) = F_{x_i}(p,t) \text{ if } p \in P \\ $
    \item put a token in the tuple instantiated by the variable \(x_i \in \X(t)\) to the place \(p_{x_i}\), i.e., for all \(x_i \in \X(t)\), \(\bar{F}_{x_i}(\tau^{pre}_t,p_{x_i}) = 1\). This amounts to selecting the mode for the transition \(t\).
\end{enumerate*}

\textbf{Post-block<t>} uses the places \(\{p \mid p \in \postfun(t)\}\), auxiliary places \(\{q_{select},q^{post}_t\} \cup \{p_{x_i} \mid i \in \abs{\X(t)}\}\), and the transition \(\tau^{post}_t\). When \(\tau^{post}_t\) is fired, it does the following:
\begin{enumerate*}
    \item Move a token from \(q^{post}_t\) to \(q_{select}\), i.e., \(\bar{F}_{x_{cs}}(q^{post}_t,\tau^{post}_t) = \bar{F}_{x_{cs}}(\tau^{post}_t,q_{select}) = 1\)
    \item simulates the firing of the post-conditions of \(t \in T\), i.e., for all \(x_i \in \X(N)\),
    $\bar{F}_{x_i}(\tau^{post}_t,p) = F_{x_i}(t,p) \text{ if } p \in P $
    \item remove the token in the tuple instantiated by the variable \(x_i \in \X(t)\) from the place \(p_{x_i}\), i.e., for all \(x_i \in \X(t)\), \(\bar{F}_{x_i}(p_{x_i},\tau^{pre}_t) = 1\).
\end{enumerate*}

For each \(i = \channelId{t}(x,y,p,q)\), we define a \textbf{move<i,t> block} that uses the place \(p\) from \(P\), auxiliary places \(\bar{p},\bar{p}^t_i,q^t_{i-1},\bar{q}_{i}^t,p_x,p_z,q_{i}^t\) for some \(z \in \X(t)\setminus \{x\}\), and the transitions \(\tau^t_i,\bar{\tau}^t_i\). Sequentially firing the transitions \(\tau^t_i,\bar{\tau}^t_i\) results in the following:
\begin{enumerate*}
    \item Moves a token from \(q_{i-1}^t\) to \(q_{i}^t\)
    \item instantiates \(x\) and \(z\) by the tuples having a token in \(p_x\) and \(p_z\) respectively when executing both the transitions
    \item the tokens in the place \(p\) of the tuple instantiated by the variable \(x\) are moved to \(\bar{p}\) in the same tuple.
\end{enumerate*} 

For each \(i = \channelId{t}(x,y,p,q)\),
\begin{enumerate}
    \item if \(x \neq y\), then we define a \textbf{rename<i,t> block} that uses the place \(q\) from \(P\), auxiliary places \(\bar{p},q^i_t,p_x,p_y,q^{i+1}_t\), and the transition \(\tau^i_t\). When \(\tau^i_t\) is fired, it performs the following:
    \begin{enumerate*}
        \item Moves a token from \(q^{i}_t\) to \(q^{i+1}_t\), i.e., \(\bar{F}_{x_{cs}}(q^{i}_t,\tau^{i}_t) = \bar{F}_{x_{cs}}(\tau^{i}_t,q^{i+1}_{t}) = 1\)
        \item instantiates \(x\) and \(y\) by the tuples having a token in \(p_x\) and \(p_y\) respectively
        \item performs a selective transfer of the tokens from the tuple instantiated by \(x\) in the place \(p\) to the place \(q\) in the tuple instantiated by \(y\).
    \end{enumerate*}

    \item if \(x=y\), then we define a \textbf{no-rename<i,t> block} which has the same underlying structure as move<i> block with the places and transitions relabeled as \(p \mapsto \bar{p}\), \(\bar{p} \mapsto q\), \(q^t_{i-1} \mapsto q^{i}_t\), \(q^t_i \mapsto q^{i+1}_t\), \(\bar{q}^t_{i} \mapsto \bar{q}^{i}_t\), \(\bar{p}^t_i \mapsto \bar{p}^{i}_t\), \(\bar{p}^t_i \mapsto \bar{p}^{i}_t\), \(p_x \mapsto p_x\), \(p_z \mapsto p_z\), \(\tau^t_i \mapsto \bar{\tau}^{i}_t\) and \(\bar{\tau}^t_i \mapsto \tau^{i}_t\).
\end{enumerate}
Notice in Fig.~\ref{fig:CNuPN_cs}, how the individual blocks sequentially arrange themselves thanks to the control sequence. We will write <block-name><i> for <block-name><$i,t$>, when \(t\) is clear from the context. To simulate \(M_i \rightarrow^{t,e} M_j\) in the \GnPN \(N\) using the \rnPN \(\bar{N}\), that we constructed above, \(\T(M_i)\), \(q_{select}\) selects the transition \(\tau^{pre}_t\) to execute the pre<t> block. 

Note that we will use the mode \(\bar{e}\) defined as \(\bar{e}(x) = \begin{cases}
    e(x) &\text{ if }x \in \X(t)\\
    m_{cs}  &\text{ if }x=x_{cs}
\end{cases}\) for all the transitions fired in \(\bar{N}\) to simulate \(t\). 

After firing the pre<t> block, the control sequence ensures execution of the following blocks in sequence:
\begin{enumerate}
    \item \label{item:move}for all \(i = \channelId{t}(x,y,p,q)\), move<i> block 
    \item \label{item:transfer}for all \(i = \channelId{t}(x,y,p,q)\), if \(x\neq y\), rename<i> block else norename <i> block
\end{enumerate}
Then, the control sequence fires the post<t> block which essentially yields \(\T(M_j)\).

Note that for some \(i=\channelId{t}(x,y,p,q)\), there are two effects on the marking \(M_i\), \begin{enumerate*}
    \item The value at the position corresponding to \(p\) on the tuple selected by \(e(x)\) becomes \(0\)
    \item The value at the position corresponding to \(q\) on the tuple selected by \(e(y)\) is added with the value at the position corresponding to \(p\) on the tuple selected by \(e(x)\)
\end{enumerate*} 
In fact, we get the same resultant effect on the marking \(\T(M_i)\) after firing the blocks referred in item \ref{item:move} and \ref{item:transfer}.

\begin{lemma} \label{lemma:GnPn=rnPn}
    Given two configurations \(M_i\) and \(M_j\) in \(N\), for some transition \(t \in T\) and mode \(e\), \(M_i \rightarrow^{t,e} M_j\) iff for the corresponding configurations \(\transform(M_i)\) and \(\transform(M_j)\) in \(\bar{N}\), there is a sequence of transitions \(t_1,t_2,\cdots,t_\ell\) and modes \(e_1,e_2,\cdots,e_\ell\) such that \(\transform(M_i) = \bar{M}_0 \rightarrow^{t_1,e_1} \bar{M}_1 \rightarrow^{t_2,e_2} \cdots \rightarrow^{t_\ell,e_\ell} \bar{M}_\ell = \transform(M_j)\).
\end{lemma}

\paragraph*{Proof of \cref{lemma:GnPn=rnPn}}
Given \GnPN \(N = (P,T,F,G)\), two configurations \(M_i\) and \(M_j\) such that \(M_i \rightarrow^{t,e} M_j\) for some transition \(t\) and mode \(e\) . We have 
\(M_i = M'' + \sum\limits_{x \in \X(t)} \fmset{m_{e(x)}} = \fmset{m_1,m_2,\cdots,m_{\abs{M_i}}}\) and \(M_j = M'' + \out + \sum\limits_{x\in \X(t)} \fmset{m_{e(x)}'} \)  where \(m_{e(x_\ell)}' = \sum\limits_{x_k \in \X(t)} \left( m_{e(x_k)}'' \ast G_t(x_k,x_\ell) \right) + F_{x_\ell} (t,P) \), \(m_{e(x_k)}'' = m_{e(x_k)} - F_{x_k}(P,t)\).

Now, take the \rnPN \(\bar{N} = (\bar{P}, \bar{T}, \bar{F}, \bar{T'}, \bar{F'}, \bar{R'})\) obtained using the above construction. Consider the configurations \(\transform(M_i)\), \(\transform(M_j)\) by the Def.~\ref{def:transform}. We show that \(\transform(M_i) \rightarrow^{\ast} \transform(M_j)\).
Take \(\transform(M_i) = \bar{M}'' + \sum\limits_{x \in \X(t)} \fmset{\bar{m}_{e(x)}} + \bar{m}_{cs}\).
Take the transition \(\tau^{pre}_t\), and the mode \(\bar{e}\).
Notice that the transition \(\tau^{pre}_t\) is enabled in mode \(\bar{e}\) for \(\transform({M}_i)\), since all the pre-conditions are satisfied, i.e., \(F_{x_{cs}}({q_{start}, \tau^{pre}_t})=1\) and for all \( t \in T, p \in P, x \in \X(t)\ \bar{F}_{x}(p,\tau^{pre}_t)=F_x(p,t)\) and \(t\) is enabled in mode \(e\) for the configuration \(M_i\). 
Now, after firing the transition \(\tau^{pre}_t\), we will get a resulting configuration \(\bar{M}_{pre} = \bar{M}'' + \sum\limits_{x \in \X(t)} \fmset{\bar{m}_{\bar{e}(x)}''} + \bar{m}_{cs}^{pre}\) where \(\bar{m}_i \in \bar{M}''\) iff \(m_i \in M''\), $\bar{m}_{cs}^{pre} = \delta_{q^{move}_t}$, \(\bar{m}_{\bar{e}(x)}''(P) = m''_{e(x)}\) and \((\bar{m}_{\bar{e}(x)}''(\bar{P}\setminus P) = \textbf{0}\). Now, after firing all the move<i> blocks, we get the configuration \(\bar{M}_{pre} = \bar{M}'' + \sum\limits_{x \in \X(t)} \fmset{\bar{m}_{\bar{e}(x)}^{move}} + \bar{m}_{cs}^{move}\), where for all \(x \in \X(t)\), \(p \in P\),\\
if  \(G_t(x,x)[p,p]=1\) then
\(\bar{m}_{\bar{e}(x)}^{move}(p) = m''_{e(x)}(p)\), \(\bar{m}_{\bar{e}(x)}^{move}(\bar{p})=0\)\\
else if for \(y \in \X(t)\), \(q \in P\), \(G_t(x,y)[p,q]=1\) and \(x\neq y\) or \(p \neq q\) then \(\bar{m}_{\bar{e}(x)}^{move}(\bar{p}) = m''_{e(x)}(p)\), \(\bar{m}_{\bar{e}(x)}^{move}(p) =0\)
and \(m''_{cs} = \delta_{q^1_t}\).
Now, the transitions <block-name><i>, for all \(i \in \channelId{t}(x,y,p,q)\) where block-name = rename if \(x \neq y\) else block-name = no-rename, are fired.
We get the configuration \(\bar{M}_{trans} = \bar{M}'' + \sum\limits_{x \in \X(t)} \fmset{\bar{m}_{\bar{e}(x)}^{trans}} + \bar{m}_{cs}^{trans}\) 
where \(\bar{m}_{\bar{e}(x)}^{trans}(P) = m'_{e(x)} + F_{e(x)}(t,P) = \sum\limits_{y \in \X(t)}(m''_{e(y)} \ast G_t(y,x))\), 
\(\bar{m}_{\bar{e}(x)}^{trans}(\bar{P}\setminus P) = 0\) and \(\bar{m}_{\bar{e}(x_{cs})}^{trans} = \delta_{q_{post}}\). 
Now, after firing the transition \(\tau^{post}_t\), we will get a resulting configuration \(\bar{M}_{post} = \bar{M}'' + \sum\limits_{x \in \X(t)} \fmset{\bar{m}_{\bar{e}(x)}'} + \bar{m}_{cs}^{post} = \T(M_j)\) where $\bar{m}_{cs}^{post} = \delta_{q_{select}}$, \(\bar{m}_{\bar{e}(x)}'(P) = m'_{e(x)}\) and \((\bar{m}_{\bar{e}(x)}'(\bar{P}\setminus P) = \textbf{0}\).

\end{document}